\documentclass[a4paper]{article}
\RequirePackage{pdf14}
\RequirePackage{fix-cm}
\usepackage{./aux/jheppub}

\usepackage{graphicx}
\usepackage{amsfonts,amsmath,amssymb,amsthm}
\usepackage{tikz,tikz-cd,./aux/tikzit}
\usetikzlibrary{positioning} 
\usepackage{verbatim}
\usepackage[parfill]{parskip}
\usepackage{array,setspace,mathrsfs,yfonts,dsfont,bbm,colonequals,amscd,enumerate}
\usepackage{relsize,suffix,mathtools,cancel,bbm}

\newtheorem{thm}{Theorem}[section]
\newtheorem{lem}[thm]{Lemma}
\newtheorem{conj}[thm]{Conjecture}
\newtheorem{prop}[thm]{Proposition}

\tikzstyle{untwisted}=[fill=white, draw=black, shape=circle, tikzit draw=black, tikzit shape=circle, tikzit fill=white]
\tikzstyle{twisted}=[fill={rgb,255: red,64; green,64; blue,64}, draw=black, shape=circle, tikzit shape=circle, tikzit fill={rgb,255: red,64; green,64; blue,64}, tikzit draw=black]

\tikzstyle{equiv}=[<->, thick]
\tikzstyle{twist}=[-, dashed, draw={rgb,255: red,64; green,64; blue,64}, thick]
\tikzstyle{twisted bg}=[-, draw={rgb,255: red,128; green,128; blue,128}, dashed, thick]

\def\IC{\mathbb{C}}

\def\IN{\mathbb{N}}

\def\IV{\mathbb{V}}
\def\IZ{\mathbb{Z}}
\def\IL{\mathbb{L}}

\def\AA{{\cal A}} 
\def\BB{{\cal B}}
\def\CC{{\cal C}}

\def\FF{{\cal F}}
\def\GG{{\cal G}}

\def\II{{\cal I}}
\def\JJ{{\cal J}}
\def\KK{{\cal K}}
\def\LL{{\cal L}}

\def\NN{{\cal N}}
\def\OO{{\cal O}}

\def\QQ{{\cal Q}}
\def\RR{{\cal R}}
\def\SS{{\cal S}}
\def\TT{{\cal T}}

\def\VV{{\cal V}}
\def\WW{{\cal W}}

\def\ZZ{{\cal Z}}

\newcommand{\ket}[1]{|{#1}\rangle}

\def\Tr{\mathop{\mathrm{Tr}}\nolimits}

\def\dim{{\rm dim}}

\def\ch{{\rm ch}}

\def\slf{\mathfrak{sl}}
\def\suf{\mathfrak{su}}

\def\af{\mathfrak{a}}
\def\bbf{\mathfrak{b}}
\def\ccf{\mathfrak{c}}
\def\df{\mathfrak{d}}
\def\ef{\mathfrak{e}}
\def\gf{\mathfrak{g}}

\def \ff{\mathfrak{f}}
\def \zf{\mathfrak{z}}

\def \eg{{\it e.g.}}
\def \ie{{\it i.e.}}
\def \cf{{\it cf.}}

\def \KL{\mathrm{KL}}
\def \id{\mathbbm{1}}
\def \rk{\mathrm{rk}}

\newcommand{\Hi}[1]{\mathrm{H}^{\frac{\infty}{2}+#1}}
\newcommand{\Li}[1]{\mathrm{\bigwedge}\,^{\frac{\infty}{2}+#1}}

\def\D{\mathcal{D}^{ch}}
\def\W{\mathbf{W}}
\def\V{\mathbf{V}}

\def \Op{\mathrm{Op}}
\def \MOp{\mathrm{MOp}}

\title{Twisted chiral algebras of class \texorpdfstring{$\mathcal{S}$}{S} and mixed Feigin--Frenkel gluing}

\author{Christopher Beem and Sujay Nair}

\affiliation{Mathematical Institute, University of Oxford, Woodstock Road, Oxford, OX2 6GG, UK}

\abstract{The correspondence between four-dimensional $\mathcal{N}=2$ superconformal field theories and vertex operator algebras, when applied to theories of class $\mathcal{S}$, leads to a rich family of VOAs that have been given the monicker \emph{chiral algebras of class $\mathcal{S}$}. A remarkably uniform construction of these vertex operator algebras has been put forward by Tomoyuki Arakawa in \cite{Arakawa:2018egx}. The construction of \cite{Arakawa:2018egx} takes as input a choice of simple Lie algebra $\mathfrak{g}$, and applies equally well regardless of whether $\mathfrak{g}$ is simply laced or not. In the non-simply laced case, however, the resulting VOAs do not correspond in any clear way to known four-dimensional theories. On the other hand, the standard realisation of class $\SS$ theories involving non-simply laced symmetry algebras requires the inclusion of outer automorphism twist lines, and this requires a further development of the approach of \cite{Arakawa:2018egx}. In this paper, we give an account of those further developments and propose definitions of most chiral algebras of class $\SS$ with outer automorphism twist lines. We show that our definition passes some consistency checks and point out some important open problems.}

\begin{document}

\maketitle
\flushbottom


\eject

\section{Introduction}
\label{sec:intro}

Since their rise to prominence in the late aughts, theories of class $\SS$~\cite{Gaiotto:2009we, Gaiotto:2009hg} have proven to constitute an important and highly structured family of four-dimensional, $\NN=2$ superconformal field theories (SCFTs).\footnote{Throughout this paper we restrict without further comment to the $\NN=2$ superconformal case when referencing SCFTs or ``theories'' more generally.} These theories are best understood as those which arise in the low energy limit when six-dimensional $(2,0)$ SCFTs are compactified (in an appropriately supersymmetry-preserving fashion) on a compact two-manifold (the \emph{UV curve}), possibly in the presence of codimension-two supersymmetric defects. Many properties of class $\SS$ theories are well-characterised in terms of two-dimensional structures, \eg, the marginal (gauge) couplings are identified with complex structure moduli of the curve. Indeed, a cottage industry of identifying four-dimensional observables with dual computations taking place in two dimensions on the UV curve has emerged and achieved a good degree of maturity.

One such observable (to use the term expansively) that is tractable yet structurally intricate is the associated vertex operator algebra (VOA) or chiral algebra \cite{Beem:2013sza}. The VOAs associated to theories of class $\SS$ were first systematically studied in \cite{Beem:2014rza} where they went under the name \emph{chiral algebras of class $\SS$}. In that work, a number of key properties of these VOAs were identified and some explicit computations performed in simple cases (with underlying $(2,0)$ theory of type $\gf=\af_1$ or $\af_2$). A speculative vision was also put forward wherein general chiral algebras of class $\SS$ might be uniquely determined by their various duality properties---a \emph{theory space bootstrap}---which was an idea that had seen prior success in the case of the superconformal index \cite{Gaiotto:2012xa}.

This speculation was largely confirmed several years later in a remarkable paper by Tomoyuki Arakawa \cite{Arakawa:2018egx}. As we will review later, Arakawa's construction involves a gluing operation, which we refer to as \emph{Feigin--Frenkel gluing}, that has no obvious four-dimensional counterpart but which can heuristically be thought of as gluing UV curves together along interior points as opposed to marked points. The crucial insight pointing towards the construction in question is that maximal punctures in class $\SS$ give rise to critical-level affine Kac-Moody subalgebras of the associated VOA, and the centres of these subalgebras (their \emph{Feigin-Frenkel centres}) should be identified amongst the different punctures. This is an affine version of a set of well-known Higgs branch chiral ring relations for class $\SS$ theories, which themselves play the key role in the construction of the Moore--Tachikawa holomorphic symplectic varieties \cite{Moore:2011ee} by Ginzburg--Kazhdan \cite{Ginzburg:2018} and Braverman--Finkelberg--Nakajima~\cite{Nakajima:2015txa, Braverman:2016wma, Braverman:2017ofm, Braverman:2016pwk}.

A curious feature of the construction in \cite{Arakawa:2018egx} is that it makes perfect sense for any choice of underlying simple Lie algebra $\gf$, whether simply laced or not. This state of affairs is to be contrasted with the physical reality of class $\SS$, where the basic story applies only for $\gf$ simply laced (because $\gf$ is the algebra labelling the parent six-dimensional $(2,0)$ theory), and the incorporation of non-simply-laced (gauge and global) symmetry algebras normally requires the extra structure of (discrete) outer automorphism twist lines being introduced on the UV curve~\cite{Vafa:1997mh, Tachikawa:2009, Tachikawa:2011}. This raises the question of how/whether to interpret in four-dimensional terms the VOAs constructed in \cite{Arakawa:2018egx} for the non-simply laced case. Perhaps more importantly, it leaves open the problem of defining in comparable generality the \emph{twisted chiral algebras of class $\SS$} corresponding to theories that are realised using outer automorphism twist lines.

The purpose of the present work is to address the latter problem. We propose a definition \emph{\`a la Arakawa} for a large family of twisted chiral algebras of class $\SS$. Our definition allows for the realisation of all trinion chiral algebras with one untwisted and two twisted punctures. This is still restricted, however, in that it does not allow for the case of the $\df_4$-type trinion theory with three twisted punctures, or more general theories whose realisation require such constituent fixtures. While we can prove a number of duality properties of these algebras, we encounter an obstruction in the proof of the most general expected dualities. This obstruction is analogous to one that appears already in \cite{Arakawa:2018egx} in the case of higher-genus chiral algebras.

The organisation of the rest of this paper is as follows. In Section \ref{sec:twisted_background} we provide an idiosyncratically selective review of class $\SS$ SCFTs and their generalisation to incorporate outer automorphism twist lines. In Section \ref{sec:chiral_background} we recall the associated vertex operator algebras of class $\SS$ theories and provide a fairly pedagogical overview of their construction in the simply laced setting by Arakawa. In Section \ref{sec:fixtures} we propose a definition of the associated VOA for the simplest twisted trinions and establish generalised $S$-duality for the twisted chiral algebras of class $\SS$ with respect to a subset of duality transformations. In Section \ref{sec:conc_remarks} we make some comments on possible extensions and future directions. We include a short review of spectral sequences, which play an important role in our proofs and those of \cite{Arakawa:2018egx} in Appendix \ref{app:spectral_sequences}. In Appendices \ref{sec:proof_closed_immersion} and \ref{app:proof_cyl}, we prove two important technical theorems used in our analysis of mixed Feigin--Frenkel gluing. In Appendix \ref{app:rewriting_characters} we give the details of a useful rewriting of class $\SS$ Schur indices, and in Appendix \ref{app:ds_free} we include some calculations regarding the Drinfel'd--Sokolov reduction from our twisted trinion chiral algebra to a free symplectic boson VOA.

\emph{Author's note: while this paper was in the late stages of preparation, the paper~\cite{Frenkel:2021bmx} appeared in which their Lemma 9.4 is a finite-dimensional cousin of our Theorem 4.1.}


\section{Review of (twisted) class \texorpdfstring{$\SS$}{S}}
\label{sec:twisted_background}

To begin, let us review some important features of four-dimensional SCFTs of class $\SS$, paying special attention to the generalisation to included twist lines. This is not meant to be a comprehensive overview; much more background on these theories can be found in,~\eg,~\cite{Gaiotto:2009we, Gaiotto:2009hg, Tachikawa:2011}.

Untwisted theories of class $\SS$ are specified by a triple $(\gf,\CC_{g,s},\{\Lambda_i\})$, where $\gf$ is the simply-laced Lie algebra labelling the underlying six-dimensional theory, the UV curve $\CC_{g,s}$ is a genus $g$ Riemann surface with $s$ marked points, upon which the $(2,0)$ is compactified, and the $\{\Lambda_i\}$ for $i=1,\ldots,s$ are a set of labels for the marked points, which amount to specifications of maximally supersymmetric, codimension-two defect operators inserted at those points in six dimensions. The only geometric data of $\CC_{g,s}$ that are relevant are its complex structure moduli, which get identified with marginal (gauge) couplings of the associated four-dimensional SCFT.

The six-dimensional $(2,0)$ theories enjoy discrete global symmetries corresponding to outer automorphisms of the underlying Lie algebra $\gf$---these can be identified with graph automorphisms of the corresponding Dynkin diagrams. One may then include twists/monodromies by these outer automorphism symmetries when compactifying; said differently, one generalises the compactification procedure to incorporate a local system of Dynkin diagrams on the UV curve. The data of this local system can be represented pictorially by including consistent networks of (potentially oriented) twist lines on the UV curve labelled by elements of the outer automorphism group. This generalisation of the class $\SS$ construction also introduces new types of punctures, called \emph{twisted punctures}, which correspond to codimension two defects that carry an outer automorphism monodromy. We will be specialising to the genus zero case in our construction and so for convenience we introduce the notation, $\CC_{0,m,n}$ to denote the Riemann sphere with $m$ untiwsted punctures and $2n$ twisted punctures. One cannot specify, unambiguously, a higher genus surface with twisted punctures by the genus and number of each type of puncture.

\subsection{\label{ssec:twisted_background/punctures}Untwisted and twisted punctures}

Marked points in class $\SS$ can be either \emph{regular} or \emph{irregular}, and in this work we will restrict our attention entirely to the regular case. Regular, untwisted punctures in a theory of type $\gf\in\{\af_n,\df_n,\ef_n\}$ are labelled by conjugacy classes of homomorphisms $\Lambda:\slf(2)\rightarrow \gf$. The defect operator corresponding to such a puncture enjoys some global symmetry that it then contributes to the total global flavour symmetry of four dimensional theory. The factor contributed takes the form $\ff_\Lambda=C_{\gf}\left(\Lambda(\suf(2))\right)$, \ie, the commutant of the image of $\suf (2)$ in $\gf$. As a consequence of the Jacobson--Morozov theorem~\cite{Collingwood:1993}, the marked points can therefore be equivalently labelled by nilpotent orbits in $\gf$.

Any simple Lie algebra has at least two especially important nilpotent orbits: the principal orbit (which is the unique, largest nilpotent orbit) and the trivial orbit. The trivial orbit corresponds to the trivial homomorphism $x\mapsto0, \forall x\in\suf(2)$, while an explicit expression for the principal orbit can be found in~\cite{Collingwood:1993}. Their respective commutants are $\mathfrak{f}=0$ for the principal orbit and $\mathfrak{f}=\gf$ for the trivial orbit. A puncture labelled by the trivial embedding is called a \textit{maximal} puncture, while a puncture labelled by the principal embedding contributes no flavour symmetry and is equivalent to having no puncture at all. With an eye towards the twisted case, we will nevertheless adopt the convention of referring to a hypothetical puncture labelled by the principal embedding as an \textit{empty} puncture.

Given a theory where a particular puncture is maximal, the corresponding theory with that puncture replaced by a sub-maximal puncture can be realised by partially Higgsing the $\gf$ flavour symmetry associated to that puncture in the four dimensional theory. A choice of nontrivial $\Lambda$ is realised by assigning an expectation value to the ``moment map'' Higgs branch operator in the conserved current multiplet that lies in the corresponding nilpotent orbit. Consequently, for many purposes it is sufficient to be able to construct theories associated to surfaces with maximal punctures, with other punctures structures being subsequently reached via partial Higgsing. This will be our approach in this paper.

\begin{table}[!t]
\centering
\begin{tabular}{c|c|c}
    Lie algebra ($\gf$) & ~~Outer automorphism twist~~ & Twisted algebra ($\gf_t$)\\
    \hline
    $\af_{2n}$ & $\IZ_2$ & $\ccf_n$ \\
    $\af_{2n-1}$ & $\IZ_2$ & $\bbf_n$ \\
    $\df_{n}$ & $\IZ_2$ &  $\ccf_{n-1} $ \\
    $\df_4$ & $\IZ_3$ & $\gf_2$ \\
    $\ef_6$ & $\IZ_2$ & $\ff_4$ \\
\end{tabular}
\caption{Simply laced Lie algebras $\gf$ and the corresponding twisted algebras $\gf_t$ for different choices of outer automorphism twist. The Lie algebras $\af_{2n}$ and $\df_{n}$ both give rise to Lie algebras of type $\ccf$ after a twist. The corresponding theories nevertheless have subtle differences---see,~\eg,~\cite{Chacaltana:2014nya,Beem:2020pry}. The algebra $\mathfrak{d}_4$ has an outer automorphism group of $S_3$, the symmetric group on three elements. Here we will only consider its abelian subgroups $\IZ_2$ and $\IZ_3$.}
\label{tab:outer_aut}%
\end{table}

Twisted punctures are also labelled by $\slf(2)$ homomorphisms/nilpotent orbits, but now of a different Lie algebra $\gf_t$. Let $\gf_t^\vee$ be the subalgebra of $\gf$ invariant under the outer-automorphism monodromy of the puncture. Then twisted punctures are labelled by nilpotent orbits in the simple Lie algebra $\gf_t$, the Langlands dual to $\gf_t^\vee$. The pairs of untwisted algebras and their twisted counterparts can be found in Table~\ref{tab:outer_aut}. The maximal twisted punctures carries $\gf_t$ flavour symmetry, while submaximal twisted punctures proceed analogously to the twisted case. Importantly, unlike in the untwisted case, the empty twisted puncture remains a nontrivial puncture (as it still carries monodromy on the UV curve; in terms of twist lines there is still a point where the relevant twist line ends, which distinguishes the point from a generic point on the UV curve).

Throughout this paper, $\gf$ and $\gf_u$ will generally refer to the type of the parent six-dimensional theory or, equivalently, to the flavour symmetry algebra contributed by a maximal untwisted puncture, while we denote by $\gf_t$ the flavour algebra of a maximal twisted puncture as in Table~\ref{tab:outer_aut}. More generally, subscripts $u$ and $t$ will differentiate between various objects associated to $\gf_u$ and $\gf_t$.

\subsection{\label{subsec:trinions_and_duality}Trinions, gluing, and duality}

Weak coupling limits of class $\SS$ theories correspond to degeneration limits of the UV curve, with maximally weakly coupled limits corresponding to pair-of-pants decompositions. The basic entities are the theories corresponding to pairs of pants/spheres with three maximal punctures; these are often referred to as \emph{trinion theories} or \emph{fixtures}. Even in the untwisted setting, there are myriad possible (combinations of) marked point labels that can appear on a single trinion; these have been extensively detailed by (various subsets of) Chacaltana, Distler, Tachikawa, and Trimm in \cite{Chacaltana:2012zy,Chacaltana:2010ks,Chacaltana:2011ze}. These (generically strongly coupled) theories can be combined by exactly marginal gauging to realise general class $\SS$ theories, but the equivalence of the SCFTs realised in this manner by inequivalent pair-of-pants decompositions constitutes a highly nontrivial set of quantum dualities often referred to as generalised $S$-dualities. The class $\SS$ dictionary identifies the (compactified) moduli space of the UV curve with the space of complexified gauge couplings of the associated theory. The various pants decompositions are, therefore, the various weak coupling limits of the associated theory---all of them related by generalised $S$-duality.

In the untwisted case, all pants decompositions can be reached by iterating two types of elementary move.

For the first type of move, consider the four puncture sphere $\CC_{0,4}$ with some labelling  $1,\dots,4$ of its punctures. We move to the singular point on the moduli space where the sphere decomposes into two trinions connected by a long tube, one of which has punctures labelled by $1,2$ and the other by $3,4$. In Figure~\ref{fig:4-move}, this is represented by the duality frame on the left. By moving to a different singular point, one can decompose $\CC_{0,4}$ into two connected trinion labelled by $1,3$ and $2,4$---shown on the right hand side of Figure \ref{fig:4-move}. We call such a move, moving between the various decompositions of $\CC_{0,4}$, a $4$-move. This can be lifted to a $4$-move acting on any collection of four punctures on a general $\CC_{g,s}$. First, one moves to a singular locus on the moduli space where $\CC_{g,s}$ decomposes into a $\CC_{0,4}$ that contains the punctures of interest and is connected by a long tube to $\CC_{g,s-4}$. Then one applies a $4$-move to swap between decompositions of $\CC_{0,4}$, before moving back out of the singular loci.

For surfaces $\CC_{g,s}$ with $g>0$, one must consider another type of move. For example, consider the one punctured torus $\CC_{1,1}$. The $S$ generator of the modular group acts by swapping the $a$ and $b$ cycles of the torus. This is a homeomorphism of the torus to itself that is not homotopic to any iterated $4$-move and so must be a generator, which we call the $ab$-move. For a surface, $\CC_{g,s}$, there is a natural generalisation of this move. In the untwisted case, the $ab$-move acts as the $S$-duality $\tau\mapsto -1/\tau$ on the complex gauge coupling associated to the handle.
\begin{figure}[h]
\ctikzfig{figs/4-move}
\caption{The $4$-move acting on four maximal untwisted punctures.}
\label{fig:4-move}
\end{figure}

In the presence of twist lines, trinions can be glued together along either twisted or untwisted maximal punctures. The possible trinions with twisted punctures\footnote{The case of $\df_4$ trinions with non-abelian twists has recently been explored in~\cite{Distler:2021cwz}, but we restrict our attention to the abelian case} have also been classified by Chacaltana, Distler, Tachikawa and Trimm~\cite{Chacaltana:2012ch,Chacaltana:2013oka}. In the terminology of \cite{Chacaltana:2012ch} we will be restricting our attention to the case of regular, typical, twisted punctures.

The presence of twist lines introduces new types of duality transformations, these can once again be split into two classes. On a surface $\CC_{0,m,n}$, there are three types of $4$-move that one can consider. The first acts on a collection of four untwisted punctures and swaps between the various pants decompositions of $\CC_{0,4}$ into untwisted trinions. Similarly, the second acts on two pairs of twisted punctures and swaps between the pants decomposistions of $\CC_{0,0,2}$. The third move is different, it acts on a pair of untwisted punctures and a pair of twisted punctures. the surface $\CC_{0,1,1}$ has two pants decompositions, described in Figure~\ref{fig:ut-move}. The third type of $4$-move swaps between these two frames. Unlike the previous moves, the type of gluing in each frame is different. On the left hand side, one gauges the diagonal $\gf_u$ action whereas one gauges a diagonal $\gf_t$ action on the right. Iterating these $4$-moves allows on to traverse a large strand of the duality web of $\CC_{0,m,n}$. For the non-abelian case of $\df_4$ there are more intricate dualities, explored in~\cite{Distler:2021cwz}, which we will not discuss here.
\begin{figure}[h]
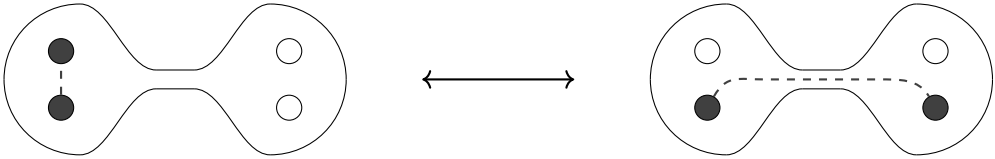

\ctikzfig{figs/ut-move}
\caption{The two degeneration limits of $\CC_{0,2,1}$. Note that the gluing on the left is untwisted but twisted on the right. We mark untwisted punctures by unfilled circles and twisted punctures by filled circles. We connect the twisted punctures by dashed twist-lines for clarity.}
\label{fig:ut-move}
\end{figure}

For twisted surfaces at higher genus, the $ab$-move has another variation. Take, for instance, the $\df_n$ theory on a maximal three punctured sphere. We can construct a gauge theory by self-gluing two punctures on the sphere together, which corresponds to gauging the diagonal action of $Spin(2n)$ on the two theories. We may also gauge with respect to a diagonal action that has been twisted by an outer automorphism,~\ie,~$Spin(2n)$ acts as $g\otimes\sigma(g)$ for (a lift of) an outer automorphism $\sigma$. Pictorially, we represent this by a cylinder, with a twist line around it, connecting the punctures. This results in a genus one surface with a twist line running around the $a$-cycle of the torus.

Now, consider, a $\df_n$ theory with one maximal puncture and two maximal, $\IZ_2$, twisted punctures. Again, we can construct a gauge theory, by self-gluing the two twisted punctures together. This time, we gauge with respect to the diagonal action. This results in a genus one surface with a twist line running along the $b$-cycle. These two theories are known to be $S$-dual, and the associated UV curves are related by the action of the $ab$-move. The two degeneration limits are shown in Figure~\ref{fig:ab-move}. Unlike the untwisted case, the $ab$-move changes the gauge group as well as moving to a weakly coupled point.
\begin{figure}[h]
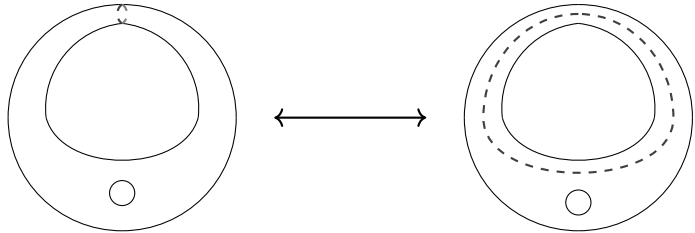

\ctikzfig{figs/ab-move}
\caption{The $ab$-move swapping between two decompositions of the once punctured torus with a twist line.}
\label{fig:ab-move}
\end{figure}
We are primarily interested in observables that do not depend on exactly marginal deformations, such as gauge couplings. Such observables are blind to continuous paths in the moduli space and only see the various degeneration points. In the untwisted case, the $ab$-move is trivial for such observables and the action of the $4$-moves can be phrased in terms of the action of permutations on punctures. The $4$-moves, in the untwisted case, act by permuting the flavour symmetries associated to the punctures, which shows up as permutations on the flavour fugacities (see the next section) of the index or permuting the moment-maps of the Higgs branch and the associated vertex algebra. It is, therefore, useful to phrase the action of $S$-duality on these observables in terms of the action of a permutation group. For a surface $\CC_{g,n}$ the $4$-moves generate an action of $S_n$.

The twisted case, however, is not so simple. The first and second types of $4$-move, acting on four punctures of the same type, can be rephrased as the action of $S_4$ on each collection of untwisted or twisted punctures. The third type of four move,~\ie\ the move of Figure~\ref{fig:ut-move}, cannot be rephrased in this way. The two duality frames are not simply related by the action of a permutation on their moment maps or flavour fugacities. Instead one must check the non-trivial agreement between both types of BRST gauging on each side. Matters are further complicated by the $ab$-move of Figure~\ref{fig:ab-move}. Like the $4$-move of Figure~\ref{fig:ut-move}, both sides of the duality feature a different type of BRST gauging, and again one must check the agreement between both gluings.

In Section~\ref{ssec:duality}, after we have constructed the twisted vertex algebras, we will construct the vertex algebraic version of these moves. We are unable to establish invariance across the full duality web, but we do prove invariance under the various types of $4$-moves we have defined. In the genus zero case, to which we are mostly specialising, this does recover the full duality web (under some reasonable assumptions about the properties of gauging).

\subsection{\label{ssec:twisted_background/index}The superconformal index}

The superconformal index of a four dimensional $\NN=2$ SCFT is defined as a graded trace over the radially quantised Hilbert space of the theory~\cite{Kinney:2005ej,Rastelli:2014jja},
\begin{equation}\label{eq:ind2}
    \mathcal{I}(p,q,t,\mathbf{x}) = \Tr\,(-1)^Fe^{-\beta(E-2j_2-2R+r)}p^{\frac{E+2j_1-2R-r}{2}}q^{\frac{E-2j_1-2R-r}{2}}t^{R+r}\prod_{i=1}^{\mathrm{rk} \mathfrak{g}_F}x_i^{T_i}~,
\end{equation}
where $p,q,t$ are superconformal fugacities and $\mathbf{x}=(x_1,x_2,\dots,x_{\rk\,\gf_F})$ are further fugacities conjugate to generators, $T_i$, of a Cartan subalgebra of the flavour symmetry $\gf_F$. A standard argument implies that the index can only receive a nonvanishing contribution from states obeying
\begin{equation}\label{eq:index_contribute}
    E-2j_2-2R+r=0~,
\end{equation}
and so is actually independent of $\beta$. As written, the index counts minimally supersymmetric states. There are a number of limits with enhanced supersymmetry, and we will be specifically interested in the Schur limit~\cite{Gadde:2011uv},
\begin{equation}\label{eq:limschur}
    q\rightarrow t~,\quad p \text{ arbitrary}~.
\end{equation}
In this limit, superconformal representation theory ensures that the index is in fact independent of $p$ and takes the schematic form
\begin{equation}\label{eq:indschur}
    \mathcal{I}(q) = \Tr(-1)^Fq^{E-R}~,
\end{equation}
with the regularising $\beta$-dependent term suppressed. The Schur index only counts states whose quantum numbers satisfy the further linear relations
\begin{equation}\label{eq:schurindop}
    E-j_1-j_2=2R~,\qquad r = j_1-j_2~.
\end{equation}
This index is a key ingredient in the SCFT/VOA correspondence \cite{Beem:2013sza}. As a graded vector space, the associated VOA can be identified with the space of Schur operators, \emph{i.e.}, the local operators satisfying the above constraints, with the holomorphic dimension of a given state being given by
\begin{equation}\label{eq:holo_dim_schur}
    h = \frac{E + (j_1+j_2)}{2} = E - R~,
\end{equation}
so the Schur index is precisely the vacuum (super)character of the associated vertex algebra~\cite{Beem:2013sza}.

For theories of class $\mathcal{S}$, the form of the full superconformal index has been shown to follow from duality properties \cite{Gaiotto:2012xa}, though for our purposes here we will restrict attention to the Schur limit. The index takes the form of a sum over highest weights in the set of integrable dominant weight representations $P^+$ (\ie, a sum over finite-dimensional $\gf$ representations), weighted by some ``structure constants'' $C_{\lambda\lambda\lambda}$, which are functions of $q$. For a UV curve $\CC_{g,s}$, which can be realised by gluing $2g-2+s$ trinions, the index takes the form~\cite{Aganagic:2005, Gadde:20102d, Gadde:2011uv, Alday:2013kda}
\begin{equation}
\label{eq:indfull}
	\II(q;\mathbf{x}_1,\mathbf{x}_2,\dots,\mathbf{x}_s) = \sum_{\lambda\in P^+}(C_{\lambda\lambda\lambda}(q))^{2g-2+s}\prod_{i=1}^s \KK(q;\mathbf{x}_i)\chi^\lambda(\mathbf{x}_i)~,
\end{equation}
where $\mathbf{x}_i$ are the flavour fugacities arising from the flavour symmetries of the $i\,$th puncture. The $\chi^\lambda$ are Schur polynomials, \ie, characters of the finite-dimensional $\gf$ representation with highest weight $\lambda$. The $\KK$ factors are defined by
\begin{equation}
	\KK(\mathbf{x}) = \sqrt{\II^V(\mathbf{q;x})}~,
\end{equation}
where $\II^V$ is the Schur index for a free vector multiplet in the adjoint representation of $\gf$.

To introduce non-maximal punctures at the level of the index amounts to a fugacity replacement (with a subtraction of certain divergences) that is detailed in~\cite{Gadde:2011uv,Rastelli:2014jja,Beem:2014rza}. Since reducing to an empty puncture amounts to removing the puncture all together, one has the relation
\begin{equation}
	C_{\lambda\lambda\lambda}(q) = \frac{1}{\KK(\times) \chi^\lambda(\times)}~,
\end{equation}
where $\times$ represents the regularised fugacity replacement for the empty puncture, which is purely a function of $q$.

We can recast the above in more representation theoretic terms by using Appendix A of~\cite{Lemos:2014lua}, which contains a useful dictionary. Let $\RR$ denote the set of roots of $\gf$, then
\begin{equation}
	\KK(\mathbf{x}) = \frac{1}{(q;q)_\infty^{\rk~\gf}}\prod_{\alpha\in\RR} \frac{1}{(qe^{ \alpha}(\mathbf{x});q)_\infty}~,
\end{equation}
where $e^{\alpha}:T\rightarrow \IC^\times$ are the roots of the simply connected Lie group, $G$, with $Lie(G)=\gf$ and maximal torus $T$. For the empty puncture, we have
\begin{equation}
	\KK(\times) = \prod_{i=1}^{\rk\,\gf} \frac{1}{(q^{d_i},q)_{\infty}}~,
\end{equation}
which one can identify with the character of the Feigin--Frenkel centre, $\zf(\gf)$. Similarly, the summand, $\KK(q;\mathbf{x})\chi^\lambda(\mathbf{x})$ can be identified with the character of the critical level Weyl module of heighest weight $\lambda$, $\IV_{\lambda}$ as
\begin{equation}
	\KK(q;\mathbf{x})\chi^\lambda(\mathbf{x}) = \Tr_{\IV_\lambda}(q^{D} \mathbf{x}^T)~,
\end{equation}
where $D$ is the quasiconformal weight. For clarity's sake we shall often suppress the flavour fugacities and use the notation $\ch~ V = \Tr_V(q^D)$. We will provide more details on the Feigin--Frenkel centre and Weyl modules in Section~\ref{sec:chiral_background}, but for now we will continue with our rewriting. At the critical level, the Weyl modules are not irreducible---instead they have a unique simple quotient $\IL_\lambda$. The characters of the Weyl module and their simple quotients are related by
\begin{equation}
    \ch~ \IL_\lambda = \frac{ \ch~\IV_{\lambda}}{\ch~\zf_\lambda} = \frac{\KK(q;\mathbf{x}) \chi^\lambda(\mathbf{x})}{\KK(\times)\chi^\lambda(\times)}~,
\end{equation}
where $\zf_\lambda$ is the Drinfel'd Sokolov reduction of $\IV_\lambda$ and will also be introduced in Section~\ref{sec:chiral_background}. The index of $\CC_{g,s}$ can then be completely rewritten in terms of characters of Weyl modules as
\begin{equation}
	\II(q;\mathbf{x_1}\mathbf{x_2},\dots,\mathbf{x}_s) = \sum_{\lambda\in P^+} (\ch~\zf_\lambda)^{2-2g} \prod_{i=1}^{s} \ch~\IL_\lambda = \sum_{\lambda\in P^+}(\ch~\zf_\lambda)^2 \prod_{i=1}^s \frac{\ch~\IV_\lambda}{\ch~\zf_\lambda}~.
\end{equation}

The Macdonald limit (and so Schur limit by further specialisation) of the superconformal index in the twisted setting was studied for type $\df_n$ theories in~\cite{Lemos:2012ph}. According to the analysis there, the presence of an single twisted puncture restricts the sum over $P^+$ to representations that are invariant under the action of the outer automorphism twist, which is equivalent to summing over the set of highest weight representations of the twisted algebra $\gf_t$ (which we denote by $P^+_t$). The overall structure constants are also modified, though they are expressed in terms of the same building blocks. For a surface of genus $g$ with $s$ untwisted punctures and $2s_t$ $\mathbb{Z}_2$-twisted punctures such that no twist lines wrap any cycles, the index then takes the form
\begin{equation}\label{eq:twistind}
    \mathcal{I}(q;\mathbf{x}_1,\dots,\mathbf{x}_s,\mathbf{y}_1,\dots,\mathbf{y}_{2s_t}) =
    \sum_{\lambda^\prime\in P^+_t} \frac{\prod_{i=1}^s\mathcal{K}_{\gf_u}(q;\mathbf{x}_i)\chi^{\lambda=\lambda^\prime}_{\gf_u}(\mathbf{x}_i)\prod_{j=1}^{2s_t}\mathcal{K}_{\gf_t}(q;\mathbf{y}_j)\chi^{ \lambda^\prime}_{\gf_t}(\mathbf{y}_j)}{{(\mathcal{K}_{\gf_u}(\times)\chi^{\lambda=\lambda^\prime}_{\gf_u}(\times))}^{2g-2+s_u+2s_t}}~.
\end{equation}
Here $\mathbf{x}_i$ are fugacities for untwisted punctures and $\mathbf{y}_j$ are fugacities for the twisted ones. We have adopted the notation of~\cite{Lemos:2012ph}, where $\lambda=\lambda'$ means that the highest weight $\lambda$ of $\gf_u$ is the one that restricts to $\lambda'$ when passing to outer automorphism invariants. Later on we shall introduce our own notation for these weights and elaborate on this restriction. For the non-abelian twist of $\df_4$, a TQFT form of the index was proposed in~\cite{Distler:2021cwz}. This agrees with the heuristics we have so far observed---namely the sum is restricted to the integral dominant weights of $\df_4$ that are invariant under the action of the twists that are present.

One can also rewrite the twisted index in terms of Weyl modules as
\begin{equation}
	\mathcal{I}(q;\mathbf{x}_1,\dots,\mathbf{x}_s,\mathbf{y}_1,\dots,\mathbf{y}_{2s_t}) = \sum_{\lambda'\in P_t^+} (\ch~\zf_\lambda^u)^{2-2g} \bigg(\prod_{i=1}^s \IL_{\lambda=\lambda'}^u\bigg)\bigg(\prod_{j=1}^{2s_t} \IL_{\lambda'}^t\bigg)~,
\end{equation}
where we use the superscripts $u,t$ to distinguish between the (simple quotients of) Weyl modules over each algebra.

It may be worth remarking that all structure constants are of ``untwisted type'' in spite of the $\KK$ factors and Schur functions of twisted type. This means that when closing a twisted puncture, the specialised puncture factor in the numerator won't cancel against a corresponding factor in the denominator. This is the index-level incarnation of the fact that minimal twisted punctures are nontrivial and cannot be completely erased from the UV curve.

\subsection{\label{ssec:twisted_background/HL}Residual gauge symmetry and derived structures}

An interesting trichotomy in four-dimensional $\mathcal{N}=2$ theories is between those theories with ``pure'' Higgs branches (there is no unbroken continuous gauge symmetry at generic points), those with ``enhanced'' Higgs branches (generic Higgs vacua still support massless Abelian vector multiplets), and those with interacting Higgs branches (where generic Higgs vacua support a nontrivial interacting SCFT). To the best of our knowledge, in the presence of regular punctures only pure and enhanced Higgs branches can arise, whereas interacting Higgs branches are common in Argyres--Douglas type theories that can be engineered with irregular punctures.

In the untwisted class $\SS$ setting, there is a simple characterisation of under what circumstances a theory will have an enhanced rather than a pure Higgs branch: theories of type $\mathfrak{a_n}$ ,$\mathfrak{d_n}$, or $\mathfrak{e_n}$ with genus $g$ have generic residual gauge symmetry with rank equal $n\times g$. This admits a straightforward heuristic explanation: maximal Higgsing will involve first closing all punctures and then moving to the maximal Higgs branch of the unpunctured theory. For the unpunctured theory, one can adopt a six-dimensional perspective and first consider moving to a point on the tensor branch of the $(2,0)$ where the effective theory is that of $n$ free tensor multiplets, and subsequently compactifying.

For twisted theories, the situation is a bit more complicated. To illustrate, we consider the $\mathfrak{d_2}$ theory. Due to the accidental isomorphism $\df_2\cong\af_1\times \af_1$, we can recast theories of type $\df_2$ into $\af_1$ theories. The twisted subalgebra is just $\af_1$, and a full twisted puncture becomes a conventional (untwisted) puncture of the $\af_1$ theory. In particular, the $Spin(4)$ gauge theory with $N_f=4$ flavours can be engineered via compactification of the $\df_2$ theory on a sphere two maximal and two minimal twisted punctures. Equivalently, it can be described as the $\af_1$ theory compactified on a genus-one surface with two punctures. Thus, the $\af_1$ surface is a double cover of the $\df_2$ surface, treating the twist lines as branch cuts. The $\df_2$ theory has residual gauge symmetry at a generic point of the Higgs branch---despite being superficially of genus zero. We will see that this phenomenon is characteristic of the twisted theories.

Somewhat more generally, for $SO(2n)$ superconformal QCD---realised in type $\df_n$ using a sphere with four twisted punctures (two minimal and two maximal)---a generic point of the Higgs branch has precisely a residual $U(1)$ gauge symmetry~\cite{Argyres:1996hc}. Though for $n\geqslant 3$ there is not an immediate relation to a higher genus class $\mathcal{S}$ theory like the $\df_2$ case, we can nevertheless observe that the presence of outer automorphism twist lines means that there is a natural covering space for the UV curve (thought of as the space where the corresponding local system of Dynkin diagrams has no monodromy), and in this example the covering space is genus one. We believe that in the general case of twisted class $\mathcal{S}$ theories, it is precisely the value of the genus (zero or nonzero) of this covering space that controls whether a theory in question has a pure or an enhanced Higgs branch. More generally, we take away the lesson that in the presence of twisted punctures, genus zero theories should perhaps nevertheless be thought of as more analgous to higher-genus untwisted theories than genus-zero untwisted theories.

From an algebraic perspective, we can think of the presence of an enhanced Higgs branch as a shadow of an underlying \emph{derived structure}. In particular, the Hall--Littlewood chiral ring for higher genus (untwisted) class $\mathcal{S}$ theories is known to include additional fermionic sectors beyond the Higgs branch chiral ring, which can be understood as a consequence of the Higgs branches of these theories being more naturally understood as derived symplectic varieties \cite{CJBDBS:HL}. We propose that for twisted class $\mathcal{S}$ theories as well there will be a correspondence between the presence of an enhanced Higgs branch and higher-fermion-number (derived) elements of the Hall--Littlewood chiral ring.

For our purposes in this paper, this nontrivial derived structure is most relevant as an indicator that the BRST cohomology that encodes the associated vertex operator algebra of a gauge theory (see the review in the next section) is not concentrated in degree zero. Consequently, we expect a construction in the twisted setting to be able to mimic the corresponding untwisted genus-zero story only so long as there are not too many twisted punctures. This expectation will be borne out in the detailed analysis of Section \ref{ssec:gluing_mixed}.


\section{Chiral algebras of class \texorpdfstring{$\SS$}{S} after Arakawa}
\label{sec:chiral_background}

Four-dimensional $\NN=2$ SCFTs give rise to their associated vertex operator algebras by way of the cohomological construction of \cite{Beem:2013sza}. In the case of SCFTs of class $\SS$, the associated vertex algebras are called chiral algebras of class $\SS$ and were studied in \cite{Beem:2014rza}. A systematic and rigorous definition and analysis of these vertex algebras in the untwisted case was developed in \cite{Arakawa:2018egx}. In this section we first present several imporant constructions that are employed in definition of chiral algebras of class $\SS$ in \cite{Arakawa:2018egx}, before ultimately recalling the definition from that work.

\subsection{Critical level current algebras and the Kazhdan--Lusztig category}
\label{subsec:AKM_and_KL}

If an SCFT has a (continuous) global flavour symmetry with Lie algebra $\ff$, then the superconformal primaries in the corresponding conserved current multiplets are moment map operators in the Higgs branch chiral ring, which are Schur operators. Their images in the associated vertex algebra are affine currents that generate an $\hat\ff$ affine Kac-Moody (AKM) vertex algebra at level
\begin{equation}
    k_{2d} = -\frac{1}{2}k_{4d}~,
\end{equation}
where $k_{4d}$ is the (positive for unitary theories) four dimensional flavour central charge. Theories of class $\SS$ have flavour symmetries arising from marked points, and for a maximal puncture the four-dimensional level of the $\gf_u$ or $\gf_t$ flavour symmetry is $k_{4d}=2h^\vee_{\gf_{u,t}}$. The corresponding $\hat\gf_{u,t}$ current subalgebras in the associated VOA are therefore at the \emph{critical level} $k_{2d}=\kappa_c=-h^\vee_{\gf_{u,t}}$. This is an important value of the level where there is a large null ideal.\footnote{In fact, the null ideal is a commutative vertex operator subalgebra (centre) known as the \emph{Feigin--Frenkel centre}, to which we will return momentarily.}
 
Null states in these current subalgebras need not be null with respect to the full VOA. To describe the situation, it is useful to introduce the universal affine vertex algebra $V^{k_{2d}}(\gf)$, where \emph{none} of the null states are set to zero,
\begin{equation}\label{eq:univoa}
    V^{k}(\gf) \coloneqq \mathrm{Ind}_{\gf[\![t]\!]\oplus\IC\id}^{\hat \gf_k}\IC= U(\hat \gf_k) \otimes_{U(\gf[\![t]\!]\oplus\IC K)}\IC~.
\end{equation}
Here $\gf[\![t]\!]$ acts trivially on $\IC$ and $K$ is the central element, which acts as the constant level $k\in\IC$. This can be thought of as the specialisation of the generic-level current algebra to any particular level $k$ without taking any further quotient if this turns out to be reducible. (By contrast, taking the maximal quotient gives the simple affine vertex algebra $V_{k_{2d}}(\gf)$, whose underlying graded vector space can jump as a function of $k$.) $V^k(\gf)$ is generically conformal with stress tensor given by the Sugawara construction. At the critical level, this construction fails and the vertex algebra is no longer conformal. Rather it is \emph{quasiconformal} and still satisfies many nice properties one expects from a conformal vertex algebra \cite{Frenkel:2004}.

Chiral algebras of class $\SS$ for theories with at least one maximal marked point will have the structure of modules over (potentially several copies of) $V^{\kappa_c}(\gf)$, and this turns out to be a useful organising structure. The Kazhdan--Lusztig category (at the critical level), $\KL$, is the full subcategory of modules of the universal affine vertex algebra $V^{\kappa_c}(\gf)$ whose objects are $\IZ$-graded $V^{\kappa_c}(\gf)$ modules that have a locally finite action of $\gf$,\footnote{A module $M$ over $V^{k}(\gf)$ is a smooth $\hat{\gf}_{k}$ modules,~\ie, for each $x\in\gf$ and $m\in M$, $(xt^n)m =0$ for sufficiently large $n\in\IN$.
}
\begin{equation}
	\forall m\in M~,\quad \dim\left(U(\gf)m\right)<\infty~.
\end{equation}
This definition allows a full $\IZ$-grading, so includes modules with conformal weight unbounded below. For a VOA associated to a physical four-dimensional theory thought of as an element of $\KL$, this would be pathological. We take $\KL^{ord}$ to denote the full subcategory of $\KL$ admitting a $\IZ_{\geqslant0}$-grading such that each homogeneous subspace is finite dimensional---these can be thought of as positive energy representations.

An important example of an object in $\KL^{ord}$ are the Weyl modules $\IV_\lambda$, defined by
\begin{equation}
	\mathbb{V}_\lambda \coloneqq \mathrm{Ind}_{\gf [\![t]\!]\oplus \IC K}^{\hat{\gf}_{\kappa_c}} V_\lambda~,
\end{equation}
where $V_\lambda$ is the finite-dimensional irreducible representation of $\gf$ with highest (integral, dominant) weight $\lambda$. The universal affine vertex algebra is therefore the vacuum module, $\IV_0$. The Weyl modules have a unique simple quotient, $\IL_{\lambda}$.

A vertex algebra object in $\KL$ is a vertex algebra $V$ with a vertex algebra morphism $\mu_V: V^{\kappa_c}(\gf)\rightarrow V$, such that $V$ is a direct limit of objects in $\KL^{ord}$. The chiral algebras of class $\SS$ will be vertex algebra objects in $\KL$. Given two vertex algebra objects, $V_1,V_2\in\KL$, the product $V_1\otimes V_2$ is a vertex algebra object in $\KL_{-\kappa_g}$, the Kazhdan--Lusztig category at level $-\kappa_g=-2h^\vee$.

\subsection{The Feigin--Frenkel centre}
\label{subsec:ff_centre}

The critical-level universal affine vertex algebra, $V^{\kappa_c}(\gf)$ has a large centre $\zf(\gf)$ known as the Feigin--Frenkel centre. This centre plays a starring role in Arakawa's construction of chiral algebras of class $\SS$, though its properties may not be common knowledge to all four-dimensional physicists.

The Feigin--Frenkel centre is a commutative vertex algebra with strong generators $\{P_i(z)\}$ of (quasi)conformal weight $\Delta_i = d_i$, where $d_i$ are the degrees of the fundamental invariants of $\gf$~\cite{Frenkel:2007} (see Table \ref{tab:FF-twists}). For many purposes, it is useful to work in terms of the universal enveloping algebra of the Fourier modes of the generators of $\zf(\gf)$, which is denoted by $\ZZ$. This is a (commutative) polynomial ring on infinitely many variables,
\begin{equation}
	\ZZ \coloneqq \IC[P_{i(n)}~|~i=1,\dots,\mathrm{rk}\,\gf,\, n\in\IZ]~,
\end{equation}
where we are adopting the mathematician's conventions for the grading of Fourier modes,
\begin{equation}
	P_i(z) = \sum_{n\in\IZ}P_{i,(n)}z^{-n-1}~.
\end{equation}
By contrast, the physicist's gradings are shifted by the (quasi)conformal weights minus one,
\begin{equation}
	P_i(z) = \sum_{n\in\IZ}P_{i,n}z^{-n-\Delta_i}~.
\end{equation}
Two important subalgebras of $\ZZ$ are as follows,
\begin{equation}
	\ZZ_{<0}=\{P_{i,n}\in\ZZ ~|~ n<0\}~, \quad \ZZ_{(<0)} = \{P_{i,(n)}\in\ZZ ~|~ n<0\}~.
\end{equation}
The second of these can be identified with $\zf(\gf)$ itself with respect to the operation of normally ordered product (which is commutative and associative since $\zf(\gf)$ is a commutative vertex algebra), while the first includes a subset of modes that encode the singular terms subleading to the zero mode in the OPE action of the $P_i(z)$.

A positive-energy representation of $\zf(\gf)$ is a $\ZZ$ module $M$ with grading $M=\bigoplus_{d\in p+\IN} M_d$ for some $p\in\IC$, and any object in $\KL^{ord}$ is in particular a positive-energy representation of $\zf(\gf)$. Henceforth, $\ZZ$-Mod will denote the category of $\ZZ$ modules that are positive-energy representations of $\zf(\gf)$.

Important examples of such objects are the Weyl modules $\mathbb{V}_\lambda$, which is an object in $\ZZ$-Mod, but is simply annihilated by all $P_{i,n}$ with $n >0$. Let $\chi_\lambda:\ZZ\rightarrow \IC$ be the character defined by $P_i V_{\lambda} = \chi_\lambda(P_{i,0}) V_{\lambda}$, where $P_i$ is the $i$th fundamental invariant of $\gf$ and $\chi_\lambda(P_{i,n})=0$ otherwise. The zero modes, $P_{i,0}$, act on $\IV_{\lambda}$ by the character $\chi_\lambda$.
The potential for interesting structure then lies with its behaviour as a module over $\ZZ_{<0}$. Defining the quotient
\begin{equation}\label{eq:ff_quot}
	\zf_{\lambda} \coloneqq \ZZ_{<0}/\mathrm{Ann_{\ZZ_{<0}}}(\IV_{\lambda})~,
\end{equation}
the action of $\ZZ_{<0}$ on $\IV_{\lambda}$ factors through $\zf_\lambda$, which manifestly acts freely.

One can ``glue together'' Weyl modules using the Feigin--Frenkel centre by taking tensor products over $\zf_\lambda$. For example, one can define the bimodule $\IV_{\lambda}\otimes_{\zf_\lambda} \IV_{\lambda^*}$, where $\lambda^*$ is the dual weight to $\lambda$ and where the action of $\zf_\lambda$ on $\IV_{\lambda^*}$ is twisted by the Cartan involution $\tau:\mathfrak{z}\rightarrow\mathfrak{z}$~\cite{frenkel2004}.\footnote{This involution corresponding to $\alpha\mapsto - w_0(\alpha)$, where $w_0$ is the longest element of the Weyl group. The morphism $\tau$ sends a finite dimensional, highest-weight representation $V_\lambda$ of $\gf$ to its dual $V_\lambda^* = Hom_{\gf}(V_{\lambda},\IC)$. For simple Lie algebras, the Weyl group elements $w_0$ can be found in~\cite{Bourbaki:2002}. For non-simply laced Lie algebras and the algebras $\ef_7$ and $\ef_8$, $w_0=-1$ and the involution is trivial. For $\af_n$ and $\df_{2n+1}$ the involution, $\tau$, is precisely the $\IZ_2$ outer automorphism twist. For $\df_{2n}$, the involution is again trivial.} This module then has two critical-level actions of $\hat{\gf}$, one on each Weyl module. But the action of the Feigin--Frenkel centres are identified, so this is a module over $\IV_{2,0}\colonequals (V^{\kappa_c}(\gf))^{\otimes_{\zf(\gf)}2}$.

This generalises to the case of $s$ copies of the critical level current algebra,
\begin{equation}
	\IV_{\lambda,s} = \IV_{\lambda}\otimes_{\zf_\lambda}\IV_{\lambda^*}\otimes_{\zf_\lambda}\dots\otimes_{\zf_\lambda}\IV_{\lambda^{(*)}}~,
\end{equation}
where the final module is $\IV_{\lambda}$ if $s$ is odd, or $\IV_{\lambda^*}$ if $s$ is even. This is a module over $\IV_{s,0}\colonequals	 (V^{\kappa_c}(\gf))^{\otimes_{\zf(\gf)}s}$. We denote the contragredient dual module by $D(\IV_{\lambda,s})$. 

A key property of Weyl modules is that they are projective over $t^{-1}\gf[t^{-1}]$ and their contragredient duals are injective over $t\gf[t]$. Similarly $\IV_{\lambda,s}$ and its contragredient dual are projective/injective over each $t^{-1}\gf[t^{-1}]$ and $t\gf[t]$ action, respectively. These modules play an important role in the analysis of \cite{Arakawa:2018egx}, and this will continue to be the case for twisted chiral algebras of class $\SS$.

\subsection{BRST and semi-infinite cohomology}
\label{subsec:BRST_semiinfinite}

Exactly marginal gauging of flavour symmetries in four dimensions is realised at the level of the associated VOA by a BRST construction \cite{Beem:2013sza}. In particular, the associated VOA of an $\ff$-valued vector multiplet is the subalgebra of an adjoint-valued $(b,c)$ ghost system that is annihilated by the zero mode of the $b$ ghosts---these are sometimes called symplectic fermions, or the ``small algebra'' of the ghost system. The vertex algebra of a gauge theory is realised in terms of the vertex algebra $\VV$ of the original ungauged theory and these symplectic fermions. One constructs a BRST current,
\begin{equation}
	J_{\rm BRST}(z) = (J_ic^i)(z) + \frac{1}{2}f_{ij}\,^k(c^ic^jb_k)(z)~,
\end{equation}
where the $J_i$ generate an $\hat{\ff}$ affine current subalgebra of $\VV$ and $f_{ij}\,^{k}$ are the structure constants of $\ff$. Precisely when that current subalgebra has level $\kappa_{\ff}=-2h^\vee_\ff$, the zero mode squares to zero so defines a differential on $C^\bullet=\VV\otimes \langle b,c \rangle$. The grading on $C^\bullet$ is by ghost number, which in many cases can be identified with $U(1)_r$ charge.

Restricting to the kernel of $b_{0}$ (\ie, dropping the zero-modes of the $c$ ghosts) and taking $\ff$ invariants defines the relative subcomplex $C_{rel}^\bullet$,
\begin{equation}
	C_{rel}^\bullet \coloneqq \{v\in \VV\otimes \langle b,c \rangle~|~ b^i_{(0)}v=0 \text{ for}\, i=1,\dots,\dim\, \ff\}^{\ff}~.
\end{equation}
The BRST differential descends to the relative subcomplex and its cohomology defines the relative BRST cohomology. This gives the vertex algebra of the gauged four dimensional theory. Importantly, one cannot neglect the higher cohomology groups in this construction.\footnote{Indeed, in examples such as $\NN=4$ super Yang-Mills theory and the genus-two theory of type $\af_1$ where explicit computations can be performed, there is cohomology outside of degree zero. More generally we expect cohomology outside degree zero at least for theories with enhanced Higgs branches.}

The above construction can be realised as an example of the semi-infinite cohomology functor of Voronov \cite{Voronov:1993aw}; this more abstract formalism is somewhat useful in the analysis of chiral algebras of class $\SS$. The state space of the $\ff$-valued ghost system is the same thing as the space of \emph{semi-infinite forms} on $\ff$ and denoted $\Li{\bullet}\ff$. For an $\hat\ff$-module $M$ at level $\kappa_\ff$, the \emph{Feigin standard complex} is
\begin{equation}
	C^{\frac{\infty}{2}+\bullet}\left(\hat\ff,M\right)=\left(\Li{\bullet}\hat\ff\otimes M,d\right)~,
\end{equation}
where $d=J_{{\rm BRST},(0)}$. The semi-infinite cohomology of $\hat\ff$ with coefficients in $M$ is just the cohomology of this complex,
\begin{equation}
\label{eq:feig_semi}
    \Hi{\bullet}\left(\hat\ff,M\right) = H^{\bullet}\left(C^{\frac{\infty}{2}}\left(\hat\ff,M\right),d\right)~.
\end{equation}
Voronov identifies this cohomology as a two-sided derived functor of the functor of semi-variants, and proves an important vanishing theorem that characterises situations where semi-infinite cohomology is concentrated in degree zero.

\smallskip

\begin{thm}\label{thm:vor_vanish} {\rm (Theorem 2.1 of~\cite{Voronov:1993aw})} Let $M\in\mathcal{O}$ be a semijective module,~\ie, it is injective as a $\ff_+$ module and projective as a $\ff_-$ module. Then
\begin{equation}
    \Hi{p}(\ff,M) = \begin{cases}
    \makebox[.7cm][l]{$M^{\ff_+}_{\ff_-}$}\quad p=0~,\\
    \makebox[.7cm][l]{$0$}\quad p\neq0~,
\end{cases}
\end{equation}
where $M^{\ff_+}_{\ff_-}$ is the projection of $\ff_+$-invariants onto the $\ff_-$-coinvariants. More concretely,
\begin{equation}
    M^{\ff_+}_{\ff_-} = \{ m\in M\,|\,\ff_+m=0\} {\bigg /} \{m\in M\,|\, \ff_+m=0\quad \text{and} \quad m=xm^\prime \,\,\,\text{for some } x\in \ff_-\}~.
\end{equation}
\end{thm}

The relative BRST cohomology described above is also understood as relative semi-infinite cohomology with respect to the $\ff\subset\hat{\ff}_\kappa$. As in~\cite{Voronov:1993aw,Arakawa:2018egx}, we denote this by
\begin{equation}\label{eq:feig_rel}
    \Hi{\bullet}(\hat{\ff}_\kappa,\ff,M)~,
\end{equation}
so the prescription for gauging a flavour symmetry $\ff$ in the associated VOA amounts to taking the relative semi-infinite cohomology. The procedure for gluing together two chiral algebras of class $\SS$ $V_1$ and $V_2$ of type $\gf$ then amounts to taking the relative semi-infinite cohomology,
\begin{equation}
\label{eq:gauge_gluing}
    V_1\circ V_2\colonequals\Hi{\bullet}(\hat{\gf}_{-\kappa_g},\gf,V_1\otimes V_2)~.
\end{equation}
An application of Theorem~\ref{thm:vor_vanish} in the relative context gives the following vanishing theorem, which is important in establishing generalized $S$-duality at the level of vertex algebras.

\begin{thm}~\label{thm:gauge-van} {\rm (Section 3.11.2 of~\cite{Voronov:1993aw} and Proposition 3.4 of~\cite{Arakawa:2018egx})} Suppose $V\in \KL_{-\kappa_g}$ and $V$ is projective as a $U(t^{-1}\gf[t^{-1}])$ module and injective as a $U(t\gf[t])$ module,~\ie, $V$ is a semijective object in $\KL$. Then,
\begin{equation}
	\Hi{p}(\hat{\gf}_{-\kappa_g},\gf,V) = 0\quad \text{for } p\neq 0~.
\end{equation}
\end{thm}

\subsection{Drinfel'd--Sokolov reduction}
\label{subsec:DS_reduction}

Starting with chiral algebras of class $\SS$ for theories with only maximal punctures, one realises the VOAs with non-maximal punctures by a chiral analogue of the Higgsing procedure described in Section \ref{sec:twisted_background}. Namely, this type of Higgsing is accomplished via quantum Drinfel'd--Sokolov (DS) reduction, which we review briefly below.\footnote{This identification of Higgsing with DS reduction was presented as a \emph{conjecture} in ~\cite{Beem:2014rza}. It remains an open problem to prove this from first principles---there is by now extensive evidence for the claim.}

A nontrivial homomorphism $\Lambda:\suf(2)\hookrightarrow \gf$ induces a grading on $\gf$ by the image of the Cartan generator,\footnote{If the grading is genuinely half-integral and not integral, then there is a subtlety and one must introduce an alternate Cartan generator, whose grading is integral, and impose modified constraints in what follows. See~\cite{Feher:1992ed,deBoer:1992sy,deBoer:1993iz} for more details.}
\begin{equation}
    \gf = \bigoplus_{m\in\frac{1}{2}\IZ}\gf_m \,,\quad \gf_m \coloneqq \{t\in\gf\,|\, \mathrm{ad}_{\Lambda(h)}\,(t) = mt\}~.
\end{equation}
Let $\{t_\alpha\}$ denote a basis for the negatively graded part of $\gf$, with corresponding affine currents $J_\alpha$ in $V^{k_c}(\gf)$. Drinfel'd--Sokolov reduction can be thought of as imposing the constraint
\begin{equation}
    J_\alpha = \delta_{\alpha\alpha_-}~,
\end{equation}
where $t_{\alpha_-}$ is the image of the negative weight generator $f\in\suf(2)$. This constraint is imposed via BRST by introducing a ghost system $(c^\alpha,b_\alpha)$ valued in the adjoint of $\bigoplus_{m<0}\gf_{m}$ and forming the BRST current
\begin{equation}
   J_{DS}(z) = \sum_\alpha\left((J_\alpha - \delta_{\alpha\alpha_-})c^\alpha\right)(z) - \frac{1}{2}\sum_{\alpha\beta}^{\gamma} f_{\alpha\beta}^{\phantom{\alpha\beta}\gamma}(b_\gamma c^\alpha c^\beta)(z)~,
\end{equation}
with the BRST differential $d_{DS} = J_{{\rm DS},(0)}$. We define the DS reduction of a vertex algebra $\VV$, which has the structure of a $V^{k}(\gf)$ module, as the cohomology
\begin{equation}
	H^\bullet_{f,DS}(V) = H_{\rm BRST}^\bullet(V\otimes\langle c^\alpha,b_\beta\rangle,d_{DS})~.
\end{equation}
Taking the coefficients to be exactly $V^{k}(\gf)$ gives the $\WW$-algebra $\WW(f,\hat{\gf}_{\kappa})=H^0_{f,DS}(V^{\kappa}(\gf))$.

Like the BRST operation associated to gauging, DS reduction can be phrased in terms of semi-infinite cohomology (see, \eg,~\cite{Arakawa2017introduction}). We will make extensive use of the result~\cite{Arakawa:2007rep}.

\smallskip

\begin{thm}[\textup{Main Theorem 1 of~\cite{Arakawa:2007rep}}]\label{thm:ds_vanish}
     For any $M\in\KL$, $H^{i}_{f,DS}(M)=0$ for $i\neq0$.
\end{thm}

The case of principal DS reduction at the critical level is especially interesting because $\WW(f_{prin},\hat{\gf}_{\kappa_c})\cong \zf(\gf)$. For non vacuum modules, there is an isomorphism~\cite{Frenkel:2010}
\begin{equation}\label{eq:ds_weyl}
	H^0_{f_{prin},DS}(\IV_\lambda) \cong \zf_\lambda~.
\end{equation}
In general we get an exact functor $H^0_{f_{prin},DS}: \KL\rightarrow \ZZ$-Mod. We will mostly be interested in principal DS reduction in what follows, and will suppress the explicit $f_{prin}$ argument from now on. 

Pictorially, principal DS reduction corresponds to removing a puncture. The key idea in the construction of chiral algebras of class $\mathcal{S}$ is to develop a method to \emph{invert} this operation and add punctures back onto the UV curve. We discuss this in the following sections.

\subsection{Feigin--Frenkel gluing}
\label{ssec:chiral_background/ff-gluing}

For chiral algebras of class $\SS$ with maximal punctures, the current subalgebras are all at the critical level, so there could theoretically be several copies of the Feigin--Frenkel centre present. This turns out to not be the case; in fact the current subalgebras all \emph{share a common Feigin--Frenkel centre}, so the chiral algebras of class $\SS$ with $s$ punctures are modules over $\IV_{0,s}$.

This phenomenon of shared FF centres is a chiral analogue of certain well-known Higgs branch chiral ring relations for theories of class $\SS$. We recall that for theory of type $\af_n$ associated to a UV curve with $s$ punctures, there are $s$ moment map operators $\mu_s$ subject to the relations
\begin{equation}\label{eq:Higgs_FF_relations}
	\Tr \mu_1^k=\Tr \mu_2^k=\cdots=\Tr \mu_s^k~,\qquad k=2,\ldots,n+1~.
\end{equation}
More generally, for the $\df_n$ and $\ef_n$ theories there are analogous relations corresponding to the respective fundamental invariants of those algebras. These play a central role in the construction of Higgs branches of class $\SS$ by Ginzburg--Kazhdan~\cite{Ginzburg:2018}.

It is not immediate that these Higgs chiral ring relations lead to the identification of Feigin--Frenkel centres. This is because the Feigin--Frenkel generators are related to, but not equal to, the corresponding Higgs branch operators.\footnote{To be precise, the Feigin--Frenkel operators are identified with the corresponding Higgs branch operators upon passing to the associated graded of the $R$ filtration \cite{Beem:2017ooy}.} To illustrate, consider the case $k=2$ of \eqref{eq:Higgs_FF_relations}. The Higgs branch operators associated to the quadratic fundamental invariant are related to the Segal--Sugawara operators, $P_{1,i}$ but also receive a nonzero contribution from the VOA stress-energy tensor $T$~\cite{Beem:2018duj},
\begin{equation}
	P_{1,i}= \Tr \mu^2_i  + \alpha T~,
\end{equation}
where $\alpha$ is a fixed (nonzero) constant computed in \cite{Beem:2018duj} and $P_{1,i}$ is the quadratic Feigin--Frenkel generator associated to the $i$'th puncture. The Higgs chiral ring relations set the $\Tr \mu^2_i$ equal and, importantly, there is a \emph{unique} $\hat{\CC}_{0(0,0)}$ multiplet (the four-dimensional stress tensor multiplet) so the operator $T$ is the same for each $i$. As a result, the quadratic generators of the Feigin--Frenkel centre are identified across different punctures,~\ie, $P_{1,1}=P_{1,2}=\dots P_{1,s}$. For higher order invariants, the same argument doesn't go through without additional information about the structure of the VOA; in the case of the cubic invariant of $\af_n$ there will be mixing between Higgs branch operators ($\widehat\BB_3$ multiplets) and $\CC_{1(0,0)}$ multiplets, but the uniqueness of the latter is not apparent. Nevertheless, precisely for the class $\mathcal{S}$ theories, the known expression for the Schur index indicates that the identification of the higher Feigin--Frenkel generators should indeed hold.

As in the definition of $\IV_{0,2}$, given two modules $M,N\in\ZZ$-Mod, we can equip $M\otimes N$ with a $\ZZ$-Mod structure via a twisted diagonal action,
\begin{equation}
	P\otimes1 - 1\otimes\tau(P)\quad \forall P\in\ZZ~.
\end{equation}
The twist by $\tau$ is a matter of convention, but will prove convenient in later constructions when Feigin--Frenkel gluing and gauging appear together---roughly speaking, one can anticipate this issue by observing that when gauging together two VOAs, the gauge invariant states arise from pairing representations with highest weights $\lambda$ and $\lambda^*$ respectively, so if one wishes to be able to identify a common action of $\ZZ$ on these two representations then there should be a twist by $\tau$.

As an abelian Lie algebra, $\ZZ$ has a semi-infinite structure and so can be used to define semi-infinite cohomology with coefficients in some object of $\ZZ$-Mod. Specifically, we will have coefficients of the form $M\otimes N$ for $M,N\in\ZZ$-Mod. The Feigin standard complex is then given by
\begin{equation}
	C^\bullet(\ZZ,M\otimes N) \coloneqq  M\otimes N \otimes\Li{\bullet}\ZZ~,
\end{equation}
with BRST current
\begin{equation}
	Q(z) = \sum_{i=1}^{\mathrm{rk}\,\gf}(P_ic^i)(z)~,
\end{equation}
whose zero mode $Q_{(0)}$ acts as a differential for the cochain complex. Suppose we have two objects $V_1$ and $V_2$ in $\ZZ$-Mod, we define Feigin--Frenkel (FF) gluing as
\begin{equation}
	V_1\ast V_2\coloneqq \Hi{0}(\ZZ,V_1\otimes V_2)~.
\end{equation}
In the case when $V_1$ and $V_2$ are vertex algebras, $V_1\ast V_2$ is also a vertex algebra. The vector space $\V_1\otimes V_2$ also has the structure of a $\ZZ$-module via the left (or right) action of $\ZZ$ on the tensor factors, which have become identified through the gluing.

As was the case with the gauge theory BRST problem, we have a vanishing theorem here.
\begin{thm}[Theorem 9.10 of~\cite{Arakawa:2018egx}]~\label{thm:van-ff}
    Let $M\in\ZZ$-Mod be free as a $\ZZ_{(<0)}$ module, then
    \begin{equation*}
        \Hi{i}(\ZZ,M) = 0 \quad \textrm{for}~~i<0~.
    \end{equation*}
\end{thm}
This is a weaker conclusion than in the vanishing theorem \ref{thm:gauge-van}, as the cohomology is not necessarily concentrated in degree zero. Nevertheless for many purposes it is sufficient.\footnote{From a derived perspective, truncating the cohomology at degree zero is somewhat unnatural. We will return to this point in a later section.}

We see that the BRST procedure enforces that the action of the Feigin--Frenkel centre on $V_1$ and $V_2$ are identified. This is a chiral version of the Hamiltonian reduction procedure described in~\cite{Ginzburg:2018}.

\subsection{Vertex algebras at genus zero}
\label{ssec:chiral_background/genus-zero}

The pieces are now all in place to recount the construction of (genus zero) untwisted chiral algebras of class $\SS$ from \cite{Arakawa:2018egx}.

\subsubsection*{Chiral differential operators and the cylinder}
\label{ssec:chiral_background/cylinder}

The starting point of the construction of~\cite{Arakawa:2018egx} is the cylinder VOA; from here one can define the cap chiral algebra by Drinfel'd--Sokolov reduction and, it turns out, construct all genus zero VOAs by FF gluing. The form of the cylinder algebra for $\af_n$ theories was identified concretely in \cite{Beem:2014rza}, but it was subsequently recognised in \cite{Arakawa:2015ro} that this reproduced a more general, and purely algebraic, construction that makes no explicit reference to four dimensional physics. The construction is universal and depends only on the choice of an algebraic group $G$, which here we take to be semisimple.

Starting from such a $G$, the arc space $J_\infty G$ is a scheme over $\IC$ defined as
\begin{equation}
    J_\infty G \coloneqq Hom_{Sch/\IC}(D, G)~,
\end{equation}
where $D$ is the formal disc $D=\mathrm{Spec}\, \IC[\![t]\!]$. See~\cite{Arakawa2017introduction} for more concrete discussion of arc spaces. The Lie algebra $\gf=Lie(G)$ acts on $\OO(G)$ via derivations, and this action lifts to one of $J_\infty(\gf)\cong \gf[\![t]\!]$ on $\OO(J_\infty G)$. Therefore, $\OO(J_\infty G)$ has the structure of a $\gf[\![t]\!]\oplus \IC K$ module where $K$ acts as the level $k\in\IC$. One can further produce a $\hat{\gf}_{\kappa_c}$ module via induction, which defines the \emph{chiral differential operators} on $G$ (at the critical level)~\cite{ag} (see also~\cite{Gorbounov:1992,Gorbounov:2001,Gorbounov:2004, Malikov:1999}),
\begin{equation}
	\D_G = U(\hat{\gf}_{\kappa_c})\otimes_{U(\gf[\![t]\!]\oplus \IC K)}\OO(J_\infty G)~.
\end{equation}
This has the structure of a conformal vertex algebra~\cite{Gorbounov:2001} with central charge
\begin{equation}
	c_{2d} = 2\,\dim\, G~,
\end{equation}
which matches with the central charge for a twice punctured sphere (see, \eg,~\cite{Chacaltana:2012zy}).

By construction, $\D_G$ is a vertex algebra object in $\KL$ and there is an embedding of the universal affine vertex algebra $\pi_L:V^{\kappa_c}(\gf)\hookrightarrow\D_G$. This vertex algebra homomorphism is induced by the embedding of $\gf$ as left invariant vector fields of $G$. The Lie algebra $\gf$ is also isomorphic to the right invariant vector fields of $G$ and this embedding is also lifted to a vertex algebra homomorphism $\pi_R:V^{\kappa_c}(\gf)\hookrightarrow \D_G$, such that the images of $\pi_L$ and $\pi_R$ commute~\cite{Gorbounov:2001}. The left and right embeddings of $V^{\kappa_c}(\gf)$ restrict to embeddings of the Feigin--Frenkel centre $\mathfrak{z}(\gf)$, and the two embeddings of the Feigin--Frenkel centre coincide~\cite{frenkel2004},
\begin{equation}
	\pi_L(\mathfrak{z}(\gf))\cong\pi_R(\mathfrak{z}(\gf)) \cong (\D_G)^{\gf[t]\times\gf[t]}~.
\end{equation}
$\D_G$ is free as a module over $U(t^{-1}\gf[t^-1])$ and cofree over $U(t\gf[t])$~\cite{Arakawa:2018egx}; thus the conditions of Theorem~\ref{thm:gauge-van} are met and for any $M\in \KL$, the cohomology $\D_G\circ M$ is concentrated in degree zero, and furthermore~\cite{Arakawa:2018egx},
\begin{equation}
	\Hi{i}(\hat{\gf}_{-\kappa_g},\gf,\D_G\otimes M) \cong \Hi{i}(\hat{\gf}_{-\kappa_g},\gf, M\otimes\D_G) \cong \delta_{i,0}M~.
\end{equation}
It is easy to see that $\D_G$ is the unique vertex algebra object in $\KL$, which satisfies this property. Physically, the chiral algebra of the cylinder must act as the identity object under $\circ$, when restricted to the chiral algebras of class $\SS$. This is a weaker condition than what is satisfied by $\D_G$. One could a imagine situation in which the chiral algebra of the cylinder acts as the identity for the chiral algebras of class $\SS$ but not on all of $\KL$. In this case, the chiral algebra of the cylinder may not necessarily be isomorphic to $\D_G$. However, this pathological possibility seems to not be realised.

\subsubsection*{Equivariant affine $\WW$-algebras and the cap}

Starting from the cylinder VOA, the cap algebra is recovered by completely reducing one maximal puncture. Arakawa has named the resulting VOA the (principal) \emph{equivariant affine $W$-algebra} $\W_G$ (it is an affine analogue of the equivariant $W$-algebra of \cite{Losev:2007Qua}),
\begin{equation}\label{eq:defcap}
	\W_G \coloneqq H^{0}_{DS}(\D_G)~.
\end{equation}
In the usual way, this vertex algebra inherits a conformal structure from the cylinder, with central charge
\begin{equation}
	c_{\W_G} = \dim\,\gf +\mathrm{rk}\,\gf+ 24 \kappa_g(\rho,\rho^\vee)~,
\end{equation}
where $\rho$ is half the sum of all positive roots and $\rho^\vee$ is the half sum of all positive coroots.

By Propositions $6.4$ and $6.5$ of~\cite{Arakawa:2018egx}, $\W_G$ is free over $U(t^{-1}\gf[t^{-1}])$ and cofree over $U(t\gf[t])$, and so is the $\hat{\gf}_{-\kappa_g}$ module $\W_G\otimes M$. The cohomology when gauging is therefore concentrated in degree zero. Theorem $6.8$ of~\cite{Arakawa:2018egx} then establishes that for any $M\in\KL$, and in particular for vertex algebra objects in $\KL$,
\begin{equation}
    \Hi{\bullet}(\hat{\gf}_{-\kappa_g},\gf,\W_G\otimes M) = H^0_{DS}(M)~.
\end{equation}
When $M$ is a chiral algebra of class $\SS$, this corresponds, pictorially, to the fact that gluing a cap and a surface together along a maximal puncture has the effect of closing the puncture. 

Feigin--Frenkel gluing a cap onto another vertex algebra provides a sort of inverse to the principal DS reduction functor; by FF gluing a cap onto a vertex algebra $V\in\ZZ$-Mod we provide it with a $V^{\kappa_c}(\gf)$ action, and it becomes a vertex algebra object in $\KL$. The cap is free over $\ZZ_{(<0)}$~\cite{Arakawa:2018egx} and so $\W_G\ast-:\ZZ\text{-Mod}\rightarrow\KL$ is a left-exact functor that acts, almost, as an inverse to $H^0_{DS}$. For any vertex algebra object $V$ in $\KL$, Theorem $9.11$ of~\cite{Arakawa:2018egx} guarantees that
\begin{equation}
\label{eq:dsid}
	V\cong \Hi{0}(\ZZ,\W_G\otimes H^{0}_{DS}(V))~,
\end{equation}
and for the subcategory of $V$-module objects in $\KL$,
\begin{equation}
	\Hi{0}(\ZZ,\W_G\otimes H^0_{DS}(-))~,
\end{equation}
is the identity functor. Indeed, this means that the cylinder VOA can be recovered from the equivariant affine $W$-algebra by FF gluing two caps together \cite{Arakawa:2018egx},
\begin{equation}
	\Hi{0}(\ZZ,\W_G\otimes\W_G)\cong \D_G~.
\end{equation}
This is only almost an inverse functor because $\KL$ and $\ZZ$-Mod are not equivalent. If we restrict to the subcategory $\KL_0$ of $\ZZ$-Mod, which is the image of $H^0_{DS}$, then $H^0_{DS}(\Hi{0}(\ZZ,\W_G\otimes V))\cong V$ for any $V\in\KL_0$.

\subsubsection*{General genus zero VOAs}

This situation leads to a conceptually straightforward construction of all genus zero chiral algebras of class $\SS$: one starts with $\D_G$ and repeatedly applies $\W\ast-$ to add punctures. If one takes for granted the general expectations of chiral algebras of class $\SS$ as outlined in \cite{Beem:2014rza}, then this must give the correct answer. However, the formalism of \cite{Arakawa:2018egx} allows Arakawa to go further and demonstrate, \emph{a priori}, the expected behaviour of these VOAs \emph{vis \`a vis} their gluing and duality properties.

To this end, the construction given in~\cite{Arakawa:2018egx} is superficially different, amounting to gluing together $s$ cap VOAs \emph{at once} rather than iteratively adding punctures. Take the chain complex
\begin{equation}
\label{eq:ara_chain_cons}
	C^\bullet\left(\bigoplus_{i=1}^{s-1}\ZZ^{i,i+1},\W^{s}\right)\coloneqq \W^{s}\otimes\bigg( \Li{\bullet}(\mathfrak{z}(\gf))\bigg)^{s-1}~,
\end{equation}
for $s>1$, with differential
\begin{equation}\label{eq:ffmulti}
	\begin{split}
		Q(z)  &= \sum_{i=1}^{s-1} Q^{i,i+1}(z)~,\\
		Q^{i,i+1}(z) &= \sum_{j=1}^{\mathrm{rk}\,\gf} (\rho_i(P_j)-\rho_{i+1}(\tau(P_j)))\rho_{gh,i}(c^j)(z)~,
	\end{split}
\end{equation}
where $\rho_i$ represents the action of $\mathfrak{z}(\gf)$ on the $i$-th factor of $\W$ and $\rho_{gh,i}(c^j)$ acts on the $i$-th factor of the ghost system $\Li{\bullet}(\mathfrak{z}(\gf))$. The vertex algebra of a sphere of type $\gf$ with $s$ maximal punctures is then defined to be
\begin{equation}\label{eq:voaspheres}
\begin{split}
    \V_{G,1} &= \W~,\\
    \V_{G,s} &= \Hi{0}(\bigoplus_{i=1}^{s-1}\ZZ^{i,i+1},\W^{s})~.
\end{split}
\end{equation}
It is a lemma of~\cite{Arakawa:2018egx}, that this agrees with the recursive construction of the genus zero vertex algebras.

As a consequence of their construction in terms of the cylinder VOA, the vertex algebras $\V_{G,s}$ are endowed with an increasing filtration~\cite{Arakawa:2018egx}, $\FF^{\bullet}\V_{G,s}$ such that $\FF^0\V_{G,s}=0$ and
\begin{equation}
	\frac{\FF^i\V_{G,s}}{ \FF^{i-1}\V_{G,s} }\cong \IV_{s,\lambda}~,
\end{equation}
for some integral, dominant weight $\lambda$. They further possess a decreasing filtration $\GG^{\bullet}\V_{G,s}$, with $\GG^{0}\V_{G,s} = \V_{G,s}$ and successive quotients isomorphic to $D(\IV_{s,\lambda})$ for some integral dominant $\lambda$. This means that the genus zero chiral algebras are limits of projective and injective objects in $\KL$, and so, importantly, satisfy the requirements of Theorem~\ref{thm:gauge-van}. This means that when gluing these VOAs together via gauge theory BRST, the cohomology $\Hi{\bullet}(\hat{\gf}_{-\kappa_g},\gf,\V_{G,s}\otimes \V_{G,s'})$ is concentrated in degree zero, and via a spectral sequence argument this establishes the associativity of gauging for the genus zero chiral algebras.\footnote{Furthermore, all cohomologies being concentrated in degree zero is compatible with the expectation that in genus zero there is no residual gauge symmetry on the Higgs branch, and so no Hall--Littlewood chiral ring beyond the Higgs chiral ring.}

An interesting feature of this construction is that it trivialises $S$-duality. The caps involved in the construction are identical and the gluing is done simultaneously. The $\V_{G,s}$ are manifestly invariant under permutations. Since the construction makes no reference to a choice of pants decomposition, there is a disconnect to the physical chiral algebras of class $\SS$---built out of the BRST gauging of trinion algebras. Arakawa has shown in~\cite{Arakawa:2018egx} that the operations of FF-gluing and BRST gauging are compatible, in the sense that the cap gluing construction agrees with the pants decomposition.

\section{\label{sec:fixtures}Twisted trinions from mixed Feigin--Frenkel gluing}

Arakawa's construction is universal and is well defined for any semisimple $\gf$---including non simply laced Lie algebras.\footnote{The same is true for the putative construction of Moore--Tachikawa symplectic varieties in \cite{Ginzburg:2018}.} At the same time, there is no construction for vertex algebras that have both $\gf_u$ and $\gf_t$ symmetries (both twisted and untwisted punctures). A calculation comparing the character of~\cite{Arakawa:2018egx} and the index of~\cite{Lemos:2012ph} shows that the two vertex algebras cannot be isomorphic; we show this in Appendix~\ref{app:rewriting_characters}. There are some notable exceptions; namely, the index of the twisted cylinder and the character of $\D_{G_t}$ agree. Consequently, so must the index of the twisted cap and the character of $\W_{G_t}$.\footnote{Strictly speaking, the characters of the cylinder and cap VOAs don't exist due to infinite-dimensional weight spaces. However, one can proceed formally by working term-by-term in the sum over integral dominant weights; this can be understood from a vertex algebra perspective as considering the decomposition into blocks belonging to $\KL^\lambda$ for each $\lambda\in P^+$.} This suggests that $\D_{G_t}$ and $\W_{G_t}$ will be key building blocks for our construction; indeed, the algebras associated to the sphere with two maximal twisted punctures and one maximal/one minimal twisted puncture should be identified with $\D_{G_t}$ and $\W_{G_t}$, respectively.

In this section, we posit a construction for more general genus zero mixed vertex algebras. Our proposal takes advantage of Feigin--Frenkel gluing to define the mixed trinion with one untwisted and two twisted punctures. From this, in principle, all other genus zero vertex algebras can be constructed. Recalling our discussion of residual gauge symmetries in Section~\ref{ssec:twisted_background/HL}, we have some expectations regarding the behaviour of these mixed vertex algebras under gauging. In particular, unlike the untwisted case, we expect that the BRST cohomology of the \emph{untwisted} gauging of two such vertex algebras will not be concentrated in degree zero, but will rather enjoy some derived enhancement.

To establish such properties regarding untwisted and twisted gauging, we need to examine the decomposition of the mixed vertex algebras into semijective objects in both twisted and untwisted $\KL$ categories. This will be depend in a crucial way on the structure of critical-level Weyl modules over $\hat\gf_{t,\kappa_c}$ as modules over the untwisted Feigin--Frenkel centre $\zf(\gf_u)$. This will lead us to a technical interlude involving opers.

In this section we assume that we are dealing with commuting twists, either $\IZ_2$ or $\IZ_3$. We will remark on some peculiarities of the $\IZ_3$ case near the end of the section.

\subsection{\label{subsec:twisted_untwisted_FF}The (un)twisted Feigin--Frenkel centre and the mixed trinion}

The mixed vertex algebras should simultaneously be vertex algebra objects in $\KL_u$ and in $\KL_t$, so they will admit actions of both Feigin--Frenkel centres. The construction of $\V_{G,s}$ suggests that the action of these Feigin--Frenkel centres should be identified, but of course the twisted and untwisted centres are not isomorphic. It will be useful, therefore, to first examine how the actions of these two Feigin--Frenkel centres interact with each other.

Let $\sigma\in\mathrm{Out}(\gf_u)$ be an outer automorphism, $\gf_u^\sigma$ be the $\sigma$-invariant subalgebra of $\gf_u$, and $\gf_t = \left( \gf_u^\sigma \right)^\vee$. There exists a projection~\cite{Fuchs1996},
\begin{equation}
	\pi_\sigma:\mathfrak{h}_u\twoheadrightarrow \mathfrak{h}_t~,
\end{equation}
from the Cartan subalgebra of $\gf_u$ to that of $\gf_t$ that projects to elements that are invariant under $\sigma$. The outer automorphism lifts to an automorphism of $U(\gf_u)$ and we have a surjection
\begin{equation}
	Z(U(\gf_u))\twoheadrightarrow Z(U(\gf_t))~,
\end{equation}
which is just the projection of the centre of $U(\gf_u)$ to its $\langle \sigma\rangle$-coinvariants,~\ie,~we set the $\sigma$-non-invariant generators of $Z(\gf_u)$ to zero. The action of $\sigma$ can be lifted to $\hat{\gf}_{u,\kappa_c}$ according to $\sigma(xt^n)=\sigma(x)t^n$. This gives a projection
\begin{equation}\label{eq:FF_projection}
	\ZZ_u \twoheadrightarrow \ZZ_t\cong (\ZZ_u)_\sigma~,
\end{equation}
where $(\ZZ_u)_\sigma$ is the space of $\langle\sigma\rangle$-coinvariants of the untwisted Feigin--Frenkel centre.

The projection $\pi_\sigma$ also induces an embedding of weight spaces
\begin{equation}
	\iota: P^+_t \hookrightarrow P^+_u~,
\end{equation}
with image $\iota(P_t^+)$ equal to the subset of elements in $P^+_u$ that are invariant under the action of $\sigma$. For example if $\gf_u=D_n$ and $\gf_t=C_{n-1}$ (so $\sigma$ is the generator of the $\IZ_2$ outer automorphism), we have
\begin{equation}
	\iota(\lambda_1,\lambda_2,\dots,\lambda_{n-1}) = (\lambda_1,\lambda_2,\dots,\lambda_{n-1},\lambda_{n-1})~.
\end{equation}
In~\cite{Lemos:2012ph}, this was indicated with the notation $\lambda'=\lambda$. We will abuse notation and use $\lambda$ for both the weight in $P_t^+$ and its image under $\iota:P_t^+\hookrightarrow P_u^+$. So $V_\lambda^t$ denotes the finite dimensional irreducible representation of $\gf_t$ with highest weight $\lambda$, and $V^u_{\lambda}$ is the finite dimensional irreducible representation of $\gf_u$ with highest weight $\iota(\lambda)$.

Given an object $M\in\ZZ_t$-Mod, we can lift it to a module in $\ZZ_u$-Mod via the restriction of scalars associated to $\ZZ_u\twoheadrightarrow \ZZ_t$, giving a functor $\ZZ_t\text{-Mod}\rightarrow\ZZ_u\text{-Mod}$. The mixed vertex algebras must simultaneously be objects in the $\KL$ categories of $\gf_u$ and $\gf_t$. In the case where there are no untwisted punctures, the algebra must still have a $\ZZ_u$-Mod structure. This is because any such algebra can be obtained via DS reduction of a vertex algebra object in $\KL_u$, and the DS reduction functor lands in $\ZZ_u$-Mod.

The first mixed object we would like to construct is the vertex algebra corresponding to $\CC_{0,1,1}$, which has two $\KL_t$ actions and one $\KL_u$ action. The most natural guess for a construction of this algebra would be to FF glue an untiwsted cap ($\W_u$) to the twisted cylinder ($\D_t$) via the untwisted Feigin--Frenkel centre. The preceeding discussion assures us that $\D_t$ can be lifted from $\ZZ_t$-Mod to $\ZZ_u$-Mod, and as such we propose that the mixed trinion associated to $\CC_{0,1,1}$ is
\begin{equation}
	V_{1,1}\coloneqq\Hi{0}(\ZZ_u,\W_u\otimes\D_t)~.
\end{equation}
This mixed trinion is the most important vertex algebra of our construction. Indeed, it is the basic building block of the TQFT along with the untwisted trinion $\V_{G_u,3}$. Before proceeding to more general surfaces, it will be useful to establish the key properties of $\V_{1,1}$.

The genus zero vertex algebras, $\V_{G,s}$ have filtrations whose succesive quotients are $\IV_{\lambda,s}$ or $D(\IV_{\lambda,s})$, for some integral dominant $\lambda$. This is key in establishing that $\V_{G,s}$ are semijective in $\KL_u$, which in turn implies that gauging is concentrated in degree zero due to Theorem~\ref{thm:gauge-van}. Such a filtration would be extremely useful in establishing the properties of $\V_{1,1}$ under the various types of gluing.

Ideally, in the mixed case the successive quotients would simultaneously be objects in $\KL_u$ and $\KL_t$, and the natural proposal for such an object would be of the form
\begin{equation}
	\IV_{\lambda'}^u\otimes_{\zf^u_{\lambda'}}\IV_{\lambda}^t
\end{equation}
where $\lambda\in P_t^+$ and $\lambda'\in P_u^+$. While it is obvious that $\ZZ_t$ modules can be lifted to $\ZZ_u$ modules by way of the projection $\ZZ_u\twoheadrightarrow \ZZ_t$, it is not so obvious that this should hold for the quotient modules $\zf^u_{\lambda'}$, which are more complicated. Furthermore, there is a question of how $\lambda$ and $\lambda'$ are related. Altogether, there is room for doubt over whether the suggested tensor product over $\zf^u_{\lambda'}$ is even well-defined. To show that such a lift does indeed exist---and that the tensor product is well defined---for special values of $\lambda'$, we will work in the language of opers.

\begin{table}
	\centering
	{\setlength{\extrarowheight}{2.0pt}
\begin{tabular}{ c | c | c | c }
 $\gf_u$ & $\mathfrak{z}(\gf_u)$ & $\gf_t$ & $\mathfrak{z}(\gf_t)$\\
\hline
\hline
$\af_{2n-1}$ & $P_2, P_3,\dots,P_{2n}$ & ${\mathfrak b}_n$ & $P_2,P_4,\dots,P_{2n}$ \\ \hline
$\af_{2n}$ & $P_2, P_3,\dots,P_{2n+1}$ & ${\mathfrak c}_n$ & $P_2,P_4,\dots,P_{2n}$\\ \hline
$\df_{n}$ & $\tilde{P}_n;P_2,P_4,\dots,P_n,\dots,P_{2n-2}$ & ${\mathfrak c}_{n-1}$ & $P_2,P_4,\dots,P_{2n-2}$\\ \hline
$\ef_{6}$ & $P_2,P_5,P_6,P_8,P_9,P_{12}$ & $\ff_4$ & $P_2,P_6,P_8,P_{12}$\\ \hline
$\df_4$ & $\tilde{P}_4,P_2,P_4,P_6$ & $\gf_2$ & $P_2,P_6$
\end{tabular}}

	\caption{\label{tab:FF-twists}The monomial generators of the Feigin--Frenkel centres of the untwisted algebra $\gf_u$ and its associated twisted algebra $\gf_t$. Note that algebras of type $\df_n$ have two generators of degree $n$, only one of which is invariant under the outer automorphism. For the $\df_4$ case, neither generator of degree four is invariant under the $\IZ_3$ outer automorphism.}
\end{table}

\subsection{\label{subsec:opers_and_FF}Opers and the Feigin--Frenkel centre}

This subsection is largely self-contained and may be skipped over by the reader who is less interested in these technical details. At the end of this subsection, we formulate a theorem in the language of opers that implies a lifting of $\zf_{t,\lambda}$ modules to $\zf_{u,\lambda'}$ modules. In the following subsection we will move on to investigate the properties of $\V_{1,1}$.

We will not require the full machinery of opers (as can be defined on general algebraic curves), but instead will proceed with only a selective review of opers on the (punctured) disk with a choice of local coordinate. For more details, we direct the reader to~\cite{Frenkel:2004,Frenkel:2007}.

Let $D=\mathrm{Spec}\, \IC[\![t]\!]$ be the formal disc and $D^\times=\mathrm{Spec}\, \IC(\!(t)\!)$ be the formal punctured disc. Let $G$ be a semisimple algebraic Lie group with $Lie(G)=\gf$, and $B\subset G$ a Borel subgroup of $G$ with $Lie(B)=\mathfrak{b}$. A $G$-oper on $D$ is the triple,
\begin{equation}
	(\mathscr{F},\nabla,\mathscr{F}_B)~,
\end{equation}
where $\mathscr{F}$ is a principal $G$-bundle on $D$, $\nabla$ is a $G$-connection on $\mathscr{F}$, and $\mathscr{F}_B$ is a reduction of $\mathscr{F}$ to the Borel subgroup $B$ such that $\mathscr{F}_B$ is transversal to $\nabla$. A $G$-oper on $D^\times$ is defined analogously.

The space of opers on the disc can be made more concrete as follows. Let $\Lambda_p(h)$ be the image of the Cartan generator of $\mathfrak{sl}_2$ under the principal embedding, $\Lambda_{p}:\mathfrak{sl}_2\rightarrow \gf$. This defines gradings
\begin{equation}
	\gf = \bigoplus_{i} \gf_i~,\quad \mathfrak{b} = \bigoplus_i\mathfrak{b}_i~,
\end{equation}
on $\gf$ and $\mathfrak{b}$ by eigenvalues of $\mathrm{ad}\, \Lambda_p(h)$. Let $p_{-1} = \Lambda_p(f)$ be the image of the $\mathfrak{sl}_2$ lowering operator and let $p_1=\Lambda_p(e)$ be the image of the raising operator. We define $V^{can}$ to be the subspace of $\mathfrak{n}= \bigoplus_{i>0}\mathfrak{b}_i$ that is invariant under $p_1$,~\ie,~it is the subspace spanned by highest weight vectors with respect to the embedded $\mathfrak{sl}_2$. The subspace has the decomposition
\begin{equation}
	V^{can} = \bigoplus_{d_i} V_{d_i}^{can}~,
\end{equation}
where $d_i$ are the exponents of $\gf$. Let $\{p_i\}$ be a basis of $V^{can}$ such that $p_i\in V^{can}_{d_i}$.\footnote{In the case of $\df_{2n}$, the subspace $V^{can}_{2n}$ is two dimensional and we have two linearly independent vectors $p_{2n}$ and $\tilde{p}_{2n}$.} There is a transparent similarity to the generators of $\zf(\gf)$.

We denote the space of all such $G$-opers on the $D$ and $D^\times$ by $\mathrm{Op}_{G}\, D$ and $\mathrm{Op}_{G}\, D^\times$, respectively. Let $t$ be a formal coordinate on $D$. By an appropriate gauge transformation, any $G$-oper can be put into the canonical form~\cite{Frenkel:2007},
\begin{equation}
	\nabla = d + p_{-1} + \sum_{i=1}^{\mathrm{rk}\,\gf}v_i(t)p_i~,
\end{equation}
where $v_i(t)\in\IC[\![t]\!]$ or $v_i(t)\in\IC(\!(t)\!)$ for $D$ or $D^\times$, respectively. A celebrated theorem of Feigin and Frenkel states that~\cite{Feigin:1984},
\begin{equation}
	\zf(\gf) \cong \mathrm{Fun}( \mathrm{Op}_{^L G}\, D)~,
\end{equation}
and
\begin{equation}
	\ZZ \cong \mathrm{Fun}(\mathrm{Op}_{^L G}\, D^\times)~,
\end{equation}
where $^L G$ is the Langlands dual group of $G$.

The Feigin--Frenkel centre acts on the Weyl modules $\IV_\lambda$, and this action factors through a certain quotient of the algebra of functions on opers~\cite{Frenkel:2010}. Namely,
\begin{equation}\label{eq:weyl-proj}
	\ZZ \xrightarrow{\sim} 	\mathrm{Fun}\,\mathrm{Op}_{^L G}(D^\times) \twoheadrightarrow\mathrm{Fun}\, \mathrm{Op}_{^L G}^{\lambda} \cong \mathrm{End}_{\hat{\gf}_{\kappa_c}}(\IV_\lambda)~.
\end{equation}
The quotient arises as follows. The subscheme $\mathrm{Op}_{^L G} (D)^{\lambda, reg}\subset\mathrm{Op}_{^L G}(D)$ is that of opers on the disc with \emph{regular singularity, with residue $\lambda$, and with no monodromy} (see~\cite{Frenkel:2007} for more details). The subscheme $\mathrm{Op}_{^L G}^{\lambda}\subset \mathrm{Op}_{^L G}(D^\times)$ is the image of $\mathrm{Op}_{^L G} (D)^{\lambda, reg}$ under the natural inclusion
\begin{equation}
	\mathrm{Op}_{^L G}(D)\hookrightarrow \mathrm{Op}_{^L G}(D^\times)~.
\end{equation}
The kernel of the projection \eqref{eq:weyl-proj} onto the opers without monodromy is the annihilator ideal
\begin{equation}
	I_\lambda= \mathrm{Ann}_{\zf}\IV_\lambda~.
\end{equation}
Thus, we have
\begin{equation}
	\zf_\lambda \coloneqq \ZZ_{<0}/I_\lambda \cong \mathrm{Fun}\,\mathrm{Op}_{^L G}^{\lambda}~~,
\end{equation}
where $\zf_{\lambda}$ is the quotient we introduced in \eqref{eq:ff_quot}.

The particulars of our case lead to some simplifications. We take $^LG_u$ to be the Langlands dual of the simply connected Lie group, $G_u$, with Lie algebra $\gf_u$. Since $\gf_u$ is simply laced, $\mathrm{Lie}(^LG_u)\cong\gf_u$.

From our discussion in Section~\ref{subsec:twisted_untwisted_FF}, we have the following projection
\begin{equation}
	\mathrm{Fun\, Op}_{^LG_u}(D^\times)\twoheadrightarrow(\mathrm{Fun\, Op}_{^LG_u}(D^\times))_\sigma\cong \mathrm{Fun\, Op}_{^LG_t}(D^\times)~,
\end{equation}
which gives a closed embedding of the schemes
\begin{equation}
	\mathrm{Op}_{^LG_t}(D^\times)\cong(\mathrm{Op}_{^LG_u}(D^\times))^\sigma \hookrightarrow \mathrm{Op}_{^LG_u}(D^\times)~,
\end{equation}
where $(\mathrm{Op}_{^LG_u}(D^\times))^\sigma$ is the spectrum of $(\mathrm{Fun\, Op}_{^LG_u}(D^\times))_\sigma$.
This is nothing more than the surjection $\ZZ_u\twoheadrightarrow \ZZ_t$ of Section~\ref{subsec:twisted_untwisted_FF}, rephrased in the language of opers.

To restrict this surjection to the no-monodromy opers, it is necessary that $\lambda'$ be outer automorphism invariant, specifically $\lambda'=\iota(\lambda)$. We will continue to abuse notation and denote $\iota(\lambda)$ by $\lambda$. We then have the following.
\begin{thm}\label{thm:closed_immersion}
The restriction of the closed immersion
\begin{equation*}
	\mathrm{Op}_{ ^L G_t}(D^\times) \hookrightarrow\mathrm{Op}_{ ^L G_u}(D^\times)~.
\end{equation*}
to the subscheme $\mathrm{Op}_{^L G_t}^{\lambda}$ factors as
\begin{equation*}
	\mathrm{Op}_{^L G_t}^{\lambda}\hookrightarrow\mathrm{Op}_{^L G_u}^{\lambda}\hookrightarrow\mathrm{Op}_{ ^L G_u}(D^\times)~,
\end{equation*}
with each map a closed immersion. Equivalently, the natural surjection
\begin{equation*}
	\mathrm{Fun\, Op}_{^LG_u}(D^\times)\twoheadrightarrow (\mathrm{Fun\, Op}_{^LG_u}(D^\times))_\sigma~,
\end{equation*}
restricts to a surjection
\begin{equation*}
	\mathrm{Fun\, Op}^\lambda_{^LG_u}\twoheadrightarrow (\mathrm{Fun\, Op}^\lambda_{^LG_u})_\sigma~,
\end{equation*}
on the quotient algebras.
\end{thm}
We present a proof of this theorem that makes use of Miura opers and Cartan connections in Appendix~\ref{sec:proof_closed_immersion}. The theorem amounts precisely to the statement that there is a surjection $\zf^u_{\lambda}\twoheadrightarrow \zf^t_{\lambda}$, and this will be crucial in establishing that our construction of mixed vertex algebras of class $\SS$ makes sense and that the vertex algebras so defined satsify certain desirable properties.

\subsection{\label{subsec:mixed_trinion_properties}Properties of the mixed trinion}

In what follows we will make use of the following technical proposition.
\begin{prop}\label{prop:kl_u}
Let $N$ be a $\zf^t_\lambda$-module. Then $N$ is an object of $KL_{u,0}$,~\ie,~there exists some $M\in\KL_u$ such that
\begin{equation*}
	N = H^0_{DS}(u,M)~.
\end{equation*}
\end{prop}
\begin{proof}
We prove this by explicitly constructing an object in $\KL_u$ whose DS reduction is isomorphic to $N$ as an object of $\zf^t_\lambda$-mod. From Theorem~\ref{thm:closed_immersion}, the $\zf^t_{\lambda}$ action can be lifted to an action of $\mathfrak{z}^u_{\lambda}$.

Let $\IV_\lambda^u$ be the Weyl module of $\hat{\gf}_{u,\kappa_c}$ with highest weight $\iota(\lambda)$. The tensor product,
\begin{equation*}
	\IV_\lambda^u\otimes_{\zf_\lambda^u} N~,
\end{equation*}
is well-defined, where the $\zf_\lambda^u$ action on $N$ is from the lift. By construction, this is an object in $\KL_u$ with respect to the $\hat{\gf}_{u,\kappa_c}$ action on the untwisted Weyl module. Let us consider its DS reduction,
\begin{equation*}
	H^0_{DS}(u,\IV_\lambda^u\otimes_{\zf_\lambda^u} N)~.
\end{equation*}
By way of the K\"unneth theorem, noting that $H^0_{DS}(\IV_\lambda)\cong\zf_\lambda^u$~\cite{Frenkel:2010} is manifestly free over itself, we have that
\begin{equation*}
	H^0_{DS}(u,\IV_\lambda^u\otimes_{\zf_\lambda^u} N) \cong \zf^u_\lambda\otimes_{\zf_\lambda^u}N \cong N~,
\end{equation*}
as desired.
\end{proof}
With this in hand we can establish an important theorem concerning the DS reduction of our mixed trinion vertex algebra.
\begin{thm}\label{thm:cyl_kl0}
We have the following isomorphism:
\begin{equation*}
H^0_{DS}(u,\V_{1,1}) \cong \D_t~,
\end{equation*}
so $\D_t \in \KL_{u,0}$.
\end{thm}
This is the statement that the mixed vertex algebra we have constructed can indeed be identified with the UV curve $\CC_{0,1,1}$ insofar as closing the maximal untwisted puncture results in the cylinder of type $\gf_t$. The proof of the above theorem is not entirely straightforward because $H^0_{DS}(u,\Hi{0}(\ZZ_u,\W_u\otimes -))$ is not necessarily the identity on a generic object in $\ZZ_u$-Mod. The full proof of the the theorem is relegated to Appendix~\ref{app:proof_cyl}; here we provide a sketch.

The proof proceeds by first establishing that at the level of formal characters,
\begin{equation}\label{eq:character_inequality}
\ch~ H^0_{DS}(u,\V_{1,1})_{\lambda} \leqslant \ch~ \D_{t,[\lambda]}~.
\end{equation}
In words, each weight space (with fixed generalised eigenvalue under the action of the Feigin--Frenkel zero modes) of $H^0_{DS}(u,\V_{1,1})$ is of dimension less than or equal to that of the corresponding weight space of $\D_t$. This is argued by leveraging the fact that $\D_t$ has an increasing filtration with subquotients $\IV^t_{\lambda,2}$, which are in $\KL_{u,0}$ by Proposition~\ref{prop:kl_u}. We show that passing (in a careful sense) to the associated graded of this filtration can only increase the dimensions of the weight spaces, and on the associated graded the composition of FF gluing and DS reduction acts as the identity; this leads to \eqref{eq:character_inequality}. Since $\D_t$ is simple, we need only construct a non-zero homomorphism $\D_t\rightarrow H^0_{DS}(u,\V_{1,1})$ to establish the isomorphism. The construction of such a homomorphism follows an adaptation of the proof of Theorem $9.9$ of~\cite{Arakawa:2018egx} to this twisted setting.

Theorem \ref{thm:cyl_kl0} will serve as the foundation which lets us build up a number of other important properties of the genus zero mixed vertex algebras.
\begin{prop}\label{prop:trin_free_cofree}
The mixed trinion $\V_{1,1}$ is semijective in $\KL_t$.
\end{prop}
\begin{proof}
The cylinder $\D_t$ has an increasing filtration~\cite{frenkel2004},
\begin{equation*}
	0=N_0\subset N_1\subset N_2 \subset \dots ~,\quad N = \bigcup N_i \cong \D_t~,
\end{equation*}
whose successive quotients take the form
\begin{equation*}
	N_{i}/N_{i-1} \cong \IV_\lambda^t\otimes_{\zf^t_\lambda}\IV_{\lambda^*}^t~,
\end{equation*}
for some $\lambda\in P_t^+$ and $\lambda^*$ the dual representation. Each $N_i/N_{i-1}$ is an object in $\zf_{\lambda}$, so Proposition~\ref{prop:kl_u} applies and $N_i/N_{i-1}\in\KL_{u,0}$. Applying our Theorem~\ref{thm:cyl_kl0} and Theorem $9.14$ of~\cite{Arakawa:2018egx}, the mixed vertex algebra $\V_{1,1} \cong \W_u\ast_u \D_t$ has an increasing filtration
\begin{equation*}
	0= M_0\subset M_1 \subset M_2 \subset \dots ~,\quad M = \bigcup M_i \cong \V_{1,1}~,
\end{equation*}
with successive quotients $M_i/M_{i-1}$ isomorphic to $\Hi{0}(\ZZ_u,\W_u\otimes N_{i}/N_{i+1})$. But from Proposition~\ref{prop:kl_u}, $N_i/N_{i-1}$ is in $\KL_{u,0}$ with $N_i/N_{i-1}\cong H^0_{DS}(u,\IV^u_\lambda\otimes_{\zf^u_{\lambda}}N_{i}/N_{i-1})$. Thus
\begin{equation*}
	M_i/M_{i-1} \cong \IV_\lambda^u \otimes_{\zf^u_\lambda} \left( \IV^t_\lambda\otimes_{\zf^t_\lambda}\IV_{\lambda^*}^t \right)~.
\end{equation*}
The action of $U(t^{-1}\gf_t[t^{-1}])$ on either of the Weyl modules $\IV^t_{\lambda,\lambda^*}$ is projective, so $\V_{1,1}$ is projective over $U(t^{-1}\gf_t[t^{-1}])$. To establish that $\V_{1,1}$ is injective over $U(t\gf_t[t])$, we can repeat the same argument after taking $(\D_t)^{op}$ and using the identification
\begin{equation*}
	(\D_t)^{op}\cong \D_t~.
\end{equation*}
\end{proof}
\vspace{-6pt}

We observe that $\V_{1,1}$ is not semijective in $\KL_u$. Inituitively, this is because the extra generators of the Feigin--Frenkel centre must be set to zero when glued to the twisted cylinder, and these relations spoil projectivity. More precisely, we consider the 
vacuum vector $\ket{0}$. Any state element $P_{i,-n}\ket{0}$ can be written as $\phi\ket{0}$ for some $\phi\in U(t^{-1}\gf_u[t^{-1}])$ a regular element. However, the modes $P_{i,-n}$, which are not invariant under $\sigma$, must act as zero.
As there are regular elements in $U(t^{-1}\gf_u[t^{-1}])$ which act as zero, 
$\V_{1,1}$ cannot be torsion free---so cannot be projective---over $U(t^{-1}\gf_u[t^{-1}])$.

The semijectivity of $\V_{1,1}$ in $\KL_t$ is in accordance with our expectations regarding enhanced Higgs branches/residual gauge symmetries. The twisted class $\SS$ theories for surfaces $\CC_{0,m,1}$ formed by gluing $\CC_{0,1,1}$ along twisted punctures have no residual gauge symmetry, so the gauge theory gluing $\V_{1,1}\circ_t-$ should be concentrated in cohomological degree zero. We have just established this for our $\V_{1,1}$ algebra by showing semijiectivity in $\KL_t$. On the other hand, gluing along the untwisted puncture may lead to a higher-genus local system covering space (\cf\ \ref{ssec:twisted_background/HL}), which falls in line with our observation that $\V_{1,1}$ is not semijective in $\KL_u$.

Having shown that $\V_{1,1}$ has the expected properties under $\circ_t$ and $\circ_u$, we move on to some more intrinsic properties of $\V_{1,1}$. Though our construction ensures that $\V_{1,1}$ is a vertex algebra, it is not at all clear that it has the properties expected from four-dimensional unitarity. Namely, $\V_{1,1}$ must be a conical, conformal vertex algebra with negative central charge. Let us first address the issue of the character.

\begin{prop}\label{prop:mixed_trin_char}
The character of the vertex algebra $\V_{1,1}$ is given by
\begin{equation*}
	\mathrm{Tr}_{\V_{1,1}}(q^{L_0}\mathbf{a}\,\mathbf{b}_1\mathbf{b}_2) = \sum_{\lambda\in P_t^+} \frac{\KK_u(\mathbf{a})\chi^\lambda_u(\mathbf{a})\KK_t(\mathbf{b}_1)\chi_t^\lambda(\mathbf{b}_1)\KK_t(\mathbf{b}_2) \chi_t^\lambda(\mathbf{b}_2)}{\KK_u(\times)\chi_u^\lambda(\times)} ~,
\end{equation*}
where $\mathbf{a}$ is a $G_u$ fugacity and the $\mathbf{b}_i$ are $G_t$ fugacities. Furthermore, $\V_{1,1}$ is conical.
\end{prop}
\begin{proof}
By Theorem~\ref{thm:cyl_kl0}, we have that
\begin{equation*}
	H^0_{DS}(u, \V_{1,1}) \cong \D_t~.
\end{equation*}
As a graded vector space $\V_{1,1}\cong \bigoplus_{\lambda\in P_u^+} \V_{1,1,\lambda}$, and by Proposition $8.4$ of~\cite{Arakawa:2018egx}.
\begin{equation*}
	\mathrm{ch} \V_{1,1,\lambda} = q^{\lambda(\rho^\vee)}\mathrm{ch}\, \mathbb{L}_\lambda \mathrm{ch}H^0_{DS}(u,\V_{1,1,\lambda})~.
\end{equation*}
From the structure of the cylinder, we know that $H^0_{DS}(u,\V_{1,1,\lambda})$ is zero unless it is in the image of $\iota:P_t^+ \hookrightarrow P_u^+$. Therefore, we have
\begin{equation*}
	\mathrm{ch}\, \V_{1,1,\lambda} = q^{\lambda(\rho^\vee)}\mathrm{ch}\, \mathbb{L}^u_\lambda \mathrm{ch} \IV^t_\lambda \otimes_{\zf^t_\lambda}\IV_{\lambda^*}^t~.
\end{equation*}
Recalling Section~\ref{ssec:twisted_background/index} and the appendix of~\cite{Lemos:2014lua}, we can rewrite this in the notation of $\KK$-factors, giving the desired result.

To show that $\V_{1,1}$ is conical, note that the cylinder, $\D_t$, is non-negatively graded and $\lambda(\rho^\vee)\geqslant 0$ since it is integral dominant, with equality only for $\lambda=0$. This establishes that $\V_{1,1}$ is non-negatively graded by weight. The character $\ch~ \V_{1,1,0}=1+\dots$ since $\IL_{\lambda=0}^u$ and $\D_{t,\lambda=0}$ are both conical. Thus the mixed trinion is conical.
\end{proof}

\begin{prop}\label{prop:mixed_trin_conf}
The vertex algebra $\V_{1,1}$ is conformal with central charge
\begin{equation*}
	c_{\V_{1,1}} = 2\dim\,\gf_t + \dim \gf_u -\rk \gf_u -24 \rho_u\cdot \rho_u^\vee~.
\end{equation*}
\end{prop}
\begin{proof}
This proof relies on ideas from the proof of Proposition $10.7$ of~\cite{Arakawa:2018egx}, but requires modifications. The vertex algebras $\W_u$, $\D_t$, and the ghost system $\Li{\bullet}(\zf(\gf_u))$ are all conformal, and we denote their respective conformal vectors by $\omega_\W$, $\omega_{\D}$ and $\omega_{gh}$. Clearly, $\omega = \omega_\W+\omega_{\D}+\omega_{gh}$ is a conformal vector for the complex, $\W_u\otimes \D_t \otimes \Li{\bullet}(\zf(\gf_u))$. We write
\begin{equation*}
	\omega(z) = \sum_{m\in\IZ} L_m z^{-m-1}~,
\end{equation*}
for the associated field.

By Lemma $9.4$ of~\cite{Arakawa:2018egx}, the Feigin--Frenkel centre of $\W_u$ is preserved by the action of $L_m$ for $m\geqslant -1$. For a generator $P_i\in\zf(\gf_u)$,
\begin{equation*}
	\omega(z) P_i(w) \sim \frac{\partial P_i}{z-w} +\frac{(d_i+1) P_i}{(z-w)^2}+\sum_{j=2}^{d_i+2}\frac{(-1)^j j!}{(z-w)^{j+1}} q_j^{(i)}(w)~,
\end{equation*}
where $q_j^{(i)}$ is some homogeneous state in $\zf(\gf_u)$ with weight $d_i-j+2$. Let us denote by $\widetilde{P_i}$ the image of $P_i$ under the projection $\zf(\gf_u)\twoheadrightarrow \zf(\gf_t)$. One then has
\begin{equation*}
	\omega(z) \widetilde{P_i}(w) \sim \frac{\partial \widetilde{P_i}}{z-w} +\frac{(d_i+1) \widetilde{P_i}}{(z-w)^2}+\sum_{j=2}^{d_i+2}\frac{(-1)^j j!}{(z-w)^{j+1}} \widetilde{q_j}^{(i)}(w)~,
\end{equation*}
where we think of $\widetilde{P_i}$ as a state in $\zf(\gf_t)\subset\D_t$. Let $Q(z)$ be the BRST differential for Feigin--Frenkel gluing. We have that
\begin{equation*}
	Q_{(0)}(z) \omega(w) = \sum_{i=1}^{\mathrm{rk}\,\gf}\sum_{j=2}^{d_i+1} \partial^j \big(\rho_\W(q_j^{(i)}) -\rho_{\D_t}(\tau(q_j^{(i)}))\big)c_i~,
\end{equation*}
where $\rho_W:\zf(\gf_u)\hookrightarrow \W_u$ and $\rho_{\D_t}:\zf(\gf_u)\hookrightarrow \D_t$ denote the action of the untwisted Feigin-Frenkel centres on $\W_u$ and on $\D_t$ via the projection to $\zf(\gf_t)$. Unfortunately, $\omega$ does not descend directly to cohomology, so correction terms must be introduced to construct a putative conformal vector in cohomology.

If the right hand side of the above equation equals $Q_{(0)}\chi$ for some state $\chi$, then $\widetilde{\omega} = \omega + \chi$ is $Q$-closed and defines a vector in $\V_{1,1}$. to show that such a $\chi$ exists, it is sufficient to show that $\widetilde{q_j}^{(i)}= \pi(q_j^{(i)})$.

The action of $L_m$ for $m\geqslant -1$ on $\zf(\gf_u)$ is given by the action of $\mathrm{Der}(\OO_D)$ on $\mathrm{Op}_{^LG_u}(D)$, which correspond to infinitesmal coordinate changes on the formal disc~\cite{Frenkel:2007}. We postpone the full argument to Appendix~\ref{sec:proof_closed_immersion} for brevity, but by \eqref{eq:der_od_comm} the action of $\mathrm{Der}(\OO_D)$ intertwines and $\widetilde{q_j}^{(i)}= \pi(q_j^{(i)})$. Thus $\widetilde{\omega}\in \V_{1,1}$.

Now, we wish to show that $\widetilde{\omega}$ is a conformal vector. The vector $\chi$ can be written as
\begin{equation*}
	\chi = \sum_{i=1}^{\mathrm{rk}\, \gf_u}\sum_{j=2}^{d_i+2} \partial^j (\rho_{\W}\otimes \rho_{\D_t}\otimes \rho_{gh})(z_{ij})~,
\end{equation*}
for some $z_{ij} \in \zf(\gf_u)\otimes\zf(\gf_u)\otimes \Li{0}(\zf(\gf_u))$. Therefore, $\widetilde{\omega}_{(i)}=\omega_{(i)}$ for $i=0,1$, so the OPEs agree up to the quadratic pole. Since $\V_{1,1}$ is non-negatively graded by Proposition~\ref{prop:mixed_trin_char}, Lemma 3.1.2 of~\cite{Frenkel:2007} says that all we need to check is that $\widetilde{\omega}_{(3)}\widetilde{\omega}= c/2\ket{0}$ for some $c\in\IC$.,~\ie,~the quartic pole in the $Vir\times Vir$ OPE is a multiple of the identity. However, as $\V_{1,1}$ is conical, the only operator of dimension zero that can appear in the OPE is the identity. Thus, $\widetilde{\omega}$ is a conformal vector of $\V_{1,1}$.

Finally, we wish to show that $\widetilde{\omega}$ and $\omega$ have the same central charge in cohomology. Note that $\V_{1,1}=\sum_{\Delta\in\IN} \V_{1,1}^{\Delta}$ with $\dim\, \V_{1,1}^\Delta<\infty$ and is conical---so $\V_{1,1}$ is of CFT type. As a result, Lemma 4.1 of~\cite{Moriwaki_2020} applies. Namely, for any $x\in \V_{1,1}$, if $\Delta(x)\geqslant 2$ and $x_{(0)}\widetilde{\omega}=0$ then $x \in \mathrm{im}_{\widetilde{\omega}_{(-1)}}(\V_{1,1})$.

Now suppose $\omega'$ was some other conformal vector in $\V_{1,1}$, such that $\omega'_{(1)}$ agrees with $\widetilde{\omega}_{(1)}$. Then, we have that $(\omega'-\widetilde{\omega})_{(0)}\widetilde{\omega}=0$, so $\omega'-\widetilde{\omega} = \partial x$ for some $x\in\V_{1,1}$ with $\Delta=1$. Since $\widetilde{\omega}_{(i)}=\omega'_{(i)}$ for $i=0,1$, we must have that $\partial x_{(i=0,1)}=0$ so $x$ is central in $\V_{1,1}^{\Delta=1}$. Repeating the argument with $(\omega-\widetilde{\omega})_{(1)}$ and using the Borcherds identities, leads us to conclude that $\partial x =0$, so $\widetilde{\omega}$ is unique.

Under DS reduction, $H^0_{DS}(u,\V_{1,1})\cong \D_t$ and the image of $\widetilde{\omega}$ gives rise to a conformal vector in $\D_t$. By a similar argument as the preceeding, one can show that this is the unique conformal vector which agrees with the grading on $\D_t$ (see the proof of Proposition $10.7$ in~\cite{Arakawa:2018egx}). The central charge of $\D_t$ is $c_{\D_t}=2\dim\,\gf_t$ and the central charge of the image of $\widetilde{\omega}$ is related to $c_{\V_{1,1}}$ by
\begin{equation*}
	c_{\D_t} = c_{\V_{1,1}} +\rk\,\gf_u - \dim \gf_{u} +24 \rho_{u}\cdot\rho^\vee_{u} ~,
\end{equation*}
under DS reduction.
\end{proof}
\vspace{-6pt}

With these propositions established, we know that $\V_{1,1}$ obeys many of the desirable properties one would expect from the chiral algebra associated to $\CC_{0,1,1}$. However, we have not explicitly tied this object to the construction of~\cite{Beem:2013sza}. One might reasonably wonder whether the object we have constructed is necessarily the mixed trinion of class $\SS$. We have the following uniqueness result.
\begin{prop}~\label{prop:mixed_unique}
The mixed trinion $\V_{1,1}$ is the unique vertex algebra object in $\KL_u$ such that
\begin{equation*}
	H^0_{DS}(u,\V_{1,1})\cong \D_t~.
\end{equation*}
\end{prop}
\begin{proof}
This follows easily from the fact that $\D_t$ is an object in $\KL_{u,0}$. Suppose $V\in\KL_u$ is a vertex algebra object such that $H^0_{DS}(u,V)\cong \D_t$. Then it must be the case that $V \cong \W_u\ast_u \D_t $, since $\W_u\ast_u(H^0_{DS}(u,-))$ is the identity in $\KL_u$. However, $\V_{1,1}\coloneqq \W_u \ast_u \D_t$, so indeed $V\cong \V_{1,1}$.
\end{proof}
\vspace{-6pt}

Suppose now that $\widetilde{\V}_{1,1}$ is the vertex algebra canonically associated to $\CC_{0,1,1}$ via the four-dimensional construction of~\cite{Beem:2013sza}. This must be a vertex algebra object in $\KL_u$, since it has an action of $V^{\kappa_c}(\gf_u)$ coming from the untwisted puncture and inherits a suitable grading from the physical grading of superconformal quantum numbers. Performing untwisted DS reduction on $\tilde{\V}_{1,1}$ must produce the cylinder $\CC_{0,0,1}$, which has corresponding vertex algebra $\D_t$.\footnote{To see that one must recover the cylinder upon DS reduction, even though this doesn't correspond to a physical four-dimensional theory, one may proceed as follows. The physical vertex algebra $\widetilde{\V}_{1,1}$ must satisfy $\widetilde{\V}_{1,1}\circ_t \widetilde{\V}_{1,1}$. Performing DS reduction, we must have $H^0_{DS}(u,\widetilde{\V}_{1,1}\circ_t\widetilde{\V}_{1,1})\cong \widetilde{\V}_{1,1}$. By means of a spectral sequence argument (we delay this until the next subsection) one can rearrange the order of cohomologies to show $H^0_{DS}(u,\widetilde{\V}_{1,1})\circ_t \widetilde{\V}_{1,1}\cong\widetilde{\V}_{1,1}$, which implies that $H^0_{DS}(u,\widetilde{\V}_{1,1})=\D_t$.}
Proposition~\ref{prop:mixed_unique} then applies, so we have
\begin{equation}
	\widetilde{\V}_{1,1}\cong \W_u\ast_u \D_t \cong \V_{1,1}~.
\end{equation}

Thus far, this has been fairly abstract. Let us provide some concrete observations and predictions. We have argued that the mixed trinion $\V_{1,1}$, as we have constructed it, is the unique vertex algebra that could be associated to $\CC_{0,1,1}$. Let us consider the vertex algebra associated to $\CC_{f}$, which is a genus zero surface with one maximal untwisted puncture, one maximal twisted puncture and an empty twisted puncture. The corresponding vertex algebra is $\V_{f}=H^0_{DS}(t,\V_{1,1})$.

The surface $\CC_f$ does correspond to a physical SCFT, and in particular, when $\gf_u=\df_n$ and $\gf_t=\mathfrak{c}_{n-1}$ the corresponding SCFT is a free hypermultiplet theory (hence the subscript $f$). In this case, $\V_f$ should be a symplectic boson vertex algebra with a commuting $\hat{\gf}_{u,\kappa_c}\times\hat{\gf}_{t,\kappa_c}$ subalgebra. Unfortunately, this does not hold for the other choices of $\gf_u$ (at generic rank). There is, however, another example that has appeared in recent literature. The even rank $\af$-type Lie algebras $\gf_u=\af_{2n}$ have, as their twisted algebras, $\gf_t=\mathfrak{c}_{n}$---unlike the $\df_n$ theories, these SCFTs have global Witten anomalies. For $\gf_u=\af_2$ and $\gf_t=\mathfrak{c}_1=\af_1$, the SCFT is the $\TT_X$ theory of~\cite{Buican:2017fiq}. This was identified as the rank-two $H_2$ $F$-theory SCFT \cite{Beem:2019snk}, and the class $\SS$ realisation was given in~\cite{Beem:2020pry}. In the $\ef_6$ and $\af_{2n+1}$ cases, these vertex algebras remain unstudied to the best of our knowledge.

It is perhaps of technical interest to note that $\V_f$ should be conical, despite the fact that it is obtained from the Feigin--Frenkel gluing of just two caps. In the cases where the two caps are of the same type, one obtains the cylinder $\D_u$ or $\D_t$---neither of which are conical.

The identification of $\V_{f}$ with a symplectic boson system for $\gf_u=\mathfrak{d}_n$ leads to a curious observation. Take the symplectic boson system, $\SS\BB$ associated to the bifundamental representation of $\gf_u\times \gf_t$, its DS reduction must be isomorphic to the DS reduction of $\V_f$. This implies that for $\gf_u = \df_n$,
\begin{equation}\label{eq:ds_free}
	H^0_{DS}(u, \SS\BB) \cong H^0_{DS}(u,\V_f)\cong \W_t~.
\end{equation}
Indeed, this presents an alternate construction for the equivariant affine $\WW$-algebra, $\W_{\cf_{n-1}}$. While we have not established this isomorphism rigorously, we present some supporting evidence for the lowest rank cases in Appendix~\ref{app:ds_free}.

\subsection{\label{ssec:duality/rearr}Rearrangment lemmas}

Having established many key properties of the mixed trinion, which is the building block of the twisted chiral algebras of class $\SS$, we would like to extend our results to other genus zero, mixed vertex algebras. As we increase the number of twisted and untwisted punctures, we are naturally required to consider how the various types of gluing interact with each other. It will be useful in this endeavour to have a collection of rearrangement lemmas that establish the extent to which the various gluings associate.

In this subsection we establish a series of technical rearrangement lemmas concerning the interplay between the various cohomological operations we have defined thus far. The proofs of these lemmas, which are modifications of proofs of~\cite{Arakawa:2018egx}, rely heavily on the machinery of spectral sequences. The reader who is uninterested in highly technical details may wish to skip this section and pick up in the following, where we extend our construction to the vertex algebras associated to $\CC_{0,m,n}$.

Since we are interested in the interplay of gluings, we will need to consider objects that have multiple, commuting $\hat{\gf}_{\kappa_c}$ actions. To be completely precise, one should decorate each $\circ$ in this section with subscripts to indicate the diagonal action that is being gauged. This would be somewhat cumbersome, so we will overload the $\circ$ notation and rely on context to make the relevant actions clear. Our lemmas only ever concern two such actions, for the sake of argument we call them the left and right actions. Suppose $U,V,W$ are in $\KL$, such that $V$ has two actions of $\hat{g}_{\kappa_c}$ and $V$ is in $\KL$ with respect to both actions. One should then interpret the symbol $U\circ(V\circ W)$ as the semi-infinite cohomology $\Hi{\bullet}(\hat{\gf}_{-\kappa_g},\gf,U\otimes\Hi{\bullet}(\hat{\gf}_{-\kappa_g},\gf,V\otimes W))$, where the diagonal action of $\hat{\gf}_{-\kappa_g}$ is with respect to the right $\hat{\gf}_{\kappa_c}$ action on $V$ and the sole $\hat{\gf}_{\kappa_c}$ action on $W$. Similarly, the diagonal action on $U\otimes \Hi{\bullet}(\hat{\gf}_{-\kappa_g},\gf,V\otimes W)$ is with respect to the sole $\hat{\gf}_{\kappa_c}$ action on $U$ and the $\hat{\gf}_{\kappa_c}$ action on $V\circ W$ induced by the left action on $V$.

Similarly, for objects $U,V,W\in\ZZ$-Mod, we will denote the semi-infinite cohomology $\Hi{0}(\ZZ,U\otimes\Hi{0}(\ZZ,V\otimes W))$ by $U\ast(V\ast W)$. Here one should take the action on $V\otimes W$ for the first cohomology and on $U\otimes V$ for the second---recall that the $\ZZ$-action on $V$ descends to the cohomology $V\ast W$.

Hereafter, it should be understood that when there are many $\KL_u$ or $\KL_t$ actions present, we choose two such actions for the purposes of the rearrangement lemmas. For the chiral algebras of class $\SS$ (both untwisted and twisted), all the moment maps (from the same algebra) are related by discrete automorphisms, and one can make such a choice without loss of generality.

First, we recast some of the results of~\cite{Arakawa:2018egx} as rearrangement lemmas.

\begin{lem}\label{lem:rearr_gauge}
Let $U,V,W$ be vertex algebra objects in $\KL$ such that $U$ is semijective in $\KL$ and $V$ has two $\KL$ actions, then
\begin{equation*}
	U\circ(V\circ W)\cong (U\circ V)\circ W~.
\end{equation*}
\end{lem}
\begin{proof}
This proof is similar to the proof of Theorem 10.11 of~\cite{Arakawa:2018egx}. Consider the bicomplex
\begin{equation*}
	C^{\bullet\bullet}= U\otimes V\otimes W\otimes \Li{\bullet}(\gf)\otimes\Li{\bullet}(\gf)~,
\end{equation*}
with the differential $d_1$ acting on $U\otimes V$ and the first $\Li{\bullet}(\gf)$ and $d_2$ acting on $V\otimes W$ and the second $\Li{\bullet}(\gf)$. It is easy to see that $d_1d_2+d_2d_1=0$, so the total complex $C_{tot}^p = \bigoplus_{m+n=p}C^{m,n}$ is a cochain complex with the differential $d = d_1 + (-1)^m d_2$~\cite{Weibel:1994}. There are two spectral sequences converging to the total cohomology of $(C_{tot},d)$
\begin{equation}\nonumber
	\begin{split}
	_IE^{p,q}_2 = \Hi{p}(\hat{\gf}_{-\kappa_g},\gf,U\otimes\Hi{q}(\hat{\gf}_{-\kappa_g},\gf,V\otimes W))~,\\
	_{II}E^{p,q}_2 = \Hi{p}(\hat{\gf}_{-\kappa_g},\gf,W\otimes\Hi{q}(\hat{\gf}_{-\kappa_g},\gf,U\otimes V))~.
	\end{split}
\end{equation}

By Theorem \ref{thm:gauge-van}, the cohomologies $\Hi{p}(\hat{\gf}_{-\kappa_g},\gf,U\otimes-)$ are concentrated in degree zero, so both spectral sequences collapse at the second page. The only nonzero entries are $_{I} E_2^{0,q}$ and $_{II}E_2^{p,0}$. Thus, we have
\begin{equation*}
	_IE^{0,p}_2\cong\ _{II}E^{p,0}_2 \cong H^p_{tot}(C_{tot},d)~,
\end{equation*}
which gives the desired isomorphism.
\end{proof}
\vspace{-6pt}

The vertex algebra objects $\V_{G,s}$ are semijective in $\KL$ so the composition of vertex algebras is associative---in the untwisted setting. Furthermore, gauging at genus zero is always concentrated in degree zero. In the twisted setting, we will have to work harder.

We have a similar result for Feigin--Frenkel gluing.
\begin{lem}\label{lem:rearr_ff}
Suppose $M_1,M_2,M_3\in\ZZ$-Mod and suppose that $M_1$ and $M_3$ are free over $\ZZ_{(<0)}$. Then
\begin{equation*}
	M_1\ast (M_2\ast M_3)\cong (M_1\ast M_2)\ast M_3~.
\end{equation*}
\end{lem}
\begin{proof}
This proof is similar to the proof of Proposition 10.2 of~\cite{Arakawa:2018egx}. Consider the bicomplex
\begin{equation*}
	C^{\bullet\bullet} = M_1\otimes M_2\otimes M_3 \otimes \Li{\bullet}(\mathfrak{z}(\gf))\otimes\Li{\bullet}(\mathfrak{z}(g))~,
\end{equation*}
with differentials $d_{12}$ acting on $M_1\otimes M_2 \otimes \Li{\bullet}(\mathfrak{z}(\gf))$ and $d_{23}$ acting on $M_2\otimes M_3\otimes \Li{\bullet}(\mathfrak{z}(\gf))$. The two differentials anticommute, so we can form the total complex $C_{tot}^n = \bigoplus_{p+q=n}C^{p,q}$ with total differential $d_{tot} = d_{12}+(-1)^q d_{23}$. There are two spectral sequences converging to the total cohomology $H_{tot}$ of $C_{tot}$, whose second pages are given by
\begin{equation}\nonumber
	\begin{split}
	_{I} E_2^{p,q} \coloneqq \Hi{p}(\ZZ,M_1\otimes\Hi{q}(\ZZ,M_2\otimes M_3))~,\\
	_{II} E_2^{p,q} \coloneqq \Hi{p}(\ZZ,\Hi{q}(\ZZ,M_1\otimes M_2)\otimes M_3)~.
	\end{split}
\end{equation}
By Theorem \ref{thm:van-ff}, the entries $E^{p,q}_2$ in either spectral sequence vanish if $p<0$ or $q<0$. Thus, we have $E_2^{00}=E_\infty^{00}$, which gives the isomorphism
\begin{equation*}
	H^0_{tot}(C_{tot},d_{tot}) \cong \Hi{0}(\ZZ,M_1\otimes\Hi{0}(\ZZ,M_2\otimes M_3))\cong\Hi{0}(\ZZ,\Hi{0}(\ZZ,M_1\otimes M_2)\otimes M_3)~,
\end{equation*}
as desired.
\end{proof}
\vspace{-6pt}

Since the algebras $\V_{G,s}$ are free over $\ZZ_{<0}$ (Proposition $10.2$ of~\cite{Arakawa:2018egx}), the above lemma applies, and Feigin--Frenkel gluing $*_u$ is associative. We hold off on analysing associativity of $\ast$ for the twisted algebras, since it is a challenge to understand the twisted FF gluing between two mixed vertex algebras.

We next consider the combination of the two gluing operations, $\ast$ and $\circ$.
\begin{lem}~\label{lem:rearr_gauge_ff}
Suppose $U\in\ZZ$-Mod, $V\in\KL$ and $\W\in\KL$. Additionally, suppose $U$ is free over $\ZZ_{(<0)}$ and $W$ is semijective in $\KL$. Then we have the isomorphism
\begin{equation*}
	U\ast(V\circ W) \cong (U\ast V)\circ W~.
\end{equation*}
\end{lem}
\begin{proof}
Consider the bicomplex
\begin{equation*}
	C^{\bullet\bullet} = U\otimes V\otimes W \Li{\bullet} (\zf(\gf))\otimes \Li{\bullet} (\gf)~,
\end{equation*}
with differentials $d_\gf$ acting on $V\otimes W\otimes \otimes \Li{\bullet}(\gf)$ and $d_\zf$ acting on $U\otimes V \otimes \Li{\bullet}(\zf(\gf))$. The differentials anticommute, so we can form the total complex $C^p_{tot} = \bigoplus_{m+n=p} C^{p,q}$ with differential $d_{tot} = d_{\gf}+(-1)^q d_\zf$. There are two spectral sequences converging to the total cohomology:
\begin{equation}
	\begin{split}
	_{I} E_{2}^{p,q} = \Hi{p}(\ZZ, U \otimes \Hi{q}(\hat{\gf}_{-\kappa_g},\gf, V\otimes W))~,\\
	_{II} E_2^{p,q} = \Hi{p}(\hat{\gf}_{-\kappa_g},\gf,\Hi{q}(\ZZ,U \otimes V)\otimes W)~.
	\end{split}
\end{equation}
The cohmology $\Hi{p}(\ZZ,U\otimes - )$ vanishes for $p<0$ and the cohomology $\Hi{p}(\hat{\gf}_{-\kappa_g},\gf, - \otimes W)$ is concentrated in degree zero. Thus both spectral sequences will collapse at the second page and we have
\begin{equation}
	_{I} E_2^{0,0} \cong\ _{II} E_2^{0,0} \cong H^0_{tot}(C_{tot},d_{tot})~,
\end{equation}
as desired.
\end{proof}
\vspace{-6pt}

Now we come to the rearrangement of twisted and untwisted gluing. The results here are more limited; the obvious generalisations of the above spectral sequence arguments often fail in the twisted setting, since the mixed vertex algebras are not semijective in $\KL_u$. Nevertheless, we will manage to demonstrate that some properties of the gluing of twisted algebras are as we expect. To start with, we have a result for the interchange of twisted and untwisted gauging.
\begin{lem}\label{lem:gauge_ut}
Suppose $V_1$ is in $\KL_u$, $V_2$ is in $\KL_u$ and in $\KL_t$, and $V_3$ is in $\KL_t$. Furthermore, suppose that $V_3$ is semijective for $\gf_t$. Then, we have the isomorphism
\begin{equation*}
	V_1\circ_u(V_2\circ_t V_3) \cong (V_1\circ_uV_2)\circ_tV_3~.
\end{equation*}
\end{lem}
\begin{proof}
As usual, the proof is via a spectral sequence, but with the new feature that cohomology is not necessarily concentrated in degree zero. We define the bicomplex
\begin{equation*}
	C^{p,q} = V_1\otimes V_2\otimes V_3 \otimes \Li{p}(\hat{\gf}_u)\otimes \Li{q}(\hat{\gf}_t)~,
\end{equation*}
with the differential $d_u$ and $d_t$ acting on $V_1\otimes V_2\otimes\Li{\bullet}(\hat{\gf}_u)$ and $V_2\otimes V_2\otimes \Li{\bullet}(\hat{\gf}_t)$ respectively. The differentials anticommute, so we can form the total complex $C_{tot}^m = \bigoplus_{p+q=m}C^{p,q}$ with the differential $d_{tot}= d_u +(-1)^p d_t$. There are two spectral sequences converging to the total cohomology, given by
\begin{equation}\nonumber
	\begin{split}
	_{I} E_2^{p,q} = \Hi{p}(\hat{\gf}_{u,-\kappa_g},\gf_u,V_1\otimes\Hi{p}(\hat{\gf}_{t,-\kappa_g},\gf_t,V_2\otimes V_3))~,\\
	_{II} E^{p,q}_2 = \Hi{p}(\hat{\gf}_{t,-\kappa_g},\gf_t,\Hi{q}(\hat{\gf}_{u,-\kappa_g},\gf_u,V_1\otimes V_2)\otimes V_3)~.
	\end{split}
\end{equation}
Since $V_3$ satisfies the conditions of Theorem~\ref{thm:gauge-van}, the spectral sequences will collapse on the second page and $_{I} E^{p,q}_2=0$ for $q\neq0$ and $_{II} E^{p,q}_2 =0$ for $p\neq0$. Thus we have the isomorphism
\begin{equation*}
	H^m_{tot} =\, _{I} E^{m,0}_2 =\, _{II} E^{0,m}_2~.
\end{equation*}
\end{proof}
\vspace{-6pt}

We can also show that the twisted gauging and untwisted Feigin--Frenkel gluing are nicely compatible,
\begin{lem}\label{lem:ut-cap-gauge}
Let $M$ be an object of $\ZZ_u$-Mod such that it is free over $\ZZ_{u,<0}$ and let $W$ be semijective in $\KL_t$. Suppose, $V$ is a vertex algebra object in $\KL_t$, then we have the follwoing isomorphism.
\begin{equation*}
	M\ast_u(V\circ_t W)\cong (M\ast_u V)\circ_t W~.
\end{equation*}
\end{lem}
\begin{proof}
Let $C^{p,q}$ be the bicomplex
\begin{equation*}
	C^{p,q} = M\otimes V\otimes W \otimes \Li{p}(\mathfrak{z}(\gf_u))\otimes \Li{q}(\hat{\gf}_t)~,
\end{equation*}
with the differentials $d_\mathfrak{z}$ acting on $M\otimes V\otimes \Li{\bullet}(\mathfrak{z}(\gf_u))$ and $d_{\gf}$ acting on $V\otimes W \Li{\bullet}(\hat{\gf_t})$. The differentials $d_\mathfrak{z}$ and $d_\gf$ anticommute, so we can form the total complex $C^m_{tot} = \bigoplus_{p+q=m}C^{p,q}$ with the total differential $d_{tot}=d_\mathfrak{z}+(-1)^p d_\gf$. There are two spectral sequences converging to the total cohomology, given by
\begin{equation}\nonumber
	\begin{split}
	_{I} E^{p,q}_2 = \Hi{p}(\ZZ_u,M\otimes\Hi{q}(\hat{\gf}_{t,-\kappa_g},\gf_t,V\otimes W))~,\\
	_{II} E^{p,q}_2 = \Hi{p}(\hat{\gf}_{t,-\kappa_g},\gf_t,\Hi{q}(\ZZ_u,M\otimes V)\otimes W)~.
	\end{split}
\end{equation}
The cohomology $\Hi{q}(\hat{\gf}_{t,-\kappa_g},\gf_t,-\otimes W)$ is concentrated in degree zero, so both spectral sequence will collapse in the second page. Thus, we have the isomorphism
\begin{equation*}
	H^0_{tot}(C_{tot},d_{tot}) \cong _{I} E^{0,0}_2 \cong _{II} E^{0,0}_2~.
\end{equation*}
\end{proof}

\subsection{\label{subsec:mixed_at_genus_zero}Mixed vertex algebras at genus zero}

We first give a construction of the vertex algebras associated to spheres with only one pair of twisted punctures: $\CC_{0,m,1}$ before considering the more general case. As in the untwisted case, one could provide a recursive definition of $\V_{m,1}$ by repeatedly gluing untwisted caps. We elect, instead, to perform a simultaneous gluing and will show the equivalence between the two definitions later on. We define the vertex algebras $\V_{m,1}$ by
\begin{equation}
	\begin{split}
	\V_{0,1} &\coloneqq \D_t~,\\
	\V_{m,1} &\coloneqq H^0(C_{m,1},Q^m)~,
	\end{split}
\end{equation}
where $C_{m,1}$ is the chain complex
\begin{equation}
	C^{\bullet}_{m,1} = \W_u^{\otimes m} \otimes \D_t \otimes\bigg( \Li{\bullet}(\zf(\gf_u))\bigg)^{m} ~,
\end{equation}
with differential
\begin{equation}
	\begin{split}
	Q^{m}(z) &= \sum_{i=1}^{m}Q^{i,i+1}~,\\
	Q^{i,i+1}(z) &= \sum_{j=1}^{\mathrm{rk}\,\gf_u} : (\rho_i(P_j) -\rho_{i+1}(\tau(P_j)))\rho_{gh_{i}}(c^j):(z)~.
	\end{split}
\end{equation}
Here $\rho_i$, for $i\leqslant m$ denotes the action of $\zf(\gf_u)$ on the $i$th factor of $\W_u$ and $\rho_{m+1}$ denotes the action of $\zf(\gf_u)$ on $\D_t$ along the projection $\ZZ_u\twoheadrightarrow \ZZ_t$. 
The vertex algebras $\V_{m,1}$ live entirely in cohomological degree zero no fermions because we restrict to the zeroth cohomology---this is compatible with our expectation on the basis of residual gauge symmetries. We will later reinforce this by showing that the cohomology of $\V_{m,1}\circ_t\V_{n,1}$ is concentrated in degree zero.

The naive generalisation to $\V_{m,n}$ of our previous construction would be to take the zeroth cohomology of
\begin{equation}
	C^{\bullet}_{m,n} = \W^{m}_u\otimes \V_{G_t,2n} \otimes \bigg( \Li{\bullet}(\zf(\gf_u))\bigg)^{m}~.
\end{equation}
However, this cannot be quite right. Indeed, the vertex algebras associated to $\CC_{0,m,n}$ for $n>1$ should be supported outside of cohomological degree zero in order to express the presence of enhanced Higgs branches for the corresponding SCFTs. On the other hand, the vertex algebras, $\W_u$ and $\V_{G_t,2n}$, lie in degree zero, and the truncation to zeroth cohomology means this will persist. One should also expect to see a $\IZ_2$ symmetry exchanging positive and negative cohomological degree (a shadow of CPT in four dimensions). However, the untwisted caps $\W_u$ are projective over $\ZZ_{u,(<0)}$ and as a consequence the cohomology vanishes in negative degree. Therefore, even if we do not truncate to degree zero, the resulting vertex algebra would not have the right form. 

Thus we will have to define the vertex algebras $\V_{m,n}$, by going to some (non-canonically chosen) duality frame. We will choose to recursively consider the decomposition of $\CC_{0,m+1,n-1}$ as $\CC_{0,m+1,n-1}$ and $\CC_{0,1,1}$ connected by an untwisted cylinder. This gives us our definition:
\begin{equation}\label{eq:higher_puncture_mixed_def}
	\V_{m,n}\coloneqq \V_{1,1}\circ_u\V_{m+1,n-1}~.
\end{equation}
As defined, it is not clear whether the $\V_{m,n}$ are independent of our choice of duality frame in which to define them. For example, we could have also obtained this from a twisted gluing of $\V_{1,1}$ to $\V_{m-1,n}$. This is just the vertex algebra question of generalised $S$-duality, which is now not made manifest by our definition \eqref{eq:higher_puncture_mixed_def}. In the following sections, we shall work to establish how our definitions fit in with the duality web of class $\SS$.

\subsection{\label{ssec:gluing_mixed}Properties of the genus zero mixed vertex algebras}

First, we will check that our definition of the mixed vertex algebras with $n=1$ agrees with the recursive definition, $\V_{m,1}\overset{?}{\cong} \W_u\ast_u \V_{m-1,1}$. We have an extension of Lemma $10.1$ of Arakawa to the twisted case.
\begin{lem}\label{lem:mix_ff_cap}
For $m\geqslant 1$, we have that
\begin{enumerate}[(i)]
	\item $H^{n}(C_{m,1},Q^{m-1}_{(0)}) \cong 0\text{ for } n<0~,$
	\item $\W_u\ast_u \V_{m-1,1} \cong \V_{m,1}~.$
\end{enumerate}
\end{lem}
\begin{proof}
This proof is largely adapted from the proof of Lemma $10.1$ in~\cite{Arakawa:2018egx}.
We proceed by induction on $m$. For the base case, $m=1$, $(i)$ is true, since $\W_u$ is projective over $\ZZ_{u,(<0)}$. The second statement is true by definition, since $\V_{1,1}\coloneqq \W_u\ast_u\D_t$. Next, suppose $m>1$ and consider the bicomplex
\begin{equation*}
	C^{\bullet\bullet} = \W_u\otimes C_{m-1,1}^\bullet \otimes \Li{\bullet}(\zf(\gf_u))~,
\end{equation*}
with differentials $Q^{m-1}_{(0)}$ acting on $C_{m-1,1}$ and $d$ acting on $\W_u\otimes C_{m-1,1}\otimes \Li{\bullet}(\zf(\gf_u))$. The two differentials anticommute and the corresponding total complex is just $C_{m,1}$ with differential $Q^{m}_{(0)}$. There is a spectral sequence, with second page
\begin{equation*}
	E_{2}^{p,q} = \Hi{p}(\ZZ_u, \W_u\otimes H^{q}(C_{m-1,1},Q^{m-1}_{(0)}))~,
\end{equation*}
which converges to the total cohomology. By the inductive assumption, $H^{q}(C_{m,1},Q^{m-1}_{(0)})$ vanishes for $q<0$ and $\W_u\ast_u-$ is left exact. Therefore, $E_{2}^{p,q}=0$ for $p,q<0$, so $H^{n}(C_{m,1},Q^m_{(0)})$ vanishes for $n<0$. Moreover, the entry $E_{2}^{0,0}$ is stable and we have
\begin{equation}
	E_2^{0,0} = \W_u\ast_u \V_{m-1,1} \cong H^{0}(C_{m,1},Q^{m}_{(0)})\cong \V_{m,1}~,
\end{equation}
as desired.
\end{proof}
\vspace{-6pt}

Next, we examine the case of untwisted DS reduction.
\begin{lem}\label{lem:mix_kl0}
The mixed vertex algebras $\V_{m,n}$ are objects in $\KL_{u,0}$. In particular, for $m\geqslant 0$ and $n\geqslant 1$,
\begin{equation*}
	H^0_{DS}(u, \V_{m+1,n}) \cong \V_{m,n}~.
\end{equation*}
\end{lem}
\begin{proof}
We proceed by double induction on $m$ and $n$, first examining the base case of $n=1$ and arbitrary $m$. Note that for $m=0$, $H^0_{DS}(u,\V_{1,1})=\D_t$ from Theorem~\ref{thm:cyl_kl0}. Next, suppose $m>0$; for any object in $\KL_u$, $H^0_{DS}(u,-)$ and $\W_u\circ_u-$ are isomorphic. Therefore,
\begin{equation*}
	H^0_{DS} (u,\V_{m+1,1}) \cong \W_u\circ_u \V_{m+1,1}\cong \W_u \circ_u(\W_u\ast_u \V_{m,1})~,
\end{equation*}
where we have used Lemma~\ref{lem:mix_ff_cap}. Consider the bicomplex
\begin{equation*}
	C^{\bullet \bullet} =\W_u\otimes \V_{m,1}\otimes \W_u \otimes \Li{\bullet}(\zf(\gf_u))\otimes \Li{\bullet}(\hat{\gf}_{u,-\kappa_g})~,
\end{equation*}
with differential $d_\gf$ acting on $\W_u\otimes \V_{m,1} \otimes \Li{\bullet}(\hat{\gf}_{u,-\kappa_g})$ and $d_\zf$ acting on $\V_{m,1}\otimes \W_u \otimes \Li{\bullet}(\zf(\gf_u))$. The two differentials anticommute, so we form the total complex $C^n_{tot}=\bigoplus_{p+q=n} C^{pq}$, with total differential $d_{tot} = d_{\gf} + (-1)^q d_\zf$. There are two spectral sequences converging to the total cohomology, given by
\begin{equation}\nonumber
	\begin{split}
	_{I} E^{pq}_2 &= \Hi{p}(\ZZ_u,\W_u\otimes \Hi{q}(\hat{\gf}_{u,-\kappa_g},\gf_u, \V_{m,1}\otimes \W_u))~,\\
	_{II} E^{pq}_2 &= \Hi{p}(\hat{\gf}_{u,-\kappa_g},\gf_u ,\Hi{q}(\ZZ_u,\W_u\otimes \V_{m,1})\otimes W_u)~.
	\end{split}
\end{equation}
The cohomology, $\Hi{\bullet}(\hat{\gf}_{u,-\kappa_g},\gf_u,-\otimes \W_u)$ is concentrated in degree zero, and $\Hi{i}(\ZZ_u \W_u\otimes -)$ vanishes for $i<0$. Therefore, both spectral sequences collapse at the second page and we have
\begin{equation}\nonumber
	\W_u\ast_u(\W_u\circ_u \V_{m,1})=\ _{I} E^{00}_2\cong\ _{II} E^{00}_2 =\W_u\circ_u(\W_u\ast_u \V_{m,1})~.
\end{equation}
Therefore,
\begin{equation*}
	H^0_{DS}(u,\V_{m+1,1}) \cong \W_u\ast_u (\W_u\circ_u\V_{m,1})\cong \W_u \ast_u (H^0_{DS}(u,\V_{m,1})~,
\end{equation*}
but $\W_u\ast_u(H^0_{DS}(u,\V_{m,1}))\cong \V_{m,1}$ by Theorem 9.11 of~\cite{Arakawa:2018egx}.

Now suppose $n>1$. Then, we have that
\begin{equation*}
	H^0_{DS}(u,\V_{m+1,n}) \cong \W_u\circ_u(\V_{1,1}\circ_u \V_{m+2,n-1})~.
\end{equation*}
Since, $\W_u$ is semijective in $\KL_u$, Lemma~\ref{lem:rearr_gauge} applies and we have that
\begin{equation}\nonumber
	\begin{split}
	H^0_{DS}(u,\V_{m+1,n}) &\cong \V_{1,1}\circ_u(\W_u\circ_u\V_{m+2,n-1})~,\\
	&\cong \V_{1,1} \circ_u (H^0_{DS}(u,\V_{m+2,n-1}))~,\\
	&\cong \V_{1,1}\circ_u \V_{m+1,n-1}~,
	\end{split}
\end{equation}
where we have used the inductive assumption. Of course, $\V_{1,1}\circ_u \V_{m+1,n-1}$ is $\V_{m,n}$ by definition and we are done.
\end{proof}
Our uniqueness result for $\V_{1,1}$ can now be extended to all of the genus zero vertex algebras $\V_{m,n}$. To reiterate, $\V_{m,n}$ is the unique vertex algebra object in $\KL_u^{\otimes m}$ such that $H^0_{DS}(u, \V_{m-1,n})$.

Using Lemma~\ref{lem:mix_kl0}, we can present our version of Proposition $10.10$ of~\cite{Arakawa:2018egx}.
\begin{prop}\label{prop:mixed_free_cofree}
The vertex algebras $\V_{m,1}$ have an ascending filtration whose successive quotients are
\begin{equation*}
	\IV^u_{\lambda,m}\otimes_{\zf^u_{\lambda}}(\IV^t_{\lambda}\otimes_{\zf^t_\lambda}\IV^t_{\lambda*})~,
\end{equation*}
for some $\lambda\in P_t^+$ and $\iota(\lambda)\in P_u^+$. Additionally, the vertex algebras $\V_{m,1}$ have a descending filtration whose successive quotients are
\begin{equation*}
	D\big(\IV^u_{\lambda,m}\otimes_{\zf^u_{\lambda}}(\IV^t_{\lambda}\otimes_{\zf^t_\lambda}\IV^t_{\lambda*})\big)~.
\end{equation*}
Therefore, the vertex algebras $\V_{m,1}$ are semijective in $\KL_t$, so we have:
\begin{equation*}
	\V_{m,1}\circ_t V \cong \Hi{0}(\hat{\gf}_{t,\kappa_g},\gf,\V_{m,1}\otimes V)~.
\end{equation*}
\end{prop}
\begin{proof}
First, note that the objects $\IV_{\lambda,m}\otimes_{\zf^u_{\lambda}}(\IV^t_{\lambda}\otimes_{\zf^t_\lambda}\IV^t_{\lambda*})\big)$ are in $\KL_u,0$, by the same argument as Proposition~\ref{prop:kl_u}. We proceed via induction, noting that for the base case $m=1$ the statement is true by way of Proposition~\ref{prop:trin_free_cofree}.

Suppose the statement is true for some $m>1$,~\ie,~$\V_{m,1}$ has a filtration
\begin{equation*}
0=N_0\subset N_1\subset \dots \subset N=\bigcup_i N_i~,
\end{equation*}
with successive quotients $N_i/N_{i-1}\cong \IV_{\lambda,m}\otimes_{\zf^u_\lambda}((\IV^t_{\lambda}\otimes_{\zf^t_\lambda}\IV^t_{\lambda*})$ for some $\lambda\in P_t^+$ and $\iota(\lambda)\in P_u^+$.
we have
\begin{equation*}
	\V_{m+1,1} \cong \W_u\ast_u\V_{m,1}~,
\end{equation*}
by Lemma~\ref{lem:mix_ff_cap}. By Proposition $9.14$ of~\cite{Arakawa:2018egx}, $\V_{m+1,1}$ has an ascending filtration
\begin{equation*}
	0=M_0\subset M_1\subset\dots\subset M = \bigcup_i M_i ~,
\end{equation*}
with successive quotients $M_i/M_{i-1}\cong \W_u\ast_u N_i/N_{i-1}$. Since $N_i/N_{i-1}$ are objects in $\KL_{u,0}$, we have that
\begin{equation*}
	M_i/M_{i-1} \cong \IV_{\lambda,m+1}\otimes_{\zf^u_\lambda}((\IV^t_{\lambda}\otimes_{\zf^t_\lambda}\IV^t_{\lambda*})~,
\end{equation*}
for some $\lambda\in P_t^+$ and $\iota(\lambda)\in P_u^+$.
\end{proof}
\vspace{-6pt}

One could hope to strengthen this to the statement that all $\V_{m,n}$ are semijective in $\KL_t$ but such a result is beyond what is easily realised by our technologies. With the current arguments, we would have to independently establish that the vertex algebras $\V_{0,n}$ are semijective in $\KL_{u,0}$. Given that our definition of $\V_{0,n}$ involves an unbounded cohomology, it seems difficult to verify such a property.

From a similar argument as with $\V_{1,1}$, the mixed vertex algebra, $\V_{m,1}$ cannot be semijective in $\KL_u$, so the derived functor $\V_{1,1}\circ_u-$, which increases the number of pairs of twisted punctures by one, is not concentrated in degree zero. By construction, the $\V_{m,n}$ cannot be concentrated in cohomological degree zero for $n>1$. This is, in one sense, a good thing---going back to our discussion of residual gauge symmetry, these vertex algebras should have fermionic operators lying in non-zero cohomological degree. On the other hand, this is a rather large roadblock to our spectral sequence powered proofs of associativity. We will only be able to provide partial results for how $\V_{m,n}$ fit in the duality web.

Similar to the mixed trinion, we can derive the characters and central charges of the mixed vertex algebras.
\begin{prop}\label{prop:mixed_char}
The character of the vertex algebra $\V_{m,1}$ is given by
\begin{equation*}
	\Tr_{\V_{m,1}}\bigg(q^{L_0}\prod_{i=1}^{m} \mathbf{a}_i \,\,\mathbf{b}_1\mathbf{b}_2\bigg) = \sum_{\lambda\in P_t^+} \frac{\prod_{i=1}^m \KK_u(\mathbf{a}_i)\chi^\lambda_u(\mathbf{a}_i) \KK_t(\mathbf{b}_1)\chi^\lambda_t(\mathbf{b}_1) \KK_t(\mathbf{b}_2)\chi^\lambda_t(\mathbf{b}_2) }{\big( \KK_u(\times)\chi^\lambda_u(\times) \big)^{m}}~.
\end{equation*}
so the vertex algebra, $\V_{m,1}$ is conical for all $m\in \IN$.
\end{prop}
\begin{proof}
For $m=1$ this is just Proposition~\ref{prop:mixed_trin_char}, so we take $m>1$. We have that $H^0_{DS}(u,\V_{m,1})\cong \V_{m-1,1}$. As graded vector spaces, we have the decomposition
\begin{equation*}
	\V_{m-1,1} \cong \sum_{\lambda\in P_t^+} \IV^u_{\lambda,m-1} \otimes_{\zf_\lambda^u}(\IV^t_{\lambda}\otimes_{\zf_\lambda^t}\IV^t_{\lambda*})~.
\end{equation*}
Applying Proposition~8.4 of~\cite{Arakawa:2018egx} gives us the desired result. The fact that $\V_{m,1}$ are conical follows from the same argument as Proposition~\ref{prop:mixed_trin_char}.
\end{proof}

The vertex algebra $\V_{m,n}$ is constructed by repeatedly gauging $\V_{m+n-1,1}$ with copies of $\V_{1,1}$. We know how the character behaves under gauging, so this result can be extended to
\begin{equation}\label{eq:mixed_char}
	\Tr_{\V_{m,n}}\bigg(q^{L_0}\prod_{i=1}^{m} \mathbf{a}_i \,\prod_{j=1}^{2n}\mathbf{b}_j\bigg)=\frac{\prod_{i=1}^m \KK_u(\mathbf{a}_i)\chi^\lambda_u(\mathbf{a}_i) \prod_{j=1}^{2n} \KK_t(\mathbf{b}_j)\chi^\lambda_t(\mathbf{b}_j)} {\big( \KK_u(\times)\chi^\lambda_u(\times) \big)^{m+2n-2}}~,
\end{equation}
which agrees with the expression in~\cite{Lemos:2012ph}.

\begin{prop}
The vertex algebras $\V_{m,1}$ are conformal with central charge
\begin{equation*}
	c_{\V_{m,1}} = m\,\dim \gf_u + 2\dim \gf_t -24m\,\rho_u\cdot\rho_u^\vee -m\, \rk\,\gf_u~.
\end{equation*}
\end{prop}
\begin{proof}
From Proposition~\ref{prop:mixed_trin_conf}, we know that $\V_{m,1}$ is conformal with central charge $c_{\V_{1,1}} = 2\dim\, \gf_t + \dim\,\gf_u -\rk\,\gf_u -24 \rho_u\cdot\rho_u^\vee$ so the statement is true of $m=1$.

We shall first show that $\V_{m,1}$ has a conformal vector, then show that it is the unique conformal vector whose grading agrees with the character and finally show that this results in the correct central charge.

Consider the vector $\omega_m\in C_{m,1}$, defined by
\begin{equation*}
	\omega_m = \omega_{\D_t} + \sum_{i=1}^{m} \omega_{\W_i}+\sum_{i=1}^m \omega_{gh,i} ~,
\end{equation*}
where $\omega_{\D_t}$ is the conformal vector of $\D_t$, $\omega_{\W_i}$ is the conformal vector of the $i$th factor of $\W_u$ and $\omega_{gh,i}$ is the conformal vector of the $i$th ghost system, $\Li{\bullet}(\zf(\gf_u))$. Clearly, $\omega_m$ defines a conformal vector on the complex, $C_{m,1}$. Like we did for $\V_{1,1}$ we shall argue that this descends to a conformal vector in cohomology.

For an element $P_i\in\zf(\gf_u)$, where we think of $\zf(\gf_u)$ as a subalgebra of one of the $\W_u$ factors, $\omega_m$ acts as
\begin{equation*}
	\omega_m(z) P_i(w) \sim \frac{\partial P_i}{z-w} +\frac{(d_i+1) P_i}{(z-w)^2}+\sum_{j=2}^{d_i+2}\frac{(-1)^j j!}{(z-w)^{j+1}} q_j^{(i)}(w)~,
\end{equation*}
where $q_j^{(i)}$ is some homogeneous state in $\zf(\gf_u)$ with weight $d_i-j+2$. Let us denote by $\widetilde{P_i}$ the image of $P_i$ under the projection $\zf(\gf_u)\twoheadrightarrow \zf(\gf_t)$.
One then has
\begin{equation*}
	\omega(z) \widetilde{P_i}(w) \sim \frac{\partial \widetilde{P_i}}{z-w} +\frac{(d_i+1) \widetilde{P_i}}{(z-w)^2}+\sum_{j=2}^{d_i+2}\frac{(-1)^j j!}{(z-w)^{j+1}} \widetilde{q_j}^{(i)}(w)~,
\end{equation*}
where we think of $\widetilde{P_i}$ as a state in $\zf(\gf_t)\subset\D_t$. Let $Q^m$ be the differential of $C_{m,1}$; the action of $Q^m$ on the vector $\omega_m$ is
\begin{equation*}
	Q^m_{(0)}(z) \omega(w) = \sum_{l=1}^{m}\sum_{i=1}^{\mathrm{rk}\,\gf_u}\sum_{j=2}^{d_i+1} \partial^j \big(\rho_l(q_j^{(i)}) -\rho_{l+1}(\tau(q^{(i)}))\big)c_i ~,
\end{equation*}
where $\rho_i$ for $i\leqslant m$ denotes the action of the Feigin-Frenkel centre on the $i$th $\W_u$ factor and $\rho_{m+1}$ once again denotes the action of $\zf(\gf_u)$ on $\D_t$ along the projection to $\zf(\gf_t)$.

If the right hand side of the above equation equals $Q^m_{(0)}\chi$ for some state $\chi$, then $\widetilde{\omega}_m = \omega_m + \chi$ is $Q$-closed and defines a vector in $\V_{m,1}$.
For $l\neq m$ it is clear that $\rho_l(q_j^{(i)})-\rho_{l+1}(q_j^{(i)})$ is a coboundary, and we have addressed the $l=m$ case in the proof of Proposition~\ref{prop:mixed_trin_conf}. Therefore, such a $\chi$ exists and may be written as
\begin{equation*}
	\chi = \sum_{l=1}^m\sum_{i=1}^{\mathrm{rk}\, \gf_u}\sum_{j=2}^{d_i+2} \partial^j (\rho_{l}\otimes \rho_{l+1}\otimes \rho_{gh,l})(z_{ij})
\end{equation*}
for some $z_{ij} \in \zf(\gf_u)\otimes\zf(\gf_u)\otimes \Li{0}(\zf(\gf_u))$. Therefore, $\widetilde{\omega}_{m,(i)}=\omega_{m,(i)}$ for $i=0,1$, so the OPEs agree up to the quadratic pole. Since $\V_{m,1}$ is conical by Proposition~\ref{prop:mixed_char}, Lemma $3.1.2$ of~\cite{Frenkel:2007} applies once more and we can conclude that $\widetilde{\omega}_m$ is a conformal vector in $\V_{m,1}$.

Now we wish to show that $\widetilde{\omega}_m$ is the unique conformal vector whose $L_0$-grading agrees with $\omega_m$. The argument from Proposition~\ref{prop:mixed_trin_conf} using Lemma 4.1 of~\cite{Moriwaki_2020} still works, with minor alteration, since $\V_{m,1}$ are conical.

The DS reduction of $\widetilde{\omega}_m$ gives a conformal vector in $\V_{m-1,1}$ with central charge
\begin{equation*}
	c_{\V_{m-1,1}} = c_{\V_{1,1}} +\rk\,\gf_u - \dim \gf_{u} +24 \rho_{u}\cdot\rho^\vee_{u}~,
\end{equation*}
and which agrees with the grading by $\omega_{m-1}$. But, by the inductive assumption, such a conformal vector on $\V_{m-1,1}$ is unique.

Unwrapping the induction from the base case of $\V_{0,1}=\D_t$, we get that
\begin{equation*}
c_{\V_{m,1}} = m\,\dim \gf_u + 2\dim \gf_t -24m\,\rho_u\cdot\rho_u^\vee -m\, \rk\,\gf_u~.
\end{equation*}
\end{proof}
\vspace{-6pt}

To extend this to $\V_{m,n}$, we shall once more make use of the fact the $\V_{m,n}$ are constructed by repeated twisted gaugings of the $\V_{m,1}$ with copies of $\V_{1,1}$. It is well known that the conformal vector
\begin{equation}
	T = \widetilde{\omega}_{\V_{k,l}} + \widetilde{\omega}_{\V_{1,1}} +\omega_{gh} ~,
\end{equation}
where $\omega_{gh}$ is the conformal vector of the $b,c$ ghost system, descends to a conformal vector in the BRST cohomology with central charge equal to the central charge of $T$. Therefore the vertex algebras $\V_{m,n}$ are conformal, with central charge
\begin{equation}
	c_{\V_{m,n}} = 2n\, \dim\,\gf_t + m\, \dim\,\gf_u - (m+2n-2)\rk\,\gf_u -24(m+2n-2)\,\rho_u\cdot\rho_u^\vee~.
\end{equation}

Having established a number of intrinsic properties of the genus zero, mixed vertex algebras, let us examine how they interact with each other under gluing. This will shed some light on how the $\V_{m,n}$ fit into the class $\SS$ duality web.

Given our rearrangement lemmas, we can show that the mixed vertex algebras of the previous section have the expected behaviour under $\circ_u$ and $\circ_t$. First, we establish the partial result

\begin{lem}\label{lem:gauge_mix_trin}
We have the isomorphism
\begin{equation*}
	\V_{1,1}\circ_t \V_{m,n} \cong \V_{m+1,n}~.
\end{equation*}
\end{lem}
\begin{proof}
First, let us treat the base case of $n=1$. We have that
\begin{equation*}
	\V_{1,1}\circ_t \V_{m,1} \cong (\W_u \ast_u \D_t)\circ_t \V_{m,1}~,
\end{equation*}
so $\V_{1,1}\circ_t \V_{m,1} \cong \W_u\ast_u \V_{m,1} \cong \V_{m+1,1}$. By Lemma~\ref{lem:ut-cap-gauge}, we have that
\begin{equation*}
	(\W_u \ast_u \D_t)\circ_t \V_{m,1}\cong \W_u\ast_u(\D_t \circ_t \V_{m,1}) \cong \W_u \ast_u \V_{m,1}\cong \V_{m+1,1}~,
\end{equation*}
where we have used Lemma~\ref{lem:mix_ff_cap}. Now, we proceed by induction on $n$. We have just established the base case for $n=1$, so suppose $n>1$. Then,
\begin{equation}\nonumber
	\begin{split}
	\V_{1,1}\circ_t \V_{m,n} &\cong \V_{1,1}\circ_t (\V_{m+1,n-1}\circ_u \V_{1,1})~, \\
	&\cong (\V_{1,1}\circ_t \V_{m+1,n-1})\circ_u \V_{1,1} ~,\\
	&\cong \V_{m+2,n-1}\circ_u \V_{1,1}~, \\
	&\cong \V_{m+1,n}~,
	\end{split}
\end{equation}
where in the second line we have used Lemma~\ref{lem:gauge_ut}.
\end{proof}
\vspace{-6pt}

\begin{prop}\label{prop:gauge_mix}
Under gauging, the vertex algebras $\V_{m,1}$ behave as expected, namely,
\begin{equation}\nonumber
	\begin{split}
	\V_{m,1}\circ_t \V_{p,q}&\cong \V_{m+p,q}~,\\
	\V_{m,1}\circ_u\V_{p,q}&\cong \V_{p+m-2,q+1}~.
	\end{split}
\end{equation}
\end{prop}
\begin{proof}
We proceed via induction for each type of gluing, noting that the base case $m=1$ is true, either by Lemma~\ref{lem:gauge_mix_trin} or by definition.

Suppose $m>1$. Then,
\begin{equation}\nonumber
	\begin{split}
	\V_{m,1}\circ_t\V_{p,q} &\cong (\V_{1,1}\circ_t\V_{m-1,1})\circ_t\V_{p,q}~,\\
	&\cong \V_{1,1}\circ_t(\V_{m-1,1}\circ_t\V_{p,q})~,\\
	&\cong \V_{p+m,q}~,
	\end{split}
\end{equation}
where in the second line, we have used Lemma~\ref{lem:rearr_gauge} to arrive at the desired result.

Next we treat the $\circ_u$ case. Again, suppose $m>1$. Then,
\begin{equation}\nonumber
	\begin{split}
	\V_{m,1}\circ_u\V_{p,q} &\cong (\V_{1,1}\circ_t\V_{m-1,1})\circ_u\V_{p,q}~,\\
	&\cong \V_{1,1}\circ_t(\V_{m-1,1}\circ_u\V_{p,q})~,\\
	&\cong \V_{p+m-2,q+1}~,
	\end{split}
\end{equation}
where, in the second line we have used a slight modification of Lemma~\ref{lem:gauge_ut}---which applies, since the cohomology $\V_{1,1}\circ_t-$ is concentrated in degree zero, so the spectral sequences will collapse at the second page.
\end{proof}

Of course, we expect that these results should extend to the general case,
\begin{equation}
	\V_{m,n}\circ_t \V_{p,q} \cong \V_{m+p,n+q-1}~, \quad \V_{m,n}\circ_u \V_{p,q} \cong \V_{m+p-2,q+n}~.
\end{equation}
The obstructions to proving this are as follows. In the case of $\circ_t$, the inductive step is $\V_{m,n}\circ_t(\V_{p+1,q-1}\circ_u \V_{1,1})$ and the corresponding spectral sequence is unbounded and does not collapse at the second page. Similarly, for the untwisted gluing we have not been able to establish the putative isomorphism
\begin{equation}
	(\V_{1,1}\circ_u\V_{m+1,n-1})\circ_u\V_{p,q}\cong \V_{1,1}\circ_u(\V_{m+1,n-1})\circ_u\V_{p,q}~,
\end{equation}
for the inductive step. Neither cohomology is concentrated in degree zero, so the second page of the spectral sequence is unbounded. To make progress seems to require more sophisticated machinery or a different strategy.

\begin{prop}~\label{prop:ff-untwisted-mixed}
We have the isomorphism
\begin{equation*}
	\V_{G_u,s}\ast_u \V_{m,n} \cong \V_{m+s,n}~.
\end{equation*}
\end{prop}
\begin{proof}
We proceed by induction on $s$, noting that, for $s=1$, the statement is true since $V_{m,n}$ are in $\KL_{u,0}$. Now suppose $s>1$, and consider $H^0_{DS}(u,\V_{G_u,s}\ast_u \V_{m,n})$. Of course, $H^0_{DS}(u,-)$ and $\W_u\circ_u -$ are isomorphic. We form the bicomplex
\begin{equation*}
	C^{\bullet, \bullet} = \W_{u}\otimes \V_{G_u,s} \otimes \V_{m,n}\otimes \Li{\bullet}(\zf(\gf_u)))\otimes \Li{\bullet}(\gf_u)~,
\end{equation*}
with differentials $d_\zf$ acting on $\V_{G_u,s} \otimes \V_{m,n}\otimes \Li{\bullet}(\zf(\gf_u))$ and $d_{\gf}$ acting on $\W_u\otimes\V_{G_u,s}\otimes \Li{\bullet}(\gf_u)$. The differentials anticommute and we can form the total complex as usual. The two relevant spectral sequences are
\begin{equation}\nonumber
	\begin{split}
	_{I} E_{2}^{p,q} = \Hi{p}(\hat{\gf}_{u,-\kappa_g},\gf_u, \W_u \otimes \Hi{q}(\ZZ_u,\V_{G_u,s}\otimes \V_{m,n}))~,\\
	_{II} E_{2}^{p,q} = \Hi{p}(\ZZ_u,\Hi{q}(\hat{\gf}_{u,-\kappa_g},\gf_u, \W_u\otimes \V_{G_u,s})\otimes \V_{m,n})~.
	\end{split}
\end{equation}
The functor of DS-reduction is exact, so both sequences collapse at page two with $_{I} E_{2}^{0,q}$ and $_{II}E_{2}^{p,0}$ being the only non-zero entries. This gives the isomorphism,
\begin{equation*}
H^0_{DS}(u,\V_{G_u,s}\ast_u\V_{m,n}) \cong \V_{G_u,s-1}\ast_u\V_{m,n} \cong \V_{m+s-1,n}~,
\end{equation*}
where we used the inductive hypothesis. Acting by $\W_u\ast_u-$, we have that
\begin{equation}
	\V_{G_u,s} \ast_u \V_{m,n} \cong \V_{m+s,n}~,
\end{equation}
as desired.
\end{proof}
\vspace{-6pt}

Finally, let us consider gauging the untwisted and mixed vertex algebras together.
\begin{prop}
Under untwisted gauging of untwisted vertex algebras, the mixed vertex algebras behave as expected,~\ie,
\begin{equation*}
\V_{G_u,s} \circ_u \V_{m,n} \cong \V_{m+s-2,n}
\end{equation*}
\end{prop}
\begin{proof}
We have the following chain of isomorphisms,
\begin{equation}\nonumber
	\begin{split}
	\V_{G_u,s}\circ_u\V_{m,n} &\cong (\V_{G_u,s-1}\ast_u \W_u)\circ_u \V_{m,n}~,\\
	&\cong \V_{G_u,s-1} \ast_u (\W_{u}\circ_u \V_{m,n})~,\\
	&\cong \V_{G_u,s-1} \ast_u \V_{m-1,n}~,\\
	&\cong \V_{m+s-2,n}~,
	\end{split}
\end{equation}
where in the second line we have used Lemma~\ref{lem:rearr_gauge_ff} and in the third we have used Proposition~\ref{prop:ff-untwisted-mixed}.
\end{proof}

To conclude this subsection, we should comment on the general issue of associativity. In~\cite{Arakawa:2018egx}, the cohomology $\V_{G,s}\circ\V_{G,s'}$ was concentrated in degree zero, so ``gauging is associative'' as a consequence of a by-now-standard spectral sequence argument. In our case, the argument is not so simple---we have repeatedly remarked that zero genus is no longer sufficient for a gluing to be concentrated in degree zero. The obvious spectral sequence no longer collapses on the second page, so the proofs of rearrangement lemmas no longer hold. Nevertheless, it is our belief that associativity must still hold in general. Yanagida~\cite{Yanagida:2020kim} has defined a derived version of the construction of~\cite{Arakawa:2018egx} in the sense of derived algebraic geometry. In addition, that work also imported the Moore--Tachikawa TQFT to the derived setting. In the derived Moore--Tachikawa analysis, associativity of gauging follows from general properties of the derived pushforward---even for nonzero genus. However, the notion of associativity in the derived setting is not quite what would be expected by physicists. In short, the physical prescription is normally to pass cohomology before the second gauging---unlike the derived case, where one does not pass to cohomology but instead holds on to the full data of the chain complex (as an object in an appropriate derived category). A sufficient condition for the derived associativity to imply our version of associativity is to show that the relevant chain complexes are formal,~\ie,~are isomorphic to their cohomology (thought of as a complex with zero differential) in that derived category. While this is an interesting problem, we do not address it in this work.

\subsection{\label{ssec:duality}Generalised \texorpdfstring{$S$}{S}-duality and 4-moves}

We will restrict ourselves to the sub-web of generalised $S$-dualities that can be reached by repeated application of the various $4$-moves. Assuming gauging is associative, this recovers a large swathe of the $S$-duality landscape. For the untwisted, genus zero vertex algebras, the operation of gauging was shown to be associative in~\cite{Arakawa:2018egx}.

Let $\V$ be some mixed vertex algebra, for now we assume genus zero with only maximal punctures. There are three different types of $4$-move that can act on $\V$. The first acts on a collection of four untwisted punctures; the second acts on two pairs of twisted punctures and the third acts on a pair of twisted punctures and two untwisted punctures. We examine each case in turn to establish invariance.

In the purely untwisted case, invariance under the $4$-move is baked into Arakawa's construction, which is manifestly permutation symmetric in the moment maps. Let us present an alternate formulation, which is generalisable to the twisted case. We may endow $V_{G,s}$ with the action\footnote{As written, the action of the $4$-move is not an automorphism but an isomorphism to some vertex algebra object in $\KL$. To correct this to an automorphism, we appeal to the uniqueness property of $\V_{G,s}$ (see Remark 10.13 of~\cite{Arakawa:2018egx}) and fix an isomorphism from this vertex algebra back to $\V_{G,s}$. One can perform a similar trick with $\V_{m,n}$ using the uniqueness statements of the previous section. To avoid having to make such a choice one can work with the Moore--Seiberg groupoid~\cite{Moore:1104178762}, but we will not do so here.} of a permutation group as follows. First, we close off the punctures labelled $2$ and $3$ in sequence,~\ie,~we perform DS reduction with respect to the moment maps $\iota_2$, associated to puncture $2$, and then with respect to $\iota_3$, associated to puncture $3$. We can invert this procedure by Feigin--Frenkel gluing a cap with moment map $\iota_3$ and then Feigin--Frenkel gluing a cap with moment map $\iota_2$. Instead, we reverse the order of inversions, that is to say we first glue a cap with moment map $\iota_2$ and then a cap with moment map $\iota_3$. This is an isomorphism since
\begin{equation}
    \W_{G,3}\ast_u(\W_{G,2}\ast_u(\W_{G,3}\circ_u(\W_{G,2}\circ_u \V))) \cong \W_{G,2} \ast_u (\W_{G,3}\ast_u   (\W_{G,3}\circ_u(\W_{G,2}\circ_u \V)))\cong \V~,
\end{equation}
where we have used Lemma~\ref{lem:rearr_ff} to swap the order of Feigin--Frenkel gluings. The subscripts on the caps keep track of the labelling of the moment maps. One can realise the actions of the other transpositions in the same way, and thus generate the action of the full symmetric group on $s$ punctures. The action of the symmetric group should be understood as swapping the labellings of the moment maps associated to each puncture. This argument also establishes invariance under the four move of the first type for the mixed vertex algebras $\V_{m,n}$.

Now let us consider the case $\V=\V_{0,2}$, the mixed vertex algebra with two pairs of twisted punctures and no untwisted punctures. We can define the action of transpositions, as in the untwisted case, by closing pairs of punctures and gluing caps. Again, let us pick two twisted punctures, labelled $2$ and $3$, with moment maps $j_2$ and $j_3$ respectively. We perform $DS$ reduction once more, closing the punctures labelled $2$ and then $3$, in order. Once again, we restore the punctures by Feigin--Frenkel gluing, via the twisted Feigin--Frenkel centre, two twisted caps with moment maps $j_2$ and $j_3$. We have
\begin{equation}
    \W_{t,2}\ast_t (\W_{t,3}\ast_2(\W_{t,3}\circ_t(\W_{t,2}\circ_t \V))) \cong \W_{t,3}\ast_2(\W_{t,3}\ast_2(\W_{t,3}\circ_t(\W_{t,2}\circ_t \V)))\cong \V~.
\end{equation}
Once more, by Lemma~\ref{lem:rearr_ff}, this results in a vertex algebra that is isomorphic to $\V_{0,2}$. All transpositions of twisted punctures can be arrived at using this method and we can generate the full symmetric group on $2n$ twisted punctures. It should be noted that the automorphism group allows swaps between twisted punctures regardless of whether one has connected them by twist lines (the twist lines in this abelain setting are a fiction anyways; they just record monodromies of the punctures). The automorphisms on $\V_{0,2}$ is thus the full $S_4$.

For the $4$-move of the third kind, our analysis in terms of permutations fails. This move swaps between the degeneration limits shown in Figure~\ref{fig:ut-move}---unlike the other cases the decompositions are no longer related by a simple permutation on the punctures. Instead, we appeal to Proposition~\ref{prop:ff-untwisted-mixed}, which states that the two BRST gluings $\V_{m-1,n}\circ_u\V_{G_u,3}$ and $\V_{1,1}\circ_t\V_{m-1,n}$ are isomorphic.

These $4$-moves, in addition to the assumption that $\circ_t$ and $\circ_u$ are always associative, recovers all of the landscape of generalised $S$-dualities that don't require the $ab$-moves (which appear at higher genus). A proof of invariance under the $ab$-move eludes us but we pose this as a conjecture in the language of semi-infinite cohomology.

\begin{conj}
Let $\V_{G_u,3}$ be the trinion vertex algebra with maximal untwisted punctures and $\V_{1,1}$ be the mixed trinion vertex algebra as before. Let $i_1,i_2,i_3$ denote the three actions of $V^{\kappa_c}(\gf_u)$ on $\V_{G_u,3}$. Similarly, let $j_2,j_3$ be the actions of $V^{\kappa_c}(\gf_t)$ on $\V_{1,1}$. Then the following isomorphisms hold,
\begin{equation*}
\Hi{\bullet}(\hat{g}_{u,-\kappa_g},\gf_u,\V_{G_u,3})\cong \Hi{\bullet}(\hat{\gf}_{t,-\kappa_g},\gf_t,\V_{1,1})~,
\end{equation*}
where $\hat{\gf}_{u,-\kappa_g}$ acts on $\V_{G_u,3}$ via $i_2\otimes (i_3\circ \sigma)$, with $\sigma$ the $\IZ_2$ outer-automorphism, and $\hat{\gf}_{t,-\kappa_g}$ acts on $\V_{1,1}$ via $j_2\otimes j_3$.
\end{conj}
Establishing invariance under the $ab$-move for the one punctured torus is sufficient to ensure invariance for the vertex algebras of all other surfaces. This construction is only relevant at nonzero genus, so states of higher cohomological degrees will be present. Of course, replacing the maximal puncture with a minimal puncture in this duality just recovers $S$-duality for non-simply-laced $\NN=4$ super Yang-Mills theory; this also remains an open conjecture at the level of associated vertex algebras.

While we are unable to show that the vertex algebras themselves are invariant under the $ab$ move, it is straightforward to observe that the Schur index is. Clearly from the form of \eqref{eq:indfull}, the index is invariant under any permutation of identical punctures and therefore is invariant under the $4$-move. To establish invariance under the $ab$-move, all we have to show is that the indices of the one punctured tori agree. The following is similar to an argument in~\cite{Agarwal:2013uga}. The index of an untwisted trinion with maximal punctures is
\begin{equation}
    \mathcal{I}(\mathbf{a},\mathbf{y},\mathbf{z}) = \sum_{\lambda\in P_{u}^+} \frac{\mathcal{K}_u(\mathbf{a})\chi^\lambda_u(\mathbf{a})\mathcal{K}_u(\mathbf{y})\chi^\lambda_u(\mathbf{y})\mathcal{K}_u(\mathbf{z})\chi^\lambda_u(\mathbf{z})}{\mathcal{K}_u(\times)\chi^\lambda_u(\times)}~.
\end{equation}
The gauging procedure is then
\begin{equation}
    \mathcal{I}_a(\mathbf{a}) = \oint [\mathrm{d}\mathbf{y}][\mathrm{d}\mathbf{z}]\Delta(\mathbf{y}) \mathcal{I}(\mathbf{a},\mathbf{y},\mathbf{z}) \mathcal{I}^V(\mathbf{y}) \delta(\mathbf{y},\sigma(\mathbf{z}^{-1}))~.
\end{equation}
Notice, that the usual propagator has been modified to $\delta(\mathbf{y}_1,\sigma(\mathbf{x}_2^{-1}))$ to account for the fact that the diagonal action has been twisted. Expanding the above expression, we have
\begin{equation}
\begin{split}
    \mathcal{I}_a(\mathbf{a}) &= \oint[\mathrm{d}\mathbf{z}]\Delta(\mathbf{z}) \mathcal{I}^V(\mathbf{z})\delta(\mathbf{y},\sigma(\mathbf{z}^{-1})) \sum_{\lambda\in P_u^+} \frac{ \mathcal{K}_{u}(\mathbf{y})\chi^{\lambda}_{u}(\mathbf{y})\mathcal{K}_{u}(\mathbf{z})\chi^{\lambda}_{u}(\mathbf{z})}{\mathcal{K}_{u}(\times)\chi^{\lambda}_{u}(\times)}~,\\
    &= \sum_{\lambda\in P_u^+}\frac{\oint[\mathrm{d}\mathbf{z}]\Delta(\mathbf{z}) \mathcal{I}^V(\mathbf{z}) \mathcal{K}_{u}(\mathbf{z})\mathcal{K}_{u}(\sigma(\mathbf{z})^{-1})\chi^{\lambda}_{u}(o(\mathbf{z}^{-1}))\chi^{\lambda}_{u}(\mathbf{z})}{\mathcal{K}_{u}(\times)\chi^{\lambda}_{u}(\times)}~,\\
    &= \sum_{\lambda\in P_u^+}\frac{\oint[\mathrm{d}\mathbf{z}]\Delta(\mathbf{z}) \mathcal{I}^V(\mathbf{z}) \mathcal{K}_{u}(\mathbf{z})\mathcal{K}_{u}(\mathbf{z}^{-1})\chi^{\sigma(\lambda)}_{u}(\mathbf{z}^{-1}))\chi^{\lambda}_{u}(\mathbf{z})}{\mathcal{K}_{u}(\times)\chi^{\lambda}_{u}(\times)}~,
    \end{split}
\end{equation}
where in the third line, we have used the transformation property of the group characters to move the outer automorphism to the Dynkin label. The orthonormality property of characters restricts the sum to only those weights $\lambda\in P_u^+$ that are invariant under $\sigma$. This is equivalent to summing over $P_t^+$, so by our usual abuse of notation, we have
\begin{equation}
    \mathcal{I}_a(\mathbf{a}) = \sum_{\lambda\in P_t^+}\frac{\mathcal{K}_u(\mathbf{a})\chi^\lambda_u(\mathbf{a})}{\mathcal{K}_u(\times)\chi^\lambda_u(\times)}~.
\end{equation}
The index of $\CC_{0,1,1}$ is
\begin{equation}
    \II_{\CC_{0,1,1}}(\mathbf{a},\mathbf{b}_1,\mathbf{b}_2) = \sum_{\lambda\in P_t^+} \frac{\KK_u(\mathbf{a})\chi_u^\lambda(\mathbf{a})\KK_t(\mathbf{b}_1)\chi_t^\lambda(\mathbf{b}_1) \KK_t(\mathbf{b}_2)\chi_t^\lambda(\mathbf{b}_2)}{\KK_u(\times)\chi_u^\lambda(\times)}.
\end{equation}
Gauging the diagonal $\gf_t$ action on $\CC_{0,1,1}$ results in the character:
\begin{equation}
    \II_b(\mathbf{a}) = \oint[\mathrm{d}\mathbf{b}_1] \II_{\CC_{0,1,1}}(\mathbf{a},\mathbf{b}_1,\mathbf{b}_2)\II^V(\mathbf{b}_1)\delta(\mathbf{b}_2,\mathbf{b}_1^{-1})~,
\end{equation}
where the contour integral is over the maximal torus of the twisted algebra. It is then immediate to confirm that $\II_a(\mathbf{a})=\II_b(\mathbf{a})$, as expected.

The machinery of Feigin-Frenkel gluing has been invaluable to our proof of invariance under the $4$-move. Yet, in some sense it is an unnatural construction in that we must truncate to the zeroth cohomology by hand. The natural construction would involve some functor, which is exact in the genus-zero cases of Arakawa. In the twisted case, the higher derived functors would hopefully capture the fermions that arise from residual gauge symmetry in the genus-zero case. Unfortunately, we do not currently have a good candidate for such a functor.

\subsection{\label{ssec:z3-twist}The \texorpdfstring{$\IZ_3$}{Z3} twist of \texorpdfstring{$\df_4$}{d4}}

Much of the preceeding discussion goes through for case of $\IZ_3$-twisted punctures in the $\df_4$ theory, but there are some new features that are worth mentioning. First, let us lay out the details of the construction.

The chiral differential operators over $G_2$ are well defined. Here we take $G_2$ to be the exponentiated form of the $\gf_2$ Lie algebra,~\ie,~$G_2$ is a simply connected, semi-simple, algebraic Lie group. The superconformal index assigned to the twisted cylinder agrees (summand by summand) with the character of $\V_{G_2,2}$, and this once more motivates our construction.

Pictorially, one imagines the $G_2$ cylinder as having a puncture twisted by $\omega$ and the other by $\omega^2$. One might, \textit{a priori}, expect that there are two possible $\gf_2$ caps, $\W_\omega$ and $\W_{\omega^2}$, depending on which puncture is closed. Yet, from the uniqueness argument of~\cite{Arakawa:2018egx}, the two putative caps are isomorphic. The outer automorphism, $\sigma$, that exchanges $\omega$ with $\omega^2$ should therefore lift to an isomorphism of vertex-algebras $\W_{\omega}\xrightarrow{\sigma}\W_{\omega^2}$.

The total monodromy around all punctures must be trivial. This can be satisfied in a number of ways, but for now we will restirct our attention to the case where punctures labelled by $\omega$ and $\omega^2$, respectively, come in pairs (one might think of them as having twist lines connecting them pairwise). We denote a genus zero surface with $m$ untwisted punctures and $n$ $\omega,\omega^2$ pairs of punctures by $\CC_{0,m,n}$. We define the mixed trinion, $\V_{1,1}$ as
\begin{equation}
    \V_{1,1} \coloneqq \W_{\df_4}\ast_u\D_{G_2}~,
\end{equation}
and the construction of the $\V_{m,n}$ proceeds analogously. The ambiguity in the two versions of the cap is present here again. Namely, is there an $S$-duality move that swaps two punctures with $\omega$ and $\omega^2$ labels? We will show that the vertex algebras $\V_{m,n}$ are indeed invariant under such a move, though there is no expectation that the underslying SCFTs will enjoy the same symmetry.

Let us first establish some rearrangement lemmas for the $\IZ_3$ case. The isomorphism of vertex algebras $\W_\omega\cong\W_{\omega^2}$ lifts to a natural transformation of functors,
\begin{equation}
    \Hi{0}(\ZZ_t,\W_\omega\otimes-) \simeq \Hi{0}(\ZZ_t,\W_{\omega^2}\otimes-)~.
\end{equation}
Though unphysical, one can consider the vertex algebra $\widetilde{\V}_{1,1}\coloneqq \W_{\df_4}\ast_u(\W_\omega\ast_t\W_\omega)$. This would, naively, correspond to the trinion with two maximal punctures twisted by $\omega$ and one untwisted puncture.
From the natural isomorphism, however, we have that
\begin{equation}
    \widetilde{\V}_{1,1}\cong \V_{1,1}~.
\end{equation}
We can, therefore, use the trinion $\widetilde{\V}_{1,1}$ as the building block for an equivalent construction of genus zero vertex algebras, which are isomorphic to the $\V_{1,1}$ construction. This construction, however, makes manifest the enhanced symmetry of the vertex algebras,~\ie,~the labelling by $\omega$ \emph{versus} $\omega^2$ is redundant.

Let us reiterate, as this strikes us as a surprising result. At the level of the associated vertex algebra, there is are additional automorphisms that swap $\omega$ and $\omega^2$ punctures which, as far as we know, don't arise from an underlying $S$-duality of the four-dimensional physics. For example, the naive symmetry of $\CC_{0,0,2}$ is $\IZ_2\times\IZ_2$, which swaps punctures with the same label. This is enhanced to $S_4$---swapping between all four punctures---as if these were all identical punctures!

In addition to the mixed trinion that we described above, one can compactify on a curve with a three pronged twist line---connecting three punctures each with $\omega$ (or $\omega^2$) monodromy. The trinion with three $\omega$ punctures does indeed correspond to a physical SCFT, and one might expect that this trinion is $\V_{G_2,3}$. Comparing the superconformal index of this trinion theory with the character of $\V_{G_2,3}$, however, exposes this as wishful thinking. Indeed, this trinion is something of a mystery, and we currently have no way of constructing it with the machinery of~\cite{Arakawa:2018egx}. For now, we shall refer to the surface with the three pronged twist line by $\CC_{0,0,3\omega}$ and its corresponding vertex algebra by $\V_{0,3\omega}$. We also introduce their $\sigma$-conjugates $\CC_{0,0,3\omega^2}$ and $\V_{0,3\omega^2}$, which correspond to the surface with a three pronged twist line between three $\omega^2$ twisted punctures and its associated vertex algebra.

Unlike all of the other trinions we have considered, we expect $\V_{0,3\omega}$ to be derived/contain fermionic operators/have a cohomological grading in which it is not concentrated in degree zero. (This is predicted by our proposed diagnostic concerning the covering space of the UV curve, and is compatible with some speculative analyses of the superconformal index.) Any construction involving Feigin--Frenkel gluing is forced to be in degree zero, since we manually truncate to the zeroth cohomology, so this expectation implies the necessity of other tools to get at this vertex algebra.

\section{\label{sec:conc_remarks}Further observations}

\subsection{\label{ssec:conc_remarks/3d}Three dimensional mirrors and the non-simply laced case}

Having identified the appropriate vertex algebras to identify with the twisted theories of type $\CC_{0,m,n}$, there remains a question of how to understand Arakawa's $\V_{G_t,s}$ for non-simply laced $G$. We wish to speculatively suggest a physical interpretation of these vertex algebras, but our suggestion will require a diersion to three dimensions.

The Higgs branches of class $\SS$ theories, which in the mathematical literature have come to be known as \emph{Moore--Tachikawa varieties} following the work of \cite{Moore:2011ee}, are at present most uniformly understood in terms of circle reduction to three-dimensions. Reducing a four-dimensional $\mathcal{N}=2$ theory on a circle results in a three dimensional $\mathcal{N}=4$ theory with the same Higgs branch, and when that three-dimensional theory has a mirror dual then the Coulomb branch of that dual theory will reproduce the Higgs branch in question. The three dimensional mirrors of the $\af_n$ and $\df_n$ series of class $\SS$ theories were derived in~\cite{Benini:2010uu} using brane web technology. In the same work, mirror theories for $\df_n$ theories with twisted punctures were also presented. The twisted $\af_{2n}$ case has also been explored recently in~\cite{Beratto:2020wmn}.

In all of these cases, the three dimensional mirrors are quiver gauge theories---they are \emph{star shaped quivers} with a central node from which tails radiate outwards, one for each puncture. Each type of puncture gives rise to a different tail. For example, the maximal untwisted punctures in type $\df_n$ give rise to tails matching the Lagrangian description for the theory $T[SO(2n)]$ seen in Figure~\ref{fig:u-tail}~\cite{Gaiotto:2009sym}. For a genus zero theory with maximal punctures, the quiver is a central $SO(2n)$ node with $T[SO(2n)]$ tails radiating off. For genus $g>0$, the quiver is the same but with the addition of $g$ many hypermultiplets valued in the adjoint of $SO(2n)$, which look like $g$ loops starting and ending on the central $SO(2n)$ node.

When there is a mixture of twisted and untwisted punctures, the central node is then replaced by $O(2n-1)$, and the tails for maximal twisted punctures correspond to the Lagrangian for the $T[SO(2n-1)]$ theory, which is displayed in Figure~\ref{fig:t-tail}. Finally, in addition to the extra adjoint matter arising from positive genus, there are an extra $2s_t+2g-2$ fundamental hypermultiplets of $O(2n-1)$, where $2s_t$ is the number of twisted punctures. Note that at genus zero, these additional fundamental hypermultiplets only appear in the presence of \emph{four or more twisted punctures}. In light of our discussions on residual gauge symmetries, these extra fundamental hypers might be seen as indicative of the residual gauge symmertry/derived structure that arises in the twisted setting. (An interesting special case is for the $\df_2$ theory, where the extra fundamentals can be directly identifid with extra adjoints of the twisted algebra $\af_1$.)

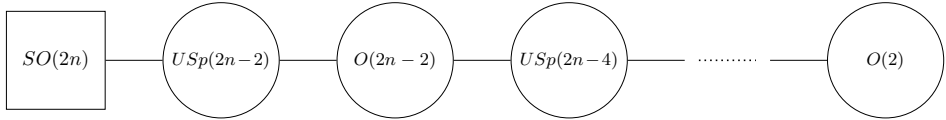
\begin{figure}[!t]
	\centering
	\scalebox{.75}{\begin{tikzpicture}[
	rnd/.style={circle, draw=black, text width={width("USp(2n-2)")+4pt},font=\small,align=center,},
	sqr/.style={rectangle, draw=black, minimum size={width("USp(2n-2)")+4pt}},]

\node[sqr] (centre) {$SO(2n)$};
\node[rnd] (t1) [right=of centre]  {$USp(2n-2)$};
\node[rnd] (t2) [right=of t1]  {$O(2n-2)$};
\node[rnd] (t3) [right=of t2]  {$USp(2n-4)$};
\node[minimum size=0] (d1) [right=of t3] {};
\node[minimum size=0] (d2) [right=of d1] {};
\node[rnd] (t4) [right=of d2]  {$O(2)$};

\draw[-] (centre.east)--(t1.west);
\draw[-] (t1.east)--(t2.west);
\draw[-] (t2.east)--(t3.west);
\draw[-] (t3.east)--(d1.west);
\draw[dotted, thick] (d1.east)--(d2.west);
\draw[-] (d2.east)--(t4.west);
\end{tikzpicture}}
	\caption{\label{fig:u-tail}The Lagrangian for the $T[SO(2n)]$ theory that is used to introduce a maximal untwisted puncture in a theory of type $\df_n$.}
\end{figure}

\begin{figure}[!t]
	\centering
	\scalebox{.75}{	\begin{tikzpicture}[
	rnd/.style={circle, draw=black, text width={width("USp(2n-2)")+4pt},font=\small,align=center,},
	sqr/.style={rectangle, draw=black, minimum size={width("USp(2n-2)")+4pt}},]

	\node[sqr] (centre) {$O(2n-1)$};
	\node[rnd] (t1) [right=of centre]  {$USp(2n-2)$};
	\node[rnd] (t2) [right=of t1]  {$O(2n-3)$};
	\node[rnd] (t3) [right=of t2]  {$O(2n-4)$};
	\node[minimum size=0] (d1) [right=of t3] {};
	\node[minimum size=0] (d2) [right=of d1] {};
	\node[rnd] (t4) [right=of d2]  {$O(1)$};

	\draw[-] (centre.east)--(t1.west);
	\draw[-] (t1.east)--(t2.west);
	\draw[-] (t2.east)--(t3.west);
	\draw[-] (t3.east)--(d1.west);
	\draw[dotted, thick] (d1.east)--(d2.west);
	\draw[-] (d2.east)--(t4.west);
	\end{tikzpicture}}
	\caption{\label{fig:t-tail}The Lagrangian for $T[SO(2n-1)]$ theory that is used to introduce a maximal $\mathbb{Z}_2$-twisted puncture in a theory of type $\df_n$.}
\end{figure}
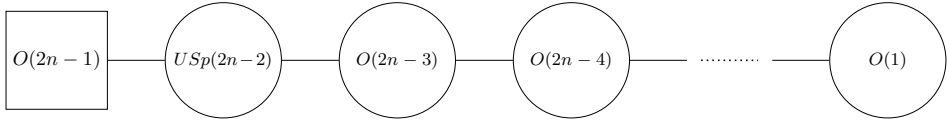

In a series of papers~\cite{Nakajima:2015txa,Braverman:2016wma,Braverman:2016pwk, Braverman:2017ofm}, (subsets of) Braverman, Finkelberg, and Nakajima (BFN) have introduced a mathematical construction of the Coulomb branches of three dimensional $\mathcal{N}=4$ gauge theories. For the $\af_n$ case, one can apply the techniques of~\cite{Braverman:2017ofm} to construct the Coulomb branches of star shaped quivers via the geometric Satake correspondence. Similarly, Ginzburg and Kazhdan~\cite{Ginzburg:2018}, in unpublished work, have provided a systematic construction for the Moore--Tachikawa varieties, which agrees with the construction of~\cite{Braverman:2017ofm}. Their construction---much like Arakawa's---is well-defined for any semisimple Lie algebra $\gf$. It stands to reason that the associated varieties of $\V_{G,s}$ should be the Coulomb branches of the star shaped quivers when $G$ is simply laced.

As we have argued, the vertex algebras $\V_{G,s}$ for non-simply laced $G$ do not correspond to SCFTs in class $\SS$. Nevertheless, they might be understood as boundary vertex algebras in the sense of \cite{Gaiotto:2016wcv,Costello:2018swh,Costello:2019} for a suitable three-dimensional $\NN=4$ theory. Arakawa has shown that the associated varieties for these algebras match the Ginzburg--Kazhdan construction for non-simply laced $\gf$, and we would suggest that these algebras $\V_{G,2s}$ should then be expected to arise as boundary VOAs in the $C$-twist for the star shaped quiver mirrors of the twisted $\CC_{0,0,s}$ theories \textit{without the extra fundamental matter}.

For example, consider the $\df_n$ theory on $\CC_{0,0,2}$. The corresponding vertex algebra has four actions of $\gf_t=\mathfrak{usp}(2n-2)$, and the three dimensional mirror of the theory is given on the left-hand side of Figure~\ref{fig:CnDn-four-punctured}. We propose to identify Arakawa's VOA $\V_{{\mathfrak c}_{n-1},4}$ as a boundary VOA for the three dimensional quiver on the right-hand side of Figure~\ref{fig:CnDn-four-punctured}.

It would be interesting to identify some indication that the $\V_{G_t,s}$ VOAs are not related to four-dimensional physics. At face value they have no serious pathologies---they are conical with negative central charge and satisfy all known unitarity bounds. These vertex algebras, therefore, may warrant some attention with an eye towards characterising precisely what vertex algebras have parent four-dimensional SCFTs.

We note that our construction also descends to a construction of the Moore--Tachikawa varieties of $\CC_{0,m,n}$---via the associated variety functor. Indeed, the associated variety functor commutes with DS reduction~\cite{Arakawa:2015ro}, in the sense that $X_{H^0_{DS}(V)} \cong H^0_{DS}(X_V))$ for any $V\in\KL$, where finite Drinfel'd--Sokolov reduction is used on the right. For the special case of $\CC_{0,m,1}$, whose vertex algebras can be constructed purely by Feigin--Frenkel gluing, this can be seen as a twisted generalisation of the Ginzburg--Kazhdan construction~\cite{Ginzburg:2018}. For the most general $\CC_{0,m,n}$, one would require a derived approach to capture the extra information present in the Hall--Littlewood chiral ring.

From a physics perspective, these associated varieties should be the Coulomb branches of the twisted star-shaped quivers \textit{including} the extra matter in the fundamental. This is of course subject to the Higgs Branch conjecture of~\cite{Beem:2017ooy}. There is scope for future work, both in ironing out the details of the dechiralised construction and recasting it in terms of the three-dimensional mirror, \textit{\`a la} Braverman--Finkelberg--Nakajima.

\begin{figure}[ht]
\begin{minipage}{.45\textwidth}
\centering
\resizebox{.9\textwidth}{!}{	\begin{tikzpicture}[
	rnd/.style={circle, draw=black, text width={width("TtO(2n-1)")+12pt},font=\small,align=center,},
	sqr/.style={rectangle, draw=black, minimum size=5mm},]

	\node[rnd] (centre) {$O(2n-1)$};
	\node[rnd] (t1) [above right=of centre] {$T_t[O(2n-1)]$};
	\node[rnd] (t2) [above left=of centre] {$T_t[O(2n-1)]$};
	\node[rnd] (t3) [below left=of centre] {$T_t[O(2n-1)]$};
	\node[rnd] (t4) [below right=of centre] {$T_t[O(2n-1)]$};


	\node[sqr] (f1) [above=of centre] {$USp(2))$};

	\draw[-] (centre.north east)--(t1.south west);
	\draw[-] (centre.north west)--(t2.south east);
	\draw[-] (centre.south west)--(t3.north east);
	\draw[-] (centre.south east)--(t4.north west);


	\draw[-] (centre.north)--(f1.south);


	\end{tikzpicture}}
\end{minipage}
\begin{minipage}{.1\textwidth}
\phantom{abc}
\end{minipage}
\begin{minipage}{.45\textwidth}
\centering
\resizebox{.9\textwidth}{!}{	\begin{tikzpicture}[
	rnd/.style={circle, draw=black, text width={width("TtO(2n-1)")+12pt},font=\small,align=center,},
	sqr/.style={rectangle, draw=black, minimum size=5mm},]

	\node[rnd] (centre) {$O(2n-1)$};
	\node[rnd] (t1) [above right=of centre] {$T_t[O(2n-1)]$};
	\node[rnd] (t2) [above left=of centre] {$T_t[O(2n-1)]$};
	\node[rnd] (t3) [below left=of centre] {$T_t[O(2n-1)]$};
	\node[rnd] (t4) [below right=of centre] {$T_t[O(2n-1)]$};


	\draw[-] (centre.north east)--(t1.south west);
	\draw[-] (centre.north west)--(t2.south east);
	\draw[-] (centre.south west)--(t3.north east);
	\draw[-] (centre.south east)--(t4.north west);




	\end{tikzpicture}}
\end{minipage}
\caption{\label{fig:CnDn-four-punctured}The three dimensional mirrors of the $D_n$ theory on $\CC_{0,0,2}$ (left) and of Arakawa's $\V_{{\mathfrak c}_{n-1},4}$ (right). We claim that the quiver variety of this mirror corresponds to the associated variety of $\V_{{\mathfrak c}_{n-1},4}$.}
\end{figure}
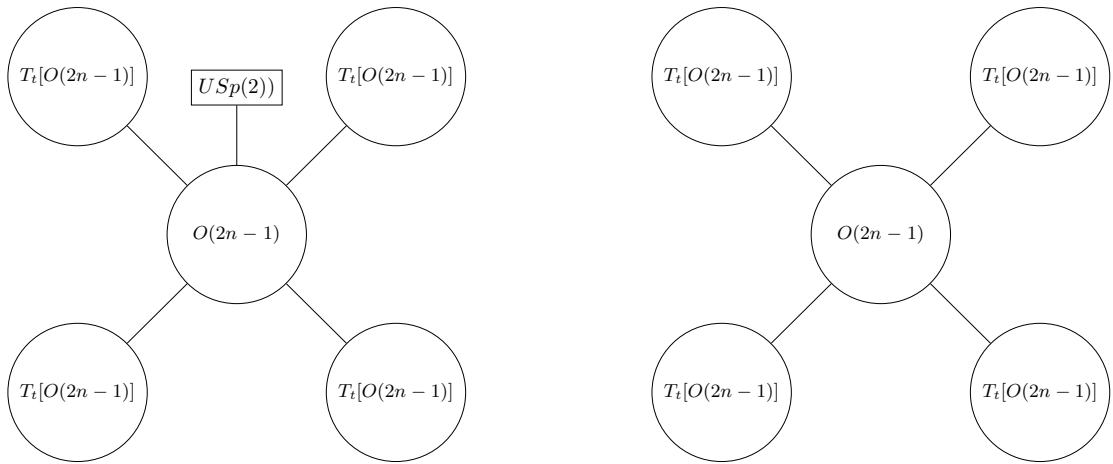

\subsection{\label{ssec:acd_gluing}Unexpected Feigin--Frenkel gluings}

Looking at Table~\ref{tab:outer_aut}, one sees that both $\af_{2n}$ and $\df_n$ share the same twisted algebra, $\mathfrak{c}_{n-1}$. Our Theorem~\ref{thm:closed_immersion} is quite flexible, and in principle allows for gluings beyond the mixed vertex algebras we have described. The vertex algebra object $\W_{\df_n}\ast_{\df_n} \W_{{\mathfrak c}_{n-1}}$ is part of our construction, and indeed corresponds to a nice physical SCFT. However, since this object is in $\KL_{{\mathfrak c}_{n-1}}$ and so in $\ZZ_{{\mathfrak c}_{n-1}}-$Mod, there is nothing stopping us from considering the gluing
\begin{equation}
	\V_{\af\ccf\df}\colonequals\W_{\af_{2n}}\ast_{\af_{2n}}(\W_{\df_n}\ast_{\df_n}\W_{{\mathfrak c}_{n-1}})~,
\end{equation}
where $\ZZ_{\af_{2n}}$ acts on the right by the surjection, $\ZZ_{\af_{2n}}\twoheadrightarrow \ZZ_{{\mathfrak c}_{n-1}}$.

This might seem like a forced construction, but it is part of a larger ``duality'' web. Consider, for example, the mixed trinions $\V_{1,1}^{\af_{2n}}$ and $\V_{1,1}^{\df_n}$. Both of these are vertex algebra objects in $\KL_{{\mathfrak c}_{n-1}}$ and so, in principle, there is nothing stopping us from considering the BRST gauging
\begin{equation}\label{eq:illegal_gauging}
	\V_{\af_{2n},1,1}\circ_{{\mathfrak c}_{n-1}}\V_{\df_n,1,1}~.
\end{equation}
This is a vertex algebra with $\KL$ actions of $V^{-\kappa_c}(\gf)$ for $\gf=\mathfrak{a}_{2n},\mathfrak{d}_{n},$ and $\mathfrak{c}_{n-1}$. This does not, obviously, correspond to any of our---or Arakawa's---constructions. Thinking pictorially, this vertex algebra should have another decomposition, given by the gauging $\V_{\af\ccf\df}\circ_{{\mathfrak c}_{n-1}} \V_{{\mathfrak c}_{n-1},3}$, where $\V_{{\mathfrak c}_{n-1},3}$ is the genus zero trinion of type ${\mathfrak c}_{n-1}$---constructed as in~\cite{Arakawa:2018egx}.

From a physical perspective, the gauging~\eqref{eq:illegal_gauging} is inconsistent; the theory associated to $\CC_{0,1,1}$ for type $\af_{2n}$ has a $\IZ_2$ Witten anomaly, while the $\df_n$ theory does not. The diagonal action of $\mathfrak{c}_{n-1}$ on the product is anomalous and cannot be gauged. Yet, at the level of vertex algebras, there seem to be no immediate obstructions to this.

Since we have opened the doors, one might wonder whether there are other such unexpected/unphysical constructions, but we don't believe this to be the case. Indeed, let $\gf_1$ and $\gf_2$ be Lie algebras such that there exists some surjection $\zf(\gf_1)\twoheadrightarrow \zf(\gf_2)$. Unlike the case where $\gf_2$ is the twisted algebra, there is no guarantee that this surjection survives to the quotients $\zf(\gf_1)_\lambda$. While the operation $\W_{\gf_1}\ast_{\gf_1}\W_{\gf_2}$ is well-defined, we expect that it will not share many of the desirable properties of the mixed vertex algebras. As a graded vector space, the zeroth cohomology decomposes into a sum $\bigoplus_{\lambda\in P^+_2} \IV_{\lambda}^1\otimes_{\zf(\gf_1)_\lambda}\IV_{\lambda}^2$. Here, we seek an embedding $P^+_{2}\hookrightarrow P_1^+$ such that the image of $\lambda\in P^+_1$ is the highest weight of a finite dimensional, irreducible $\gf_2$ representation on which the $P_{i}$ act as zero if $P_{i,(0)}$ is in the kernel of $\zf(\gf_1)\twoheadrightarrow \zf(\gf_2)$. When this embedding does not arise from an outer-automorphism, the set of valid $\lambda$ is small and possibly contains only the trivial representation.

Indeed, the constuction of $\af\ccf\df$ vertex algebras represents the only real flexibility offered by Theorem~\ref{thm:closed_immersion}. Repeating our arguments for the characters of $\V_{m,n}$, we can derive the character of the vertex algebra $\V_{\af\ccf\df}$. This character is given by
\begin{equation}
	\ch\V_{\af\ccf\df}	= \sum_{\lambda\in P_{{\mathfrak c}_{n-1}}^+}\frac{\KK_\af(\mathbf{a})\chi^\lambda_A(\mathbf{a})\KK_\ccf(\times)\chi^\lambda_\ccf(\times)\KK_\ccf(\mathbf{c})\chi^\lambda_\ccf(\mathbf{c})\KK_\df(\mathbf{d})\chi^\lambda_\df(\mathbf{d})}{\KK_\af(\times)\chi^\lambda_\af(\times)\KK_\df(\times)\chi^\lambda_\df(\times)}~,
\end{equation}
where we use the subscripts $\af,{\mathfrak c},\df$ to distinguish between the $\KK$-factors and Schur polynomials of each type and $\mathbf{a}$, $\mathbf{c}$, and $\mathbf{d}$ are fugacities of $\af_n$, ${\mathfrak c}_n$, and $\df_n$ respectively.

These triply mixed vertex algebras form another class of examples that seem to have no corresponding four-dimensional SCFT. Nevertheless, one can imagine associating (mirror) quivers, in the obvious way, to these vertex algebras in three dimensions. So it is conceivable that these vertex algebras arise from the boundary construction of~\cite{Gaiotto:2016wcv,Costello:2018swh,Costello:2019} as well.

\subsection{\label{ssec:fixtures/new-perspective}A six dimensional perspective on the index}

There is a particularly nice interpretation of the twisted class $\SS$ characters in light of our undersatnding of Feigin--Frenkel gluing and the six-dimensional origins of these SCFTs. The character for an untwisted three punctured sphere, $\V_{G,3}$ can be rewritten in terms of characters of Weyl modules as
\begin{equation}
	\ch \, \V_{G,3} = \sum_{\lambda \in P_G^+} \frac{ \ch \, \IV_\lambda \ch \, \IV_\lambda\ch \, \IV_\lambda}{\ch \, \zf_\lambda}~.
\end{equation}
Each Weyl module in the numerator arises from the decomposition of $\W$ into $\KL^{ord}$ objects. We then identify each pair of caps through their Feigin--Frenkel centres, and in doing so we quotient by a copy of $\zf_\lambda$ for each gluing. However, we have not accounted for the generators of the cap that are not present in the Weyl module. These additional generators can be heuristically understood in terms of descendants of local operators in six dimensions.

We aim to interpret Schur operators in four dimensions as being realised directly as descendants of Schur-type operators in six dimensions, possibly in the presence of wrapped defects.\footnote{Such an interpretation has been explored previously in \cite{gaioul} and \cite{BRZ}. The organisation of local operator indices (in the presence of intersecting defects) in six dimensions in terms of affine Kac--Moody and principal $W$-algebra characters has also appeared previously in \cite{Bullimore:2014upa}.} The Schur-type operators in six-dimensions themselves can be given the structure of a vertex algebra---specifically the vertex algebra $\WW(\gf,f_{prin})$~\cite{Beem:2014kka}---though in the present arrangement they should be thought of as realising the commutative limit of said $\WW$-algebra. When reducing to four dimensions, local operators in six dimensions can be directly reduced (integrated against the holomorphic volume form on the UV curve), or they can be reduced after performing supersymmetric descent and integrating against closed one- or two-forms on the UV curve, with one-forms getting paired with fermionic descendants. Consequently, as a graded vector space, the six-dimensional $\WW(\gf, f_{prin})$ operators give rise to four-dimensional operators with character
\begin{equation}
	\ch\WW(\gf_u,f_{prin})^{\chi(\CC)} = (\ch\, \zf_0)^{2-2g}~,
\end{equation}
and for $g\neq0$ this includes fermionic operators that reflect the derived structure/residual gauge symmetry of these theories. There are additional contributions from wrapped codimension two defects, which in six dimensions support the vertex algebra $V^{\kappa_c}(\gf_u)$~\cite{Beem:2014kka}. These contribute factors of $\ch\, \IV_0$ to the character. However, the Feigin--Frenkel centres of the subalgebras are identified through the bulk of the UV curve, so the contribution is $\ch\, \IV_0/\ch\,\zf_0$. For a general surface $\CC_{g,s}$, this gives
\begin{equation}
	 \left(\frac{\ch \IV_0}{\ch \zf_0}\right)^{s}\cdot (\ch \zf_0)^{2-2g}~.
\end{equation}
The four-dimensional descendants of local operators in six dimensions (in the presence of codimension-two defects), therefore account for the $\lambda=0$ summand of the index. The higher terms in the TQFT expansion of the index arise from codimension four defects, which are labelled by dominant integral weights $\lambda\in P^+$. Upon compactification, the codimension four defects can wrap around the surface, $\CC_{g,s}$ and intersect with the codimension two defects, leading to a $\hat{\gf}_{\kappa_c}$ module structure. Indeed, these modules are classified by modules of highest weight $\lambda$. The fact that the Feigin--Frenkel centre is shared can be, in a loose sense, seen as a consequence of the fact that the modules all arise from the same wrapped defect. Similarly, the local operators in six-dimensions can also be inserted on the wrapped defects, leading to modules---of highest weight $\lambda$---of the $\WW$-algebra.

After descent, the image of these operators account for the $\zf_\lambda$ factors in the summand, and for a general surface $\CC_{g,s}$, we have
\begin{equation}
	\ch\, \V(\CC_{g,s}) = \sum_{\lambda \in P_G^+} \left(\frac{\ch \IV_\lambda} {\ch \zf_\lambda}\right)^{s} \cdot (\ch \zf_\lambda)^{2-2g}~.
\end{equation}

So far, this has amounted to little more than a refactoring and reinterpretation of the summands in the untwisted index. To reap some intuitive benefit from this perspective, we will consider the twisted vertex algebras. Let us start with the mixed trinion $\V_{1,1}$, whose character is
\begin{equation}
	\ch\, \V_{1,1} = \sum_{\lambda \in P_t^+} \ch\, \IV^u_\lambda \ch\,\IV^t_\lambda \ch\, \IV^t_\lambda \cdot \left( \frac{1}{\ch\, \zf^u_\lambda \ch\, \zf^t_\lambda} \right) \cdot \ch\,\zf^t_\lambda~.
\end{equation}
The Weyl module characters in the numerator arise from the intersection of the codimension four and two defects, as in the untwisted case. Since there are twisted punctures present, only the invariant generators of the six-dimensional $\WW$-algebra can survive the compactification, so we have extra characters of $\WW(\gf_t,f_{prin})$-modules appearing in the numerators. The characters of the vertex algebras of the form $\V_{m,1}$ can all be interpreted in the same way.

Now let us consider introducing more twisted punctures. Consider $\V_{0,2}$, whose character we can suggestively organise as follows,
\begin{equation}
	\ch \V_{0,2} = \sum_{\lambda \in P_t^+} \ch \,\IV_\lambda^t \ch \,\IV_\lambda^t\ch \,\IV_\lambda^t\ch \,\IV_\lambda^t \cdot \left( \frac{1}{\ch\, \zf_\lambda^t} \right)^3 \cdot \ch\, \zf_\lambda^t \cdot \left( \frac{ \ch\,\zf^t_\lambda}{\ch\,\zf^u_\lambda} \right)^2~.
\end{equation}
Ignoring the final factor of $\ch\, \zf^t_\lambda/\ch\,\zf_\lambda^u$, we see a match to the character of $\V_{G_t,4}$,~\ie\ the genus zero vertex algebra of Arakawa. Unlike the $\CC_{m,1}$ case, from a six-dimensional perspective, there should now be one-forms with $\mathbb{Z}_2$ monodromy around the twisted punctures. Not all generators of $\WW(\gf_u,f_{prin})$ can be paired with these one-forms. Rather, it is precisely those in the kernel of $\zf_{\lambda}^u\twoheadrightarrow \zf_{\lambda}^t$. Like the positive genus case, these pairings with one-forms give rise to fermionic Schur operators in four dimensions. The kernel of $\zf^u_{\lambda}\twoheadrightarrow \zf^t_{\lambda}$ has character given by $\ch\,\zf_\lambda^u/\ch\,\zf_\lambda^t$ but since these are now descending as fermionic operators in four dimensions, they contribute to the index with minus signs, and hence show up in the denominator. The number of one-forms available for pairing can be determined via a Riemann--Hurwitz type computation. In the example above, the factor of $\ch\,\zf_\lambda^t/\ch\,\zf_\lambda^u$ should be interpreted as the contribution from these non-invariant pairings; the power of two comes from the number of (holomorphic and antiholomorphic) one-forms with monodromy.

This six dimensional perspective has been present as an undercurrent throughout this work, and certainly warrants further investigation in the future.


\section*{Acknowledgments}
The authors are grateful to
Tomoyuki Arakawa,
David Ben-Zvi,
Diego Berdeja-Su\'arez,
Dylan Butson,
Wolfger Peelaers,
Shlomo Razamat,
Leonardo Rastelli,
Pavel Safronov,
and Gabi Zafrir
for helpful discussions and comments on this and related work. Christopher Beem's work is supported in part by grant \#494786 from the Simons Foundation, by ERC Consolidator Grant \#864828 ``Algebraic Foundations of Supersymmetric Quantum Field Theory'' (SCFTAlg), and by the STFC consolidated grant ST/T000864/1. Sujay Nair is supported in part by EPSRC studentship \#2272671.
\appendix

\section{\label{app:spectral_sequences}Spectral sequences}

As in \cite{Arakawa:2018egx}, the rearrangement lemmas of Section~\ref{ssec:duality} rely primarily upon the machinery of spectral sequences. The most general treatment of spectral sequences is beyond our needs, but for the unfamiliar reader we review some key statements that underlie our results.

We will be exclusively concerned with spectral sequences associated to a bicomplex. Let $\AA$ be an abelian category. Suppose
\begin{equation}
	C^{\bullet\bullet} = \bigoplus_{p,q}C^{p,q}~,
\end{equation}
is a bicomplex, \ie,~$C^{p,q}$ are objects of $\AA$ with differentials
\begin{equation}
	\begin{split}
	d_h: C^{p,q}\rightarrow C^{p+1,q}~,\\
	d_v: C^{p,q}\rightarrow C^{p,q+1}~,
	\end{split}
\end{equation}
satisfying
\begin{equation}
	d_h\circ d_h =d_v\circ d_v = d_h\circ d_v + d_v\circ d_h =0~.
\end{equation}
There are two notions of total complex of a bicomplex that are common.
\begin{equation}
	\begin{split}
	\mathrm{Tot}_\oplus(C)^p \coloneqq \bigoplus_{m+n=p}C^{m,n}~,\\
	\mathrm{Tot}_\otimes(C)^p \coloneqq \bigotimes_{m+n=p}C^{m,n}~,
	\end{split}
\end{equation}
and for our purposes we will only require the former, and hereafter we omit the $\oplus$ subscript. The total complex is equipped with a differential
\begin{equation}
	d_{tot} = d_h +(-1)^qd_v~,
\end{equation}
making it an honest cochain complex. In this setting, spectral sequences are a tool to ``approximate'' the total cohomology,
\begin{equation}
	H_{tot}^\bullet = H^\bullet(C_{tot},d_{tot})~,
\end{equation}
in terms of the cohomology of $C^{\bullet\bullet}$ with respect to $d_h$ and $d_v$ in an iterative process which converges on to the total cohomology $H_{tot}$.

We follow~\cite{Weibel:1994} for our definition of a (cohomology) spectral sequence. A spectral sequence is a family of objects $\{E_{r}^{p,q}\}$ with differentials
\begin{equation}
	d_r^{pq}:E_{r}^{p,q}\rightarrow E_{r}^{p+r,q-r+1}~,
\end{equation}
satsifying $d^{p+r,q-r+1}_r\circ d^{p,q}_r=0$ and isomorphisms
\begin{equation}
	E^{p,q}_{r+1} =\frac{ \mathrm{ker}(d_{r}^{p,q})}{\mathrm{im}(d_r^{p-r,q+r-1})}.
\end{equation}
For certain spectral sequences there exists some finite $r_0$ such that
\begin{equation}
	E_{r}^{p,q}\cong E_{r_0}^{p,q}~,\qquad r > r_0~.
\end{equation}
This ``stable page'' is usually written as $E_\infty^{p,q}$. There are several technical definitions of what it means for a spectral sequence to converge. For our purposes a spectral sequence, $E_r^{p,q}$, converges to a cohomology $H^\bullet$ if there exists a finite, increasing filtration
\begin{equation}
	0 = F^sH^n \subset \dots \subset F^{p}H^n\subset\dots\subset F^tH^n=H^n~,
\end{equation}
where $s$ and $t$ are finite and
\begin{equation}
	E_\infty^{pq} \cong F^pH^{p+q}/F^{p-1}H^{p+q}~.
\end{equation}
One can then extract the filtration from the stable page and thus compute the total cohomology. There is a special case, where one can simply read off the cohomology without worrying about a filtration, known as \textit{collapse}. A spectral sequence collapses at some page $r>1$, if on that page there is only one nonzero row or column. In practice, we rely on spectral sequences that collapse at the second page.

The total complex $C_{tot}$ comes equipped with a natural filtration
\begin{equation}
	 F_nC_{tot}^p = \mathrm{Tot}_\oplus(\tau_n C)^p~,
\end{equation}
where
\begin{equation}
	(\tau_n C)^{p,q} =
	\begin{cases}
	\makebox[0.8cm][l]{$C^{p,q}$}\text{ if }p\le n\\
	\makebox[0.8cm][l]{$0$} \text{ if }p>n
	\end{cases}~.
\end{equation}
This filtration is natural in the sense that the subcomplex $\tau_n C$ is a truncation of $C$ at the $n$th column. It is a result in homological algebra that such a filtration gives rise to a spectral sequence that, under certain technical conditions, converges on the total cohomology. The zeroth page of the spectral sequence is simply the bicomplex
\begin{equation}
	_I E_0^{p,q} = C^{p,q}~.
\end{equation}
The differential of the first page must act vertically and is simply $d_v$. The first page is therefore
\begin{equation}
	_I E_1^{p,q} = H^q(C^{p,\ast},d_v)~,
\end{equation}
with differential $d_2 = d_h$. The second page is therefore
\begin{equation}
	_I E_2^{p,q} = H^p(H^q(C,d_v),d_h)~.
\end{equation}
The differential $d_3$ is more complicated, in practice it is best if the spectral sequence collapses by the second page, in which case
\begin{equation}
	H_{tot}^p = E^{p,n}_2~,
\end{equation}
assuming the $n$th row is the only nonzero row.

There is another spectral sequence converging to $H_{tot}$---associated to the filtration by rows. The second page of this spectral sequence is given by
\begin{equation}
	_{II} E_2^{p,q} = H^p(H^q(C,d_h),d_v)~.
\end{equation}
We refer to the earlier spectral sequence as $_I E^{p,q}_r$.

For our rearrangement lemmas, we are frequently interested in when cohomology operations commute with each other, \ie,
\begin{equation}
	H^p\tilde{H}^q\stackrel{?}{\cong}\tilde{H}^qH^p~.
\end{equation}
If the two cohomologies conspire to form a bicomplex, then this is equivalent to asking whether
\begin{equation}
	_IE_2^{p,q}\stackrel{?}{\cong}\,_{II}E_2^{p,q}~.
\end{equation}
We see that a sufficient condition for this equivalence is if the two sequences both collapse on the second page, since the two sides are then both isomorphic to the total cohomology. If we are only interested in whether the zeroth cohomologies commute then we can replace the above by a different and often weaker condition. As long as the spectral sequence is first quadrant---$E_r^{p,q}=0$ for $p,q<0$---the $E^{0,0}_r$ stabilizes by the second page. Thus, if $H^p$ and $\tilde{H}^p$ vanish for $p<0$ then
\begin{equation}
	H^0\tilde{H}^0\cong\tilde{H}^0H^0~.
\end{equation}


\section{\label{sec:proof_closed_immersion}Miura opers and a proof of Theorem~\ref{thm:closed_immersion}}

The proof of Theorem~\ref{thm:closed_immersion} relies on the machinery of Miura opers and Cartan connections. We give a brief introduction to both of these before giving a proof of the theorem.

\subsection{Miura opers and the no-monodromy condition}

Let $G$ be a reductive algebraic group and $B$ a Borel subgroup of $G$. A Miura $G$-oper on an algebraic curve $X$ is the tuple $(\mathscr{F},\nabla,\mathscr{F}_B,\mathscr{F}_B')$, where $(\mathscr{F},\nabla,\mathscr{F}_B)$ is a $G$-oper and $\mathscr{F}_B'$ is another $B$-reduction of $\mathscr{F}$ that is preserved by $\nabla$. We denote the space of Miura opers over $X$ by $\mathrm{MOp}_G(X)$ and there is a natural forgetful morphism
\begin{equation}\label{eq:miura_forget}
	\MOp_G(X) \rightarrow \mathrm{Op}_G(X)~.
\end{equation}
A Miura oper is called \textit{generic} if the two bundles $\mathscr{F}_B$ and $\mathscr{F}_B'$ are in generic relative position to each other. The space of such generic Miura opers over $X$ is denoted by $\MOp_{G}(X)_{gen}$. Once again we restrict to the cases where $X$ is the formal disc, $D$, or the formal punctured disc, $D^\times$.

Now let $N$ be the unipotent subgroup of $B$ with Lie algebra $\mathfrak{n}$ and $H$ be a choice of the maximal torus so $B=H\ltimes N$. We can define a principal $H$-bundle $\mathscr{F}_H$ on $X$ according to
\begin{equation}
	\mathscr{F}_H = \mathscr{F}_B \underset{B}\times H = \mathscr{F}_B/N~.
\end{equation}
It can be shown (see, \eg, Lemma 4.2.1 of~\cite{Frenkel:2007}) that $\mathscr{F}_H$ is isomorphic to the $H$-bundle $\Omega^{\rho^\vee}$---characterised in terms of its transition functions, which for a change of coordinates $t=\phi(s)$, are given by
\begin{equation}
	\phi'(s)\rho^\vee~.
\end{equation}
Here, $\rho^\vee$ is half the sum of positive coroots that we identify as an element of the Cartan subalgebra of $\gf$ using the identification $(\mathfrak{h}^\vee)^* \cong \mathfrak{h}$. The space of connections of $\mathrm{Conn}(\Omega^{\rho^\vee})_D$ on the bundle $\Omega^{\rho^\vee}$ over $D$ consists of connections of the form
\begin{equation}
	\nabla = \partial_t + \mathbf{u}(t)\,,\qquad \mathbf{u}(t)\in\mathfrak{h}[\![t]\!]~,
\end{equation}
for some choice of coordinate $t$ on the disc. Under a coordinate change $t=\phi(s)$, the connection transforms as
\begin{equation}\label{eq:cartan_coord_transform}
	\tilde{\mathbf{u}} = \phi(s)\mathbf{u}(\phi(s)) - \frac{\phi''(s)}{\phi'(s)}\rho^\vee~.
\end{equation}
For the $H$-bundle $\Omega^{\rho^\vee}$ on $D^\times$, the space of connections $\mathrm{Conn}(\Omega^{\rho^\vee})_{D^\times}$ is similarly populated but with $\mathrm{u}\in \IC(\!(t)\!)$.

We introduce this Cartan connection to exploit the isomorphism (see \cite{Frenkel:2007}, Proposition $8.2.2$),
\begin{equation}
	\mathrm{Conn}(\Omega^{\rho^\vee})_{D}\xrightarrow{\sim} \MOp_{G}(D)_{gen}~.
\end{equation}
Composing with the forgetful morphism \eqref{eq:miura_forget} gives the Miura transform
\begin{equation}
	\mu: \mathrm{Conn}(\Omega^{\rho^\vee})_D \rightarrow \Op_{G}(D)~.
\end{equation}
Every Miura oper on $D^\times$ is generic~\cite{Frenkel:2007}, and so we have a Miura transform
\begin{equation}
	\mu: \mathrm{Conn}(\Omega^{\rho^\vee})_{D^\times}\xrightarrow{\sim} \MOp_G(D^\times) \rightarrow \Op_{G}(D^\times)~.
\end{equation}

As we are ultimately interested in the Feigin-Frenkel centre for $\mathfrak{g}$, we work over the Langlands dual group, $^L G$, with Borel subgroup $^L B$ and Cartan subgroup $^L H$. The $^L H$-bundle of interest is $\Omega^\rho$, and the Miura transform is the morphism
\begin{equation}
	\mu : \mathrm{Conn}(\Omega^{\rho})_{D^\times}\rightarrow \Op_{^L G} (D^\times) \cong\mathrm{Spec}~ \ZZ~.
\end{equation}

Now let us specialise to the subspace $\mathrm{Conn}(\Omega^{\rho})^{\lambda}_{D}\subset\mathrm{Conn}(\Omega^{\rho})_{D^\times}$, consisting of connections of the form
\begin{equation}
	\nabla = \partial_t + \frac{\lambda}{t} + \mathbf{u}(t)\,, \qquad \mathbf{u}(t) \in \mathfrak{h}^\vee[[t]]~,
\end{equation}
where we think of $\lambda$ as an element in the Cartan subalgebra, $\mathfrak{h}^\vee$, via the identification $\mathfrak{h}^* \cong \mathfrak{h}^\vee$. The ring of functions over such connections is
\begin{equation}
	\mathrm{Fun\, Conn}(\Omega^\rho)_{D^\times}^\lambda \cong \IC[u_{i,n}|\,i = 1,\dots,\mathrm{rk}\, \gf^\vee; n<0]~,
\end{equation}
where $u_i$ are the coordinates with respect to some choice of basis, $(h_i)_{i=1}^{\mathrm{rk}\,\gf}$, of the Cartan subalgebra.

The image of this subspace under the Miura transform is exactly the subscheme $\Op_G^{\lambda}\cong \mathrm{Spec}\,\zf_\lambda$ (\cite{Frenkel:2007}, Lemma $9.6.2$). Furthermore, at the level of functions,
\begin{equation}
	\zf_\lambda \cong \Op_G^\lambda \cong (\mathrm{Fun\, Conn}(\Omega^\rho)_{D^\times}^\lambda )^{\mathfrak{n}^\vee}~.
\end{equation}
The action of $\mathfrak{n}^\vee$ on $\mathrm{Fun\, Conn}(\Omega^\rho)_{D^\times}^\lambda$ is realised by the screening charges,
\begin{equation}
	V_i[\lambda_i +1] = \sum_{j=1}^{\mathrm{rk}\,\gf}a_{ji} \sum _{n\ge\lambda_i} x_{i,n-\lambda_i} \frac{\partial}{\partial u_{j,-n-1}}~,
\end{equation}
where the polynomials $x_{i,n}$ are determined via
\begin{equation}
	\sum_{n \le 0} x_{i,n}t^{-n} = \mathrm{exp}\, \left( \sum_{m>0} \frac{u_{i,-m}}{m} t^m \right)~.
\end{equation}
In this sense, $\mathrm{Fun\, Conn}(\Omega^\rho)_{D^\times}^\lambda$, provides a kind of free-field realisation of the opers satisfying the no-monodromy condition.

\subsection{Proof of Theorem~\ref{thm:closed_immersion}}

Let us now apply the above machinery to our problem. We have the following morphisms
\begin{equation}
	\begin{split}
	\zf_{u,\lambda} &\hookrightarrow \mathrm{Fun\, Conn}(\Omega^\rho)_{^LG_u,D^\times}^\lambda~,\\
	\zf_{t,\lambda} &\hookrightarrow \mathrm{Fun\, Conn}(\Omega^\rho)_{^LG_t,D^\times}^\lambda~.
	\end{split}
\end{equation}
There is a natural action of $\mathrm{Out}(\gf_u)$ on $\mathrm{Fun\, Conn}(\Omega^\rho)_{^LG_u,D^\times}^\lambda$. We also have a canonical projection $\mathrm{Fun\, Conn}(\Omega^\rho)_{^LG_u,D^\times}^\lambda\twoheadrightarrow \mathrm{Fun\, Conn}(\Omega^\rho)_{^LG_t,D^\times}^\lambda$, via the isomorphism
\begin{equation}
	\mathrm{Fun\, Conn}(\Omega^\rho)_{^LG_t,D^\times}^\lambda \cong (\mathrm{Fun\, Conn}(\Omega^\rho)_{^LG_u,D^\times}^\lambda)_\sigma~,
\end{equation}
where $(-)_\sigma$, denotes the functor of coinvariants with respect to the subgroup $\langle\sigma\rangle\subset\mathrm{Out}(\gf)$ generated by the twisting automorphism. Concretely, we have the isomorphism
\begin{equation}
	\mathrm{Fun\, Conn}(\Omega^\rho)_{^LG_t,D^\times}^\lambda \cong \IC\left[\tilde{u}_{i,n} | i=1,\dots,\mathrm{rk}\, \gf_t;\,n\in\IZ\right]~,
\end{equation}
where
\begin{equation}
	\tilde{u}_{i,n} = \frac{1}{|\langle\sigma\rangle|}\sum_{\sigma'\in\langle\sigma\rangle}u_{\sigma'{i},n}~.
\end{equation}
These are precisely the linear combinations of the generators of $\mathrm{Fun\, Conn}(\Omega^\rho)_{^LG_u,D^\times}^\lambda$ that are invariant under $\langle\sigma\rangle$.

Assembling, we have the diagram,
\begin{equation}
	\begin{tikzcd}
	\mathrm{Fun\, Op}_{G_u}^\lambda \arrow[hookrightarrow]{r}{\iota} \arrow[dashed,"\pi"]{d}& \mathrm{Fun\, Conn}(\Omega^{-\rho})^\lambda_D\arrow[twoheadrightarrow]{d}{\pi_\sigma} \\
	\mathrm{Fun\, Op}_{^LG_t}^\lambda \arrow[hookrightarrow]{r}{\iota_\sigma} & (\mathrm{Fun\, Conn}(\Omega^{-\rho})^\lambda_D)_\sigma
	\end{tikzcd}
\end{equation}
where the horizontal morphisms are the natural inclusions of $\mathfrak{n}^\vee$-invariants and the vertical arrow is the natural projection to the coinvariants. We can define a map $\pi$ that, we claim, makes the diagram commute via $\pi = \pi_\sigma\circ\iota$. The map $\pi$ is then nothing more than the projection of $\mathrm{Fun \, Op}^\lambda_{G_u}$ to its $\langle \sigma \rangle $ coinvariants.

Our Theorem~\ref{thm:closed_immersion}, is therefore equivalent to the claim that $\pi$ is surjective. Indeed, surjectivity of $\pi$ is precisely the statement that $\zf_\lambda^u \twoheadrightarrow \zf_\lambda^t$ is a surjection.
\begin{thm}
The map $\pi:\mathrm{Fun\, Op}_{G_u}^\lambda \rightarrow \mathrm{Fun\, Op}_{^LG_t}^\lambda$, defined as above, is surjective.
\end{thm}
\begin{proof}
The action of $\mathfrak{n}_u^\sigma$ on the space of coinvariants can be realised through the symmetrised screening charges
\begin{equation}
	\tilde{V}_i[\lambda_i +1] = \frac{1}{|\langle \sigma \rangle|} \sum_{\sigma'\in\langle \sigma \rangle} V_{\sigma(i)}[\lambda_{\sigma(i)}+1]~.
\end{equation}
A polynomial $P\in\mathrm{Fun\, Conn}(\Omega^{-\rho})$ is in the space of invariants, $(\mathrm{Fun\, Conn}(\Omega^{-\rho}))^{\mathfrak{n}_u}$, if and only if it is in the intersection of the kernels of the screening charges,~\ie
\begin{equation}
	V_i[\lambda_i +1]P \overset{!}= 0\,,\qquad \text{for } i=1,\dots,\mathrm{rk}\,\gf_u~.
\end{equation}
Similarly, a polynomial $P\in(\mathrm{Fun\, Conn}(\Omega^{-\rho}))_\sigma$ in the space of coinvariants is an $\mathfrak{n}_u^\sigma$ invariant, if and only if it lies in the intersection of the kernels of the symmetrised screening charges $\tilde{V}_i[\lambda_i+1]$.

It is, therefore, sufficient to show that if $P$ is a $\sigma$-coinvariant and lies in the interection of the kernels of the symmetrised screening charges, then it must lie in the intersection of the kernels of the $\mathfrak{n}_u$ screening charges.
We make the following observations. Suppose $P$ lies in the space of coinvariants, then we can realise it as $P\in\IC[\tilde{u}_{i,n}|i=1,\dots,\mathrm{rk}\,\gf_t;\,n<0]$. Therefore,
\begin{equation}
	\frac{\partial P}{\partial u_{j,n}} = \sigma \left( \frac{\partial P}{\partial u_{j,n}}\right) = \frac{\partial P}{\partial u_{\sigma(j),n}} ~.
\end{equation}
Furthermore, the $x_{i,m}$ must transform as,
\begin{equation}
	\sigma(x_{i,m}) = x_{\sigma(i),m}~,
\end{equation}
which follows straightforwardly from the definition. Combining, we have that
\begin{equation}\label{eq:sigma_screening_charge}
	\sigma\left(V_i[\lambda_i + 1] P\right) = V_{\sigma(i)}[\lambda_{\sigma(i)}+1]P~.
\end{equation}
Let $P\in(\mathrm{Fun\, Conn}(\Omega^{-\rho}))_\sigma$ and suppose it lies in the intersection of kernels of the symmetrised operators. Additionally, let us restrict to the case where $\sigma^2=1$ and we only consider the $\IZ_2$ twist---we shall address the $\IZ_3$ twist for $\df_4$ at the end. Now, we have that
\begin{equation}
	(V_{i}[\lambda_i+1] + V_{\sigma(i)}[\lambda_{\sigma(i)}+1])P = 0~,
\end{equation}
but from (\ref{eq:sigma_screening_charge}), we must have that
\begin{equation}
	\sigma\left(V_{i}[\lambda_i+1]P\right)=-V_{i}[\lambda_i+1]P~.
\end{equation}
We shall now show that for any $P$ in the space of coinvariants, the polynomial $V_i[\lambda_i+1]P$, cannot have eigenvalue $-1$ under $\sigma$.

The image of $P$ under the $i$th screening charge is
\begin{equation}
	V_i[\lambda_i +1]P = \sum_{n\ge\lambda_i} x_{i,n-\lambda_i}\cdot \left( a_{ij} \frac{\partial P}{\partial u_{j,-n-1}} \right)~,
\end{equation}
and we have
\begin{equation}
	\sigma \left( V_i[\lambda_i +1]P \right) = V_{\sigma(i)}[\lambda_{i}+1]P= \sum_{n\ge\lambda_{i}} x_{\sigma(i),n-\lambda_{i}}\cdot \left( a_{ij} \frac{\partial P}{\partial u_{j,-n-1}} \right)~,
\end{equation}
where we have made use of the fact that the derivatives of $P$ and the highest weight, $\lambda$, are invariant under $\sigma$. From inspection, it is clear that
\begin{equation}
	V_{\sigma(i)}[\lambda_{\sigma(i)}+1]P\neq- V_i[\lambda_i +1]P ~,
\end{equation}
unless both are identically zero, as desired.

Now we address the $\IZ_3$ case of $\gf_u=\df_4$ and $\gf_t=\gf_2$. A $\mathfrak{n}_u^\sigma$ invariant $P$ must satisfy the following:
\begin{equation}
	\begin{split}
	(V_{1}[\lambda_1+1] + V_{3}[\lambda_1+1]+V_4[\lambda_1+1])P=0~,\\
	V_{2}[\lambda_2+1]P=0~.
	\end{split}
\end{equation}
We have, once more, made use of the fact that the derivatives of $P$ and the highest weight $\lambda$ are invariant under the action of $\IZ_3$. Expanding the first requirement, we have that
\begin{equation}
	\sum_{n\ge\lambda_1} \left(x_{1,n-\lambda_1}+ x_{3,n-\lambda_1}+x_{4,n-\lambda_i}\right) \left( -\frac{\partial P}{\partial u_{2,-n-1}} + 2 \frac{\partial P}{\partial u_{i,-n-1}} \right) =0~.
\end{equation}
Once again, this can only hold if each screening charge individually acts as zero.
\end{proof}

\subsection{Infinitesmal coordinate changes}

Picking a coordinate $t$, on the disc, $D$, we can write any oper $\nabla\in\mathrm{Op}_{^L G}(D)$
\begin{equation}
	\nabla = \partial_t + p_{-1} \mathbf{v}(t)~,
\end{equation}
where $\mathbf{v}(t)\in\ ^L\mathfrak{b}[[t]]$. Of course, our choice of $t$ is arbitrary and so one is free to set $t=\phi(s)$ with respect to a different coordinate. Under such a change, $\nabla$ transforms as $\nabla \mapsto \widetilde{\nabla} = \partial_s + p_{-1}+ \widetilde{\mathbf{v}}(s)$, with~\cite{Frenkel:2007}
\begin{equation}
	\widetilde{\mathbf{v}}(s) = \phi'(s){\rho}(\phi'(s))\cdot \mathbf{v}(\phi(s))\cdot {\rho}(\phi'(s))^{-1} -\rho\left(\frac{\phi''(s)}{\phi'(s)}\right)~,
\end{equation}
where $\rho$ is the Weyl vector of $\gf$, which we think of as a co-character $\rho:\IC^\times \longrightarrow\! ^LG$. This furnishes an action of $\mathrm{Aut}(\OO_D)$ on the opers and considering the infinitesimal coordinate transformations gives the action of the Lie algebra $\mathrm{Der}(\OO_{D})$. Picking a coordinate $t$, the action of $\mathrm{Der}(\IC[[t]])$ on $\mathrm{Fun}\, \mathrm{Op}_{^L G}(D)$ is generated by $L_m = t^{m+1}\partial_t$ for $m\ge-1$.

For our proof of conformality of $\V_{1,1}$ we wish to show that the action of the $\mathrm{Der}(\OO_D)$ intertwines with the closed embedding $\mathrm{Op}_{^LG_t}(D)\hookrightarrow\mathrm{Op}_{^L G_u}(D)$. The action on opers is quite complicated and so once more we shall make use of the Miura transform.

The space of Cartan connections, $\mathrm{Conn}(\Omega^\rho)_D$, carries an action of $\mathrm{Aut}(\OO_D)$, given by~\eqref{eq:cartan_coord_transform}. The Miura transform, $\mu:\mathrm{Conn}(\Omega^\rho)_D\rightarrow \mathrm{Op}_{^L G}(D)$, intertwines the action of $\mathrm{Aut}(\OO_D)$ and of the Lie algebra $\mathrm{Der}(\OO_D)$. Therefore, it is sufficient to show that the closed embedding, $\mathrm{Conn}(\Omega^\rho)_{^LG_t,D}\hookrightarrow \mathrm{Conn}(\Omega^\rho)_{^LG_u,D}$ intertwines the $\mathrm{Der}(\OO_D)$ action.

On the function ring, $\mathrm{Fun}\,\mathrm{Conn}(\Omega^\rho)_{^LG_u,D}\cong \IC[u_{i,n}|i=1,\dots,\mathrm{rk}\, \gf;\,n<0]$, the $L_m$ act~\cite{Frenkel:2007} as
\begin{equation}
	\begin{split}
	L_m u_{i,n} = - n u_{i,n+m}\,, \quad -1\le m< -n ~,\\
	L_m u_{i,-m} = -m(m+1)\,, \quad m>0~,\\
	L_m u_{i,n} =0\,, \quad m>-n~.
	\end{split}
\end{equation}
Clearly, this action intertwines the projection to the coinvariants and so we have a commutative diagram
\begin{equation}\label{eq:der_od_comm}
	\begin{tikzcd}
	\mathrm{Fun}\,\mathrm{Op}_{^LG_u}(D) \arrow[r,twoheadrightarrow] \arrow[d, "\mathrm{Der}(\OO_D)",swap] & \mathrm{Fun}\, \mathrm{\Op}_{^LG_t}(D)\arrow[d, "\mathrm{Der}(\OO_D)"]\\
	\mathrm{Fun}\,\mathrm{Op}_{^LG_u}(D)\arrow[r,twoheadrightarrow] & \mathrm{Fun}\, \mathrm{\Op}_{^LG_t}(D)
	\end{tikzcd}
	~.
\end{equation}


\section{\label{app:proof_cyl}The proof of Theorem~\ref{thm:cyl_kl0}}

First, we establish fix notation. Let $\FF$ denote the composition $H^0_{DS}(u,\Hi{0}(\ZZ,\W_u\otimes-))$, which is an endofunctor on $\ZZ_u-$Mod. Notably, on the subcategory $\KL_{u,0}$, $\FF(M)\cong M$ for any $M\in \KL_{u,0}$ (see Proposition 9.12 of~\cite{Arakawa:2018egx}), and in general $\FF$ is left-exact. We wish to show that $\FF(\D_t)\cong \D_t$.

Let $\chi_\lambda:\ZZ_t\rightarrow \IC$ be a character defined by $\chi_\lambda(P_{i,n})=0$ for $n\neq0$ and $P_{i,0}\IV^t_\lambda = \chi_\lambda(P_{i,0})\IV^t_\lambda$. The Kazhdan--Lusztig category decomposes into blocks $\KL_t\cong \bigoplus \KL_t^{[\lambda]}$ where $\KL_t^{[\lambda]}$ is the subcategory of objects on which $P_{i,0}$ acts as the \textit{generalised} eigenvalue $\chi_\lambda(0)$,~\ie, these are objects which are \textit{set-theoretically} supported on the vanishing locus of $(P_{i,0}-\chi_\lambda(P_{i,0}))$ in $\mathrm{Spec}\, \ZZ_t$.

As $\D_t$ is a vertex algebra object in $\KL_t$, it decomposes as
\begin{equation*}
	\D_t \cong \bigoplus_{\lambda\in P^+_t}\D_{t,[\lambda]}~,
\end{equation*}
where $\D_{t,[\lambda]}$ are objects in $\KL_t^{[\lambda]}$. The increasing filtration on $\D_t$ induces a filtration on each $\D_{t,[\lambda]}$,
\begin{equation*}
	0= N_0 \subset N_1 \subset \dots \subset N = \bigcup_{i} N_i = \D_{t,[\lambda]}~,
\end{equation*}
such that successive quotients are isomorphic to $\IV_{\lambda,2}^t$. While the character of $\D_t$ is ill-defined (since each weight space is infinite dimensional), the character of each $\D_{t,[\lambda]}$ is well-defined. We have that $\FF(\D_t)_{[\lambda]} \cong \FF(\D_{t,[\lambda]})$ since $\FF$ is left exact. Therefore,
\begin{equation*}
	\ch~ \FF(\D_t)_{[\lambda]} \cong \ch~ \FF(\D_{t,[\lambda]})~.
\end{equation*}
The filtration on $\D_{t,[\lambda]}$ induces a filtration on $\FF(\D_{t,[\lambda]})$ and since $\FF$ is left exact, we have that $\FF(N_{i})/\FF(N_{i-1}) \subseteq \FF(N_{i}/N_{i-1})$, and so
\begin{equation*}
	\ch~\FF(\D_{t,[\lambda]}) \leqslant \sum_i \ch~\FF(N_{i}/N_{i-1})~.
\end{equation*}
Each subquotient, $N_i/N_{i-1}$ is an object in $\KL_{u,0}$ by Proposition~\ref{prop:kl_u}. Consequently,
\begin{equation*}
	\ch~\FF(\D_{t,[\lambda]}) \leqslant \sum_i \ch~(N_{i}/N_{i-1}) = \ch \D_{t,[\lambda]}~.
\end{equation*}
From this we conclude that $\FF(\D_t)$ must be $\IZ_{\geqslant0}$ graded. Now, since $\D_t$ is simple, if we can show that there is a non-trivial vertex algebra morphism $\D_t\rightarrow \FF(\D_t)$, then we will have $\FF(\D_t)\cong \D_t$. From here we are in a very similar situation as the proof of Theorem 9.9 of~\cite{Arakawa:2018egx}, so we adapt that proof to our setting. Before we do so, however, we will need some subsidiary lemmas. Let $\gf^+_t\subset \hat{\gf}_{t,-\kappa_c}= t\gf_t[t]$. Then we have:

\begin{lem}
$\V_{1,1}^{\gf_t^+}=\Hi{0}(\ZZ_u,\W_u\otimes \D_t)^{\gf_t^+} \cong \Hi{0}(\ZZ_u,\W_u\otimes (\D_t)^{\gf_t^+})$~.
\end{lem}
\begin{proof}
Let $C^{\bullet,\bullet}$ be defined by
\begin{equation*}
	C^{p,q} = \W_u\otimes \D_t \otimes \Li{p}(\zf(\gf_u))\otimes \bigwedge\! ^q ((\gf_t^+)^*)~,
\end{equation*}
where $\W_u\otimes \D_t \otimes \Li{\bullet}(\zf(\gf_u)$ is the Feigin standard complex for computing $\Hi{0}(\ZZ_u,\W_u\otimes \D_t)$ and $\D_t \otimes \bigwedge\! ^q((\gf_t^+)^*$ is the Chevalley--Eilenberg complex for computing the \textit{ordinary} Lie algebra cohomology of $\gf_t^+$. We denote the differentials of each complex by $d_{\zf}$ and $d_{\gf}$, respectively, and extend them to $C^{\bullet\bullet}$ by letting $d_{\zf}$ act trivially on $ \bigwedge\! ^\bullet((\gf_t^+)^*$ and letting $d_{\gf}$ act trivially on $\W_u\otimes \Li{\bullet}(\zf(\gf_u))$. The two differentials anticommute and so $C^{\bullet\bullet}$ is a bicomplex from which we form the total complex $C^i_{tot} = \bigoplus_{p+q=i} C^{p,q}$ with differential $d = d_{\zf}+ (-1)^q d_{\gf}$.

There are two spectral sequences converging to the total cohomology $H_{tot}^\bullet(C)$, with second pages given by
\begin{equation}\nonumber
	\begin{split}
	 _I E_{2}^{p,q} &= \Hi{p}(\ZZ_u,\W_u\otimes H^q(\gf_t^+,\D_t))~,\\
	 _{II} E_{2}^{p,q} &= H^p(\gf_t^+, \Hi{q}(\ZZ_u,\W_u\otimes \D_t))~.
	\end{split}
\end{equation}
Note that $\D_t$ is injective over $U(t\gf_t[t])$, so $H^q(\gf_t^+,\D_t)$ is concentrated in degree zero. Furthermore, both $H^i(\gf_t^+,-)$ and $\Hi{i}(\ZZ_u,\W_u\otimes-)$ vanish in negative degrees because $\W_u$ is free over $\ZZ_{u,(<0)}$ and $(-)^{\gf_t^+}$ is left exact. Therefore, $_I E_{r}^{p,q}$ collapses at the second page and $_{II} E_{2}^{0,0}$ is stable. This gives the isomorphism
\begin{equation*}
	_I E_2^{0,0} \cong H^0_{tot}(C) \cong\ _{II}E_2^{0,0}~,
\end{equation*}
as desired.
\end{proof}
\vspace{-6pt}

\begin{lem}
$\FF(\D_t)^{\gf_t^+}= H^0_{DS}(u,\V_{1,1})^{\gf_t^+} \cong H^0_{DS}(u,\V_{1,1}^{\gf_t^+})$~.
\end{lem}
\begin{proof}
By Theorem $6.8$ of~\cite{Arakawa:2018egx}, the functors $H^0_{DS}(u,-)$ and $\Hi{0}(\hat{\gf}_{u,-\kappa_g}, \gf_u,\W_u\otimes -)$ are isomorphic. Let $C^{\bullet\bullet}$ be defined by
\begin{equation*}
	C^{p,q} = \W_u \otimes \V_{1,1} \otimes \Li{p}(\gf_u) \otimes \bigwedge\! ^q((\gf_t^+)^*)~.
\end{equation*}
Let $d_u$ be the differential on $\W_u\otimes \V_{1,1}\otimes \Li{\bullet}(\gf_u)$, which computes the relative $\hat{\gf}_{u,-\kappa_g}$ semi-infinite cohomology, and we extend $d_u$ to $C^{\bullet\bullet}$ by letting it act trivially on $\bigwedge\! ^q((\gf_t^+)^*)$. Similarly, let $d_t$ be the differential on $\V_{1,1}\otimes \bigwedge\! ^q((\gf_t^+)^*)$, which computes the ordinary Lie algebra cohomology of $\gf_t^+$, and we can extend this to $C^{\bullet\bullet}$ by letting it act trivially on $\W_u\otimes\Li{p}(\gf_u)$. The two differentials anticommute and so $C^{\bullet\bullet}$ is a bicomplex. We can form the total complex $C^{i}= \bigoplus_{p+q=i}C^{p,q}$ with total differential $d_{tot} = d_u + (-1)^q d_t$.

There are two spectral sequences converging to the total cohomology $H_{tot}(C)$, with second pages given by
\begin{equation}\nonumber
	\begin{split}
	_I E_2^{p,q} &= \Hi{p}(\hat{\gf}_{u,-\kappa_g},\gf_u, \W_u\otimes H^q(\gf_t^+,\V_{1,1}))~,\\
	_{II} E_2^{p,q} &= H^p(\gf_t^+, \Hi{q}(\hat{\gf}_{u,-\kappa_g},\gf_u, \W_u\otimes \V_{1,1}))~.
	\end{split}
\end{equation}
Since $\W_u$ is semijective in $\KL_u$, $\Hi{\bullet}(\hat{\gf}_{u,-\kappa_g},\gf_u, \W_u\otimes - )$ is concentrated in degree zero. Therefore, both spectral sequences collapse on the second page, and we have
\begin{equation*}
	_I E_2^{0,0} \cong H^0_{tot}(C) \cong _{II} E_2^{0,0}~,
\end{equation*}
as desired.
\end{proof}
\vspace{-6pt}

Combining both of the above lemmas, we have
\begin{equation*}
	\FF(\D_t)^{\gf_t^+} \cong \FF((\D_t)^{\gf_t^+})\cong \bigoplus_{\lambda\in P_t^+}\FF(\IV^t_\lambda\otimes V^t_{\lambda^*})~,
\end{equation*}
where we have used the fact that
\begin{equation*}
	(\D_t)^{\gf_t^+} \cong U(\hat{\gf}_{t,\kappa_c}) \otimes_{U(\gf_t[t])\oplus \IC K} \OO(G_t)\cong \bigoplus_{\lambda\in P_t^+}\IV^t_\lambda\otimes V^t_{\lambda^*}~.
\end{equation*}
The $\IV^t_\lambda\otimes V^t_{\lambda^*}$ are naturally modules over $\zf_\lambda^t$ and are therefore objects in $\KL_{u,0}$ by Proposition~\ref{prop:kl_u}. As a result, we have that
\begin{equation*}
	\FF(\D_t)^{\gf_t^+} \cong \bigoplus_{\lambda\in P_t^+} \IV_{\lambda}\otimes V_{\lambda^*}~.
\end{equation*}
Looking at the $\Delta=0$ subspace, $\FF(\D_t)^{\gf_t^+}_0$, and comparing with our constraint on the character, we have that
\begin{equation*}
	\FF(\D_t)_0 \cong \bigoplus_{\lambda\in P_t} V_{\lambda}\otimes V_{\lambda^*} \cong \OO(G_t)~,
\end{equation*}
as $\gf_t\otimes\gf_t$ modules. Since, $\FF(\D_t)$ is $\IZ_{\geqslant0}$-graded, the zero weight subspace $\FF(\D_t)_0$ is a unital commutative, associative algebra under the normal product. The quadratic Casimir provides a $\mathbb{Q}_{\geqslant0}$-grading,
\begin{equation*}
	\FF(\D_t)_0 = \bigoplus_{d\in \mathbb{Q}_{\geqslant0}} \FF(\D_t)_0(d)~,
\end{equation*}
where $\FF(\D_t)_0(d)$ has eigenvalue $d$ with respect to the quadratic Casimir. The natural projection $\FF(\D_t)_0\twoheadrightarrow \FF(\D_t)_0(0)$ is an algebra homomorphism.

These observations mean that we are exactly in the situation of the proof of Theorem 9.9 in~\cite{Arakawa:2018egx} and so we can apply Lemma 9.10 of \textit{loc.\ cit.} to conclude that $\FF(\D_t)_0$ is isomorphic to $\OO(G_t)$ as a commutative $G_t\times G_t$ algebra. Additionally, $\FF(\D_t)$ is a $\KL_t$ object and so we have an action of $V^{\kappa_c}(\gf_t)$. All together, this gives a nonzero homomorphism $\D_t \rightarrow \FF(\D_t)$, as desired.


\section{\label{app:rewriting_characters}Rewriting Arakawa's character}

The character of the genus-zero chiral algebra of class $\SS$ of type $\mathfrak{g}$ with $b$ punctures is given in \cite{Arakawa:2018egx} as follows,
\begin{equation}\label{eq:archarapp}
	\Tr_{\mathbf{V}_{G,b}^{\mathcal{S}}}(q^{L_0}z_1\dots z_b) =  \sum_{\lambda\in P_+} \left( \frac{q^{\langle\lambda,\rho^\vee\rangle} (q,q)_\infty^{\mathrm{rk}\mathfrak{g}}}{\prod_{\alpha\in\Delta_+} \left( 1-q^{\langle\lambda + \rho,\alpha^\vee\rangle} \right)}  \right)^{(b-2)} \prod_{k=1}^b \mathrm{tr}_{\mathbb{V}_\lambda}(q^{-D}z^k)~,
\end{equation}
where $L_0$ is the Virasoro zero mode, $z_i$ are elements of the maximal torus of $G$, and $D$ is the degree operator of the Weyl-module $\mathbb{V}_\lambda$. The inner product $ \langle \cdot,\cdot\rangle$ is the Killing form.

It will be useful to be able to rewrite this expression so that it resembles \eqref{eq:indfull}. The characters of the Weyl modules are given by the Weyl--Kac formula,
\begin{equation}
	\mathrm{Tr}_{\mathbb{V}_\lambda}(q^{-D}z) = \frac{\sum_{w\in W} \mathrm{sgn}(w) e^{w \circ \lambda}}{(q,q)_\infty^{\mathrm{rk}\mathfrak{g}} \prod_{\alpha\in\Delta_+}(e^{-\alpha},q)_\infty(qe^\alpha,q)_\infty}~,
\end{equation}
where for simplicity we suppress flavour fugacities. The sum in the numerator is over the Weyl group $W$ of $\gf$. The $q$-Pochhammer symbol is defined according to
\begin{equation}
	(a,q)_\infty \coloneqq \prod_{j=0}^{\infty} (1-aq^j)~.
\end{equation}

Our first observation is as follows:
\begin{equation}\label{eq:rearr}
	\begin{split}
	\Tr_{\mathbb{V}_\lambda}(q^{-D}z) &= \frac{\sum_{w\in W} \mathrm{sgn}(w) e^{w \circ \lambda}}{(q,q)_\infty^{\mathrm{rk}\mathfrak{g}} \prod_{\alpha\in\Delta_+}(qe^{-\alpha},q)_\infty(qe^\alpha,q)_\infty (1-e^{-\alpha})} \\
	&= \frac{\sum_{w\in W} \mathrm{sgn}(w) e^{w \circ \lambda}}{(q,q)_\infty^{\mathrm{rk}\mathfrak{g}} \prod_{\alpha\in\Delta_+}(qe^{-\alpha},q)_\infty(qe^\alpha,q)_\infty e^{-\alpha/2}(e^{\alpha/2}-e^{-\alpha/2})} \\
	&= \frac{1}{(q,q)_\infty^{\mathrm{rk}\mathfrak{g}} \prod_{\alpha \in \Delta}(qe^\alpha,q)_\infty} \; \frac{\sum_{w\in W} \mathrm{sgn}(w) e^{w(\lambda+\rho)-\rho}}{\prod_{\alpha\in\Delta_+}e^{-\alpha/2}(e^{\alpha/2}-e^{-\alpha/2})} \\
	&= \frac{1}{(q,q)_\infty^{\mathrm{rk}\mathfrak{g}} \prod_{\alpha \in \Delta}(qe^\alpha,q)_\infty} \; \frac{\sum_{w\in W} \mathrm{sgn}(w) e^{w(\lambda+\rho)}}{\prod_{\alpha\in\Delta_+}(e^{\alpha/2}-e^{-\alpha/2})}~. \\
	\end{split}
\end{equation}
Using the Weyl character formula~\cite{Fulton:1991} and the expression for $\mathcal{K}$-factors given in~\cite{Lemos:2012ph}, we achieve a final expression for the character of the Weyl module,
\begin{equation}\label{eq:weylch}
	\Tr_{\mathbb{V}_{\lambda}} (q^{-D}z) = \mathcal{K}(z)\chi^\lambda_\mathfrak{g}(z)~.
\end{equation}
What remains is to rewrite the prefactor from \eqref{eq:archarapp}. We can write the character as
\begin{equation}
	\Tr_{\V^\SS_{G,b}}(q^{L_0}z_1\dots z_b) = \sum _{\lambda\in P_+}(C_{\lambda\lambda\lambda})^{b-2} \prod_{i=1}^b \KK(z_i)\chi_\gf^\lambda(z_i)~,
\end{equation}
which is precisely the form in \eqref{eq:indfull}. The Higgsing prescription for the character is identical to that of the index. Therefore, we may apply the same argument we used to calculate $C_{\lambda\lambda\lambda}$ for the index, which results in the final expression of
\begin{equation}\label{eq:append_indfinal}
	\Tr_{\mathbf{V}^\mathcal{S}_{G,b}}(q^{L_0}z_1\cdots z_b) = \sum_{\lambda\in P_+} \frac{\prod_{i=1}^{b}\mathcal{K}(z_i)\chi^\lambda_\mathfrak{g}(z_i)}{(\mathcal{K}(\times)\chi^\lambda_\mathfrak{g}(\times))^{b-2}}~.
\end{equation}

We would like to draw attention to a point of interest. One can attempt to massage the prefactor
\begin{equation}\label{eq:rearr2}
	\frac{q^{\langle\lambda,\rho^\vee\rangle} (q,q)_\infty^{\mathrm{rk}\mathfrak{g}}}{\prod_{\alpha\in\Delta_+} \left( 1-q^{\langle\lambda + \rho,\alpha^\vee\rangle} \right)} = \frac{(q,q)^{\mathrm{rk}\mathfrak{g}}_\infty}{\prod_{\alpha\in\Delta_+}q^{\langle\rho,\alpha^\vee/2\rangle}(q^{-\langle\lambda+\rho,\alpha^\vee/2\rangle}-q^{\langle\lambda+\rho,\alpha^\vee/2\rangle})}~,
\end{equation}
into a more suitable form. We largely follow the method of Appendix A of~\cite{Lemos:2014lua}, generalised to arbitrary simple $\gf$.
\begin{equation}
	\begin{split}
	\frac{q^{\langle\lambda,\rho^\vee\rangle} (q,q)_\infty^{\mathrm{rk}\mathfrak{g}}}{\prod_{\alpha\in\Delta_+} \left( 1-q^{\langle\lambda + \rho,\alpha^\vee\rangle} \right)} &= \frac{(q,q)_\infty^{\mathrm{rk}\mathfrak{g}}(q^{-\langle \rho,\alpha^\vee/2\rangle}-q^{\langle \rho,\alpha^\vee/2\rangle})}{\prod_{\alpha\in\Delta_+}q^{\langle\rho,\alpha^\vee/2\rangle}(q^{-\langle \rho,\alpha^\vee/2\rangle}-q^{\langle \rho,\alpha^\vee/2\rangle})(q^{-\langle\lambda+\rho,\alpha^\vee/2\rangle}-q^{\langle\lambda+\rho,\alpha^\vee/2\rangle})} \\
	&= \frac{(q,q)_\infty^{\mathrm{rk}\mathfrak{g}}}{\prod_{\alpha\in\Delta_+}(1-q^{\langle \rho,\alpha^\vee\rangle})}\cdot\prod_{\alpha\in\Delta^+} \frac{ (q^{-\langle \rho,\alpha^\vee/2\rangle}-q^{\langle \rho,\alpha^\vee/2\rangle})}{(q^{-\langle\lambda+\rho,\alpha^\vee/2\rangle}-q^{\langle\lambda+\rho,\alpha^\vee/2\rangle})}~.
	\end{split}
\end{equation}
Using Equation $(76)$ of~\cite{Arakawa:2018egx} we rewrite the product over positive roots as a product over degrees $d_i$ of fundamental invariants,
\begin{equation}
	\begin{split}
	\prod_{\alpha\in\Delta_+} (1-q^{\langle \rho,\alpha^\vee \rangle}) &= \prod_{i=1}^{\mathrm{rk}\mathfrak{g}}\prod_{j=1}^{d_i-1} \\
	&= \prod_{i=1}^{\mathrm{rk}\mathfrak{g}}\frac{(q,q)_\infty}{(q^{d_i},q)_\infty} \, ,\\
	&= (q,q)^{\mathrm{rk}\,\gf}_\infty \KK(\times)~.
	\end{split}
\end{equation}
We have substituted the expression for the $\mathcal{K}$-factor of an empty puncture from~\cite{Lemos:2012ph} in the last line. Thus, the prefactor takes the simplified form
\begin{equation}
	\frac{q^{\langle\lambda,\rho^\vee\rangle} (q,q)_\infty^{\mathrm{rk}\mathfrak{g}}}{\prod_{\alpha\in\Delta_+} \left( 1-q^{\langle\lambda + \rho,\alpha^\vee\rangle} \right)} = \frac{1}{\mathcal{K}(\times)}\prod_{\alpha\in\Delta^+} \frac{ (q^{-\langle \rho,\alpha^\vee/2\rangle}-q^{\langle \rho,\alpha^\vee/2\rangle})}{(q^{-\langle\lambda+\rho,\alpha^\vee/2\rangle}-q^{\langle\lambda+\rho,\alpha^\vee/2\rangle})}~.
\end{equation}
Comparing the above to (\ref{eq:append_indfinal}), we arrive at the identity
\begin{equation}\label{eq:append_chi}
	\chi^\lambda_\gf(\times) =\prod_{\alpha\in\Delta^+} \frac{(q^{-\langle\lambda+\rho,\alpha^\vee/2\rangle}-q^{\langle\lambda+\rho,\alpha^\vee/2\rangle})}{ (q^{-\langle \rho,\alpha^\vee/2\rangle}-q^{\langle \rho,\alpha^\vee/2\rangle})}~.
\end{equation}
While this identity is interesting in its own right, it is strikingly similar to the $q$-deformed dimension. It is known that the Schur limit of the superconformal index of Class $\SS$ theories is related to $q$-deformed two dimensional Yang-Mills Theory. The $q$-deformed dimension arises naturally in this setting and for simply laced $\gf$, we have
\begin{equation}
	\dim_q\, R_\lambda = \prod_{\alpha\in\Delta^+} \frac{(q^{-\langle\lambda+\rho,\alpha^\vee/2\rangle}-q^{\langle\lambda+\rho,\alpha^\vee/2\rangle})}{ (q^{-\langle \rho,\alpha^\vee/2\rangle}-q^{\langle \rho,\alpha^\vee/2\rangle})} = \chi^\lambda_\gf(\times)~.
\end{equation}
The representation $R_\lambda$ is the tilting module~\cite{Andersen1992} of the $q$-deformed universal enveloping algebra $U_q(\gf)$~\cite{Jimbo:1985zk} associated to the highest weight $\lambda\in P^+$. We wish to bring attention to this since this correspondence does not hold for non-simply laced $\gf$. For non-simply laced $\gf$, the $q$-dimension is defined as
\begin{equation}
	\dim_q\, R_\lambda = \prod_{\alpha\in\Delta^+} \frac{(q^{-d_\alpha\langle\lambda+\rho,\alpha^\vee/2\rangle}-q^{d_\alpha\langle\lambda+\rho,\alpha^\vee/2\rangle})}{ (q^{-d_\alpha\langle \rho,\alpha^\vee/2\rangle}-q^{d_\alpha\langle \rho,\alpha^\vee/2\rangle})}~.
\end{equation}
The $d_\alpha\in\{1,2,3\}$ are defined for simple roots, $\alpha_i$ $i\in\{1,\dots,\mathrm{rk}\,\gf\}$ as the integers such that if $A_{ij}$ is the Cartan matrix of $\gf$, $d_iA_{ij}$ is symmetric. For positive roots one sets $d_\alpha = d_i$ if there is a Weyl transformation sending the positive root $\alpha$ to the simple root, $\alpha_i$.

It is worth noting that this expression for the $q$-deformed dimension no longer agrees with the expression for the fully reduced character, $\chi^\lambda_\gf(\times)$, in the non-simply laced case. We are, as yet, unsure about the significance of this.

\section{\label{app:ds_free}Drinfel'd--Sokolov reduction of the free hypermultiplet}

In this appendix we present some supporting evidence for the equivalence~\eqref{eq:ds_free}. We investigate the untwisted DS reduction of the relevant free hypermultiplet VOA in concrete examples for low ranks---making full use of accidental isomorphisms. We do not elucidate the full structure of the vertex algebra, but rather produce a list of generators, quantum numbers and null relations that are consistent with the expected identification.

\subsection{\label{subsec:D2_DS}The \texorpdfstring{$\df_2$}{d2} Case}

From the accidental isomorphism $\df_2\cong \af_1\times \af_1$, so the identification $\V_f^{\df_2}\cong \V_3^{\af_1}$ and~\eqref{eq:ds_free} follows directly from Arakawa's results in \cite{Arakawa2017introduction}. Nevertheless, direct calculation proves instructive.

The VOA $\V_3^{\af_1}$ is the trifundamental symplectic boson VOA, with generators $Q_{abc}$, where $a,b,c$ are $\suf(2)$ fundamental indices. There are three commuting $V^{\kappa_c}(\suf_2)$ current subalgebras, which are generated by the currents
\begin{equation}\label{eq:curr_embed}
	\begin{split}
	J_{ij} &= \frac{1}{2}\epsilon^{bb'}\epsilon^{cc'}Q_{ibc}Q_{jcc'}~,\\
	L_{ij} &= \frac{1}{2}\epsilon^{bb'}\epsilon^{cc'}Q_{bic}Q_{cjc'}~,\\
	R_{ij} &= \frac{1}{2}\epsilon^{bb'}\epsilon^{cc'}Q_{bci}Q_{cc'j}~.
	\end{split}
\end{equation}
The untwisted $\df_2$ DS reduction amounts to two simultaneous $\suf(2)$ reductions with respect to the $L$ and $R$ subalgebras. A spectral sequence argument, coupled with the exactness of $H^0_{DS}$, implies that this can be implemented by sequentially reducing with respect to $R$ and then $L$. Introducing fermionic ghosts $b_R$ and $c_R$, the BRST charge for the $R$ reduction is
\begin{equation}
	Q_R = c_R(R_{22}-1)~.
\end{equation}
Following the arguments of~\cite{Beem:2014rza}, the (possibly redundant) generators of the reduced VOA can be identified with the operators
\begin{equation}
	R_{11},Q_{abc}, J_{ij},L_{ij}~.
\end{equation}
A genuine set of generators of the reduced algebra are obtained via the \textit{tic-tac-toe} procedure detailed in~\cite{deBoer:1993iz}. We distinguish these generators by writing them in calligraphic font,
\begin{equation}
	\RR_{11},\QQ_{abc}, \JJ_{ij},\LL_{ij}~.
\end{equation}
The reduced conformal weights are related to the unreduced weights by a shift by the eigenvalue with respect to the $\slf(2)$ Cartan $R_{12}+R_{21}$; these are summarised in Table~\ref{tab:D2-DS}.

\begin{table}[b!]
\centering
\renewcommand{\arraystretch}{1.25}
\begin{tabular}{|>{\centering}m{2.5cm} | >{\centering}m{1.5cm} | >{\centering}m{1.5cm} | >{\centering}m{1.5cm} | >{\centering}m{1.5cm} | >{\centering\arraybackslash}m{1.5cm} |}
\hline
Generator & $\QQ_{ab2}$  &  $\QQ_{ab1}$ & $\JJ_{ij}$ & $\LL_{ij}$ & $\RR_{11}$ \\
\hline
Weight & 0 & 1 & 1 & 1 & 2\\
\hline
\end{tabular}
\caption{Generators of the principal DS reduction of $\SS\BB$ with respect to the $R$ current subalgebra.}
\label{tab:D2-DS}
\end{table}

This reduction is expected to result in a VOA that is isomorphic to $\D_{\af_1}$. The generators $\QQ_{ab2}$, $\JJ_{ij}$, and $\LL_{ij}$ have the correct quantum numbers to serve as the generators of $\D_{\af_1}$. The generators $\QQ_{ab1}$ and $\RR_{11}$ should therefore not be independent, but should be related to composites of the other generators through null relations.

The symplectic boson VOA has no relations \emph{per se}, but the method of~\cite{Beem:2014rza} treats the identity \eqref{eq:curr_embed} as a relation between the current subalgebras and the $Q_{abc}$ (\ie\, one thinks of the affine currents as independent generators subject to some relations). It is precisely the incarnation of these relations in the reduced algebra that eliminates the redundant generators.

Principal DS reduction of the critical level affine subalgebra produces the Feigin--Frenkel centre and so we should identify $\RR_{11}$ with the quadratic generator $P_1$ of the shared Feigin--Frenkel centre. Therefore, we expect
\begin{equation}
	\RR_{11} = \alpha\left(:\frac{1}{8}(J_{12}+J_{21})(J_{12}+J_{21}): + \frac{1}{2}:J_{11}J_{22}: + \frac{1}{2}:J_{22}J_{11}: \right)+\ldots~,
\end{equation}
where the $(\dots)$ represent $Q_R$-exact terms and $\alpha\in\IC^\times$ is some normalization. From the embedding of the current algebras, we have that
\begin{equation}
	P_1^J=P_1^R=P_1^L~,
\end{equation}
where $P_1^{J,L,R}$ are the quadratic Casimirs of the respective current algebras.

Through the tic-tac-toe procedure, we have that
\begin{equation}
	\RR_{11} = R_{11} + \frac{1}{4} :RR:  - \frac{1}{2}\partial R~,
\end{equation}
where
\begin{equation}
	R = \frac{1}{2}\left(R_{12}+R_{21}\right) +2:b_Rc_R:~.
\end{equation}
Note that the difference $	\RR_{11}- P_1^R $ is $Q_R$-exact, namely
\begin{equation}
	(Q_{R})_{(0)} :R_{11}b: = \RR_{11}-P_1^R~.
\end{equation}
Thus $\RR_{11}$ can be expressed as the quadratic Casimir of the surviving currents $J_{ij}$ or $L_{ij}$.

Redundancy of $Q_{ab1}$ is a result of the relations
\begin{equation}
	\begin{split}
	\epsilon^{ij}R_{i2} Q_{12j} &=  \epsilon^{ij}L_{i2}Q_{1j2}~,\\
	\epsilon^{ij}R_{i2} Q_{11j} &=  2\epsilon^{ij}L_{i1}Q_{1j2}+3/2 \partial Q_{112}~,\\
	\epsilon^{ij}R_{2i} Q_{21j} &=  1/2\epsilon^{ij}L_{1i}Q_{2j2}-3/4 \partial Q_{212}~,\\
	\epsilon^{ij}R_{2i} Q_{22j} &=  \epsilon^{ij}L_{2i}Q_{2j2}~.
	\end{split}
\end{equation}
In cohomology, $J_{22}$ acts as the identity and so the vertex algebraic incarnation of the above relates $\QQ_{ab1}$ to a linear combination of $\LL_{ij}Q_{ab2}$ and the derivatives $\partial \QQ_{ab2}$.

\subsection{\label{subsec:d3_DS}The \texorpdfstring{$\df_3$}{d3} Case}

The strong generators of the $\df_3$ symplectic bosons can be written $Q_{Aa}$, where $A=1,\dots,6$ is a $\df_4\cong\suf(4)$ index---valued in the six-dimensional real representation---and $a=1,\dots,4$ is a $\ccf_2\cong\mathfrak{usp}(4)$ fundamental index. The OPE is
\begin{equation}
	Q_{Aa}(z)Q_{Bb}(w) \sim \frac{ \delta_{AB}\Omega_{ab}}{(z-w)}~,
\end{equation}
where $\Omega$ is the invariant $\mathfrak{usp}(4)$ symplectic matrix. The $\hat{\gf}_{u,\kappa_c}\times\hat{\gf}_{t,\kappa_c}$ current subalgebras are generated by the affine currents
\begin{equation}\label{eq:d3_embed}
\begin{split}
	J_{AB} &= \Omega^{ab}(Q_{Aa}Q_{Bb})~,\\
	K_{ab} &= \delta^{AB}(Q_{Aa}Q_{Bb})~.
\end{split}
\end{equation}
$\df_3$ reduction of the $J_{AB}$ current algebra produces $\JJ_1,\JJ_2$ and $\JJ_3$ with conformal weights $\Delta=2,3,4$. These are to be identified with the generators of the Feigin--Frenkel centre $\ZZ$ of $V^{\kappa_c}(\suf(4))$, and in particular $\JJ_2$ is the degree three generator of the Feigin--Frenkel centre. However, this state is set to zero in the symplectic boson system as it is not invariant under the outer-automorphism of $\df_3$. Thus, only the $\JJ_1$ and $\JJ_3$ survive. Similar to the $\df_2$ case we believe that null-relations arising from the embedding \eqref{eq:d3_embed} relate $\JJ_1$ and $\JJ_3$ to the Feigin--Frenkel generators of the $K_{ab}$ algebra.

Under the principal embedding $\suf(2)\hookrightarrow\suf(4)$ the six-dimensional representation decomposes as $\mathbf{6}_{\suf(4)}\to\mathbf{5}_{\suf(2)}\oplus \mathbf{1}_{\suf(2)}$. Accordingly, we will decompose $Q_{Aa}$ into $\suf(2)$ irreducible representations as follows,
\begin{equation}
	Q_{Aa} \rightarrow \left(Q_{1a} , Q_{2a}, Q_{3a}, Q_{4a}, Q_{5a} \right) \oplus Q_{6a}~.
\end{equation}
The conformal weights of the generators of the cohomology are shifted by their eigenvalues under the image of the Cartan generator $\Lambda(h)$. The weight spaces of the cohomology are given in Table~\ref{tab:D3-DS}. We have nothing to compare these against so we proceed by reducing the cylinder, $\D_t$, to obtain generators of $\W_t$.

\begin{table}[t!]
\centering
\renewcommand{\arraystretch}{1.25}
\begin{tabular}{|>{\centering}m{2cm} | >{\centering}m{1cm} | >{\centering}m{1cm} | >{\centering}m{1.5cm} | >{\centering}m{1cm} | >{\centering}m{1cm} | >{\centering}m{1cm} |>{\centering\arraybackslash}m{1cm} |}
\hline
Generator & $\QQ_{5a}$ & $\QQ_{4a}$ & $\QQ_{3a},\QQ_{6a}$ & $\QQ_{2a}$ & $\QQ_{1a}$ & $\JJ_1$ & $\JJ_3$ \\
\hline
Weight & $-3/2$ & $-1/2$ & $1/2$ & $3/2$ & $5/2$& $2$ & $4$\\
\hline
\end{tabular}
\caption{Conformal weights of the (possibly redundant) generators of $H^0_{DS}(\SS\BB)$.}
\label{tab:D3-DS}
\end{table}

\begin{table}[b!]
\renewcommand{\arraystretch}{1.25}
\centering
\begin{tabular}{|>{\centering}m{2.5cm} | >{\centering}m{1.5cm} | >{\centering}m{1.5cm} | >{\centering}m{1.5cm} | >{\centering}m{1.5cm} | >{\centering}m{1.5cm} | >{\centering\arraybackslash}m{1.5cm} |}
\hline
Generator & $\GG_{a1}$ & $\GG_{a2}$ & $\GG_{a3}$ & $\GG_{a4}$ & $\KK_{\alpha_1+\alpha_2}$ & $\KK_{2\alpha_1}$ \\
\hline
Weights & $-3/2$ & $-1/2$ & $\phantom{-}1/2$ & $\phantom{-}3/2$& $\phantom{-}2\phantom{/2}$ & $\phantom{-}4\phantom{/2}$\\
\hline
\end{tabular}
\caption{Conformal weights of the (possibly redundant) generators of $\W_{C_2}$.}
\label{tab:C2_DS}
\end{table}

The twisted cylinder, $\D_t$, is strongly generated by $\widehat{\ccf}_{2,\kappa_c}$ currents $K_\alpha$, and dimension zero operators $g_{ab}$, which obey the relations of $\OO(C_2)$. We have labelled the affine currents by the roots of $\ccf_2$. We denote the simple roots of $\ccf_2$ by $\alpha_1$ and $\alpha_2$ with positive roots given by the linear combinations
\begin{equation}
\alpha_1-\alpha_2, \alpha_1+\alpha_2,2\alpha_1, 2\alpha_2~.
\end{equation}
The principal embedding of the positive root $e\in\suf(2)$ lands in the subspace
\begin{equation}
	\Lambda(e) \in \mathrm{Span} (K_{\alpha_1-\alpha_2} , K_{2\alpha_2})~.
\end{equation}
The cohomology is generated by the currents that commute with the image of $\Lambda(e)$, which are
\begin{equation}
	K_{\alpha_1-\alpha_2}\text{ and } K_{2\alpha_1}~.
\end{equation}
The conformal weights of the associated vertex algebra generators is given in Table~\ref{tab:C2_DS} and we see that they agree with the expected weights of the generators of the Feigin--Frenkel centre. As expected, the images of $K_{\alpha_1-\alpha_2}$ and $K_{2\alpha_1}$ should be identified with the generators, $P_1$ and $P_3$, of the Feigin--Frenkel centre. On the hypermultiplet side, this centre is generated by $\JJ_1$ and $\JJ_3$ and this lets us identify the two graded subspaces.

Under the principal embedding, the fundamental representation $\mathbf{4}_{\mathfrak{usp}(4)}$ of $\mathfrak{usp}(4)$ decomposes as $\mathbf{ 4}_{\mathfrak{usp}(4)}\rightarrow \mathbf{4}_{\suf(2)}$. Thus, the weights of the dimension zero generators are shifted to those in Table~\ref{tab:C2_DS}. Comparing Tables~\ref{tab:D3-DS} and~\ref{tab:C2_DS}, we can identify the weights of the generators of $\W_{C_2}$ with those of $H^0_{DS}(u,\SS\BB)$, with some redundancies. So far, we have only shown that
\begin{equation}
	\W_{C_2}\subset H^0_{DS}(\SS\BB)~,
\end{equation}
as graded vector spaces. For now, we wish to highlight that $H^0_{DS}(u,\SS\BB)$ contain redundant generators, like in the $\df_2$ case. These redundant generators are set to zero by the image of the null relations, \eqref{eq:d3_embed}, in the reduced vertex algebra. To improve this result to an isomorphism of vertex algebras, we require both the OPE structure and the null relations. We leave the exact details for future work.

\bibliographystyle{./aux/ytphys}
\bibliography{./aux/refs}

\end{document}